\documentclass[12pt,a4paper,reqno]{revtex4-1}
\usepackage{graphicx}
\usepackage{dcolumn}
\usepackage{bm}
\usepackage{amsmath, amsthm, amscd, amsfonts, amssymb, graphicx, color}
\usepackage{url}
\usepackage{amsthm}
\usepackage{layout}
\usepackage{epsfig}
\usepackage{graphicx}
\usepackage{epstopdf}
\usepackage{booktabs}
\usepackage{float}
\usepackage{subfigure}
\usepackage[hyperindex,breaklinks]{hyperref}
\usepackage{xcolor}
\usepackage{enumitem}
\usepackage{setspace}
\usepackage{times}

\usepackage{extarrows}
\usepackage{bbm}
\usepackage{threeparttable}

\def\Tr{{\rm Tr}}
\newtheorem{theorem}{Theorem}[section]

\newtheorem{corollary}[theorem]{Corollary}

\newtheorem{proposition}[theorem]{Proposition}
\newtheorem{definition}[theorem]{Definition}

\numberwithin{equation}{section}

\allowdisplaybreaks[4]
\hypersetup{colorlinks=true,linkcolor=blue, anchorcolor=green, citecolor=cyan, urlcolor=red, filecolor=magenta, pdftoolbar=true}

\usepackage{cleveref}

\crefformat{section}{#2#1#3}
\Crefformat{section}{#2#1#3}
\crefformat{subsection}{#2#1#3}
\Crefformat{subsection}{#2#1#3}
\crefformat{subsubsection}{#2#1#3}
\Crefformat{subsubsection}{#2#1#3}

\crefname{section}{sec.}{secs.}
\Crefname{section}{Sec.}{Secs.}
\crefname{subsection}{subsec.}{subsecs.}
\Crefname{subsection}{Subsec.}{Subsecs.}
\crefname{subsubsection}{subsubsec.}{subsubsecs.}
\Crefname{subsubsection}{Subsubsec.}{Subsubsecs.}

\makeatletter
\renewcommand{\p@subsection}{}
\renewcommand{\p@subsubsection}{}
\makeatother

\usepackage{titlesec} 

\titleformat{\section}
    {\centering\large\bfseries} 
    {\thesection} 
    {1em} 
    {} 
    [\vspace{0.5ex}] 

\titleformat{\subsection}
    {\centering\normalsize\bfseries}  
    {\thesubsection}
    {1em}
    {}
    [\vspace{0.3ex}]

    \titleformat{\subsubsection}
    {\centering\normalsize\bfseries}  
    {\thesubsubsection}
    {1em}
    {}
    [\vspace{0.1ex}]

\titlespacing*{\section}{0pt}{3ex}{2ex} 
\titlespacing*{\subsection}{0pt}{2.5ex}{1.5ex}

\renewcommand{\thesection}{\arabic{section}}           
\renewcommand{\thesubsection}{\thesection.\arabic{subsection}} 
\renewcommand{\thesubsubsection}{\thesubsection.\arabic{subsubsection}} 

\titlespacing*{\section}{0pt}{3ex}{2ex} 
\titlespacing*{\subsection}{0pt}{2.5ex}{1.5ex}

\renewcommand{\thesection}{\arabic{section}}           
\renewcommand{\thesubsection}{\thesection.\arabic{subsection}} 
\renewcommand{\thesubsubsection}{\thesubsection.\arabic{subsubsection}} 

\titlespacing*{\section}{0pt}{12pt}{6pt} 

\usepackage{appendix}

\usepackage{xcolor}
\definecolor{blue}{RGB}{0, 0, 255}

\newcommand{\yaqi}[1]{\textcolor{black}{#1}}

\usepackage{chngcntr}
\counterwithout{equation}{section}

\begin{document}

\title[]{An Operator Limit Model of Entanglement Percolation in Series-Parallel Quantum Networks}

\author{Kan He}
\email{hekanquantum@163.com}
\affiliation{College of Mathematics, Taiyuan University of Technology, Taiyuan, 030024, China; School of Mathematics and Statistics, Shaanxi Normal University, Xi'an, 710119, China}

\author{Jinchuan Hou}
\email{jinchuanhou@aliyun.com}
\affiliation{College of Mathematics, Taiyuan University of Technology, Taiyuan, 030024, China}

\author{Yaqi Zhao}
\email{205892@xaut.edu.cn}
\affiliation{School of Mathematics, Xi’an University of Technology, Xi’an, 710054, China}

\thanks{{The authors are listed in alphabetical order by surname. All authors contributed equally to this work.} \yaqi{We thank Professor Jun Shu for help and discussions.} This work was supported by the National Natural Science Foundation of China under Grants No.~12271394 {and No.~12571138}.}


\begin{abstract}

 \textbf{Abstract:} The realization of entanglement distribution (ED) in large-scale quantum networks (QNs), as a fundamental theoretical framework of quantum communication, faces significant challenges. Percolation theory, drawn from statistical physics, has been identified as a potential solution to this problem. Nevertheless, when addressing ED, percolation theory primarily relies on numerical simulations and approximate algorithms. This approach often lacks interpretability and controllability for many phenomena occurring during the ED process in large-scale networks, including emergent phenomena near the threshold. To address these limitations, in this paper, we {focus on} leveraging operator theory to construct a mathematical framework for entanglement percolation in series-parallel QNs. Through the application of operator-theoretic methodologies, we analyze several critical issues within the entanglement percolation process. These issues include the sequence of operations and the existence of a percolation threshold. We represent the three categories of operations in entanglement percolation using super-operators. Consequently, the entire percolation process can be depicted as the limit process of a superoperator sequence, briefly as $${\rm lim}_{N\rightarrow +\infty}{\rm \Lambda}_N\circ \Lambda_{N-1}\circ \dots \Lambda_1(\rho(N)),$$ where $\{{\rm \Lambda}_i\}$ is a  sequence of  superoperators, $\rho(N)$ is a density operator corresponding to the QN. We introduce the concept of order independence in the sequence of various percolation operations and provide criteria for determining this property. Utilizing an operator limit model, we analyze the conditions for the existence of a percolation threshold, which are equivalent to the conditions for the success of percolation. Unlike existing numerical approximation methods for handling entanglement percolation, we provide a more analytical operator-theoretic solution framework. This framework enables us to gain a more precise understanding of the entanglement percolation process and the mechanism behind the special phenomena involved.

\vspace{10pt}

\noindent\textbf{Keywords:} Superoperators, quantum networks, entanglement percolation, quantum entanglement.
\end{abstract}

\pacs{03.67.Mn, 03.65.Ud, 03.67.-a}
\maketitle

\tableofcontents \newpage

\section{Introduction}

The objective of entanglement distribution is to facilitate the sharing of quantum entanglement between two long-distance parties. However, the implementation of the long-distance quantum communication necessitates the involvement of multiple repeaters because of the constraints imposed by {current devices and technologies}~\cite{pan4600, repeater1}.
Large-scale and long-distance quantum communications give rise to quantum networks (QNs) \cite{Quantum_network1997, Qinter,Network_2023,Network_2020}. Consequently, a new challenge emerges {naturally}: the investigation and practical realization of entanglement distribution within the context of large-scale quantum networks. {Percolation, a theory derived from {statistical physics}, offers a feasible approach to addressing this problem.}

Percolation theory~\cite{Percolation_Shklovskii1984} studies the movement and connectivity of substances through disordered media, such as liquids filtering through porous materials or information spreading in networks. It originated in statistical physics and mathematics to model phase transitions and critical phenomena, where systems exhibit abrupt changes in macroscopic behavior (e.g., from isolated clusters to a connected ``infinite'' cluster) depending on local connectivity rules. The term ``percolation'' was formally introduced in 1957 by mathematicians Broadbent and Hammersley~\cite{Broadbent_Hammersley_1957}, who modeled fluid flow through random lattices. Their work laid the groundwork for analyzing probabilistic connectivity in systems with random structures, such as crystals or mazes.
Percolation theory mathematically evolved through foundational contributions by Broadbent, Hammersley, Sykes, Essam, Fisher, and others, advancing critical thresholds, phase transitions, and high-dimensional models, while Ukrainian mathematicians (e.g., Fil'chakov and Lavrik)  expanded its frameworks, collectively shaping modern statistical mechanics and probability~\cite{SykesEssam1964,M_E_Fisher_1967,Percolation_filchakov1984,Kirkpatrick1973}. Hugo Duminil-Copin's breakthrough work rigorously resolved longstanding probabilistic challenges in understanding phase transitions, particularly in three- and four-dimensional systems~\cite{Hugo2021}, and thereby significantly advancing the mathematical foundations of statistical physics and earning him the 2022 Fields Medal.
Percolation theory has propelled research in statistical physics regarding system connectivity and critical phenomena, investigated connectivity and robustness in communication and transportation networks, and explored in geometric group theory the relationship between group properties and the non-amenability of percolation. Moreover, it has raised many natural questions and conjectures in probability theory and combinatorics, providing impetus for the development of new methods.

The breakthroughs integrate percolation theory into quantum networks to enhance entanglement distribution and error correction.
The percolation theory posits that within the realm of random graph theory, a critical probability threshold exists. Once this threshold is surpassed, a giant node cluster will propagate across the entire network~\cite{Classical_percolation1999_2,Classical_percolation2006_2}. This theory primarily investigates the pathways that can initiate at one extremity and conclude at the opposite end, effectively permeating the entire network. In 2007, this theory is utilized to analyze the process of entanglement distribution between two long-distance parties through repeaters~\cite{Percolation_CEP2007}, for which, the concept of the classical entanglement percolation and the quantum entanglement percolation were first introduced which  demonstrates that quantum entanglement swapping protocol~\cite{MAMA1993,Enetanglement_swapping_1999} can reduce the threshold of the entanglement in specific regular lattice, thereby improving the robustness of entanglement distribution in this QN. This distinctive phenomenon highlights the quantum advantage over classical bond percolation~\cite{Classical_percolation1980}. Nowadays the entanglement percolation has become a fascinating concept that lies at the intersection of quantum information theory~\cite{Nielsen2010}, statistical mechanics~\cite{Statistical_mechanics1956}, and network science~\cite{Statistical_ComplexNet2002,Complex_net_percolation2010}. It examines the propagation and distribution of quantum entanglement across complex networks~\cite{Complex_network2006,Complex_net_2011}, drawing parallels with classical percolation theory~\cite{Classical_percolation1980,Classical_percolation1980_2,Classical_percolation1982,Classical_percolation1994,Classical_percolation1999,Classical_percolation2015}, which studies the connectivity of networks under random occupation of their links or nodes.

{
Since then, the research {on} entanglement percolation is generalized to mixed-state QNs~\cite{Percolation_mixed2010,Percolation_noisy}, multipartite-state QNs~\cite{Percolation_multi2010} and complex network models~\cite{Percolation2008,Percolation2009,Percolation_complex2009,Percolation_cluster2011,Percolation_LargeScale2022,path-percolation_UVflower2024}. Entanglement percolation theory has made great progress. However, the network percolation protocols in these studies remained probabilistic. A more effective deterministic entanglement transmission (DET) scheme~\cite{DET2023} was found to correspond to a new
variant of percolation, known as concurrence percolation theory (ConPT)~\cite{Percolation_ConPT2021,Percolation_LargeScale2022}.
In the DET scheme, each local operation and classical communication (LOCC) step succeeds deterministically (with a probability of one), thus eliminating interference from classical randomness and {in contrast to} its probabilistic predecessors~\cite{Percolation_CEP2007}; studying the corresponding ConPT subsequently reveals that large-scale DET fundamentally differs from network percolation under classical noise~\cite{netw-percolation_ceah00} and bears a distinct quantum interpretation. More recently, new correspondences have been identified, e.g.,~between quantum memories and continuum percolation~\cite{Percolation_QMemo2024}, as well as between entanglement routing and path percolation~\cite{path-percolation_mhrk24}. These correspondences further underscore the profound changes in percolation landscapes introduced by additional quantum facilities such as repeaters. It has also revealed the correspondence between quantum memories and continuum percolation~\cite{Percolation_QMemo2024}, as well as between entanglement routing and path percolation~\cite{path-percolation_mhrk24}.

The theory of entanglement percolation uses statistical and numerical methods to investigate the pathways to achieve entanglement distribution within quantum networks. However, existing studies still rely substantially on numerical simulations and approximation methods, which limits the analytical characterization of the underlying mechanisms and the asymptotic behavior of large-scale entanglement percolation. These limitations are not necessarily intrinsic to percolation problems. Recently, some amazing works have been obtained that explore mathematical  explanations for the emergent phenomana of complex large models (\cite{math-model-bigmodel1, math-model-bigmodel2, Operator_model_emergent}). Specially in ~\cite{Operator_model_emergent}, Shu, Jia, Meng and Xu have developed operator-theoretic approaches to provide mathematical descriptions of emergent behavior and scaling properties in other complex systems. In particular, a limit-theoretic framework for foundation models describes a large-scale system through the composition of basic operators and characterizes its asymptotic behavior through the limit of the resulting operator sequence. This perspective suggests that the emergence of macroscopic behavior in a large-scale system can, in suitable settings, be studied through the composition and limiting behavior of its underlying operators. Inspired by aforementioned researches, we establish an operator-theoretic model for entanglement percolation in quantum networks. Within this framework, the elementary operations involved in entanglement percolation are represented by superoperators, and the overall percolation process is formulated through their composition. This enables us to analytically characterize the ordering of percolation operations, the existence of percolation thresholds, and the associated saturation behavior.

\section{Preliminary}

In this section, we {introduce} the operator-theoretic representation of the fundamental concepts of quantum mechanics and outline the fundamental topics of  percolation theory.

\subsection{Operator and Quantum theory}

Since the mathematical framework of quantum mechanics is based on Hilbert spaces and the theory of operators on them \cite{vonNeumann1996}, this {provides a natural framework for studying} the entanglement percolation problem in quantum networks {using the operator theory}.
Here we briefly review the fundamental concepts and properties of quantum systems and quantum states in terms of operator theory, and introduce the notation used throughout this paper.

{\bf Quantum system.} A quantum system is described by a separable complex Hilbert space $\mathcal{H}$, with vectors denoted by $|\psi\rangle\in\mathcal{H}$ and inner product of $|\psi\rangle$ and $|\phi\rangle$ by $\langle\phi|\psi\rangle$. In this paper, a system is called discrete-variable
(DV) if $\dim \mathcal{H}<\infty$; is single-mode continuous-variable (CV) if $\dim \mathcal{H}=\infty$ with a given orthonormal basis---the Fock basis $\{|n\rangle\}_{n=0}^\infty$---as well as the annihilation operator $\hat{a}$ and its adjoint $\hat{a}^\dag$---the creation operator---determined by
\begin{eqnarray*}
    \hat{a}|n\rangle=\sqrt{n}|n-1\rangle, \quad \hat{a}^\dagger |n\rangle=\sqrt{n+1}|n+1\rangle.
\end{eqnarray*}
In a CV system $\mathcal{H}$, the position operator $\hat{x}$ and the momentum operator $\hat{p}$ are defined as
\begin{eqnarray*}
    \hat{x}=\sqrt{\frac{\hbar}{2\omega}}(\hat{a}+\hat{a}^\dagger),\quad\hat{p}=-i\sqrt{\frac{\hbar\omega}{2}}(\hat{a}-\hat{a}^\dagger).
\end{eqnarray*}
As usual, we set $\hbar=\omega=1$, then $\hat{x}$ and $\hat{p}$ become
\begin{eqnarray*}
    \hat{x}=\frac{1}{\sqrt{2}}(\hat{a}+\hat{a}^\dagger),\quad\hat{p}={\frac{-i}{\sqrt{2}}}(\hat{a}-\hat{a}^\dagger).
\end{eqnarray*}
Let two
systems A and B correspond to Hilbert spaces $\mathcal{H}_A$ and $\mathcal{H}_B$, respectively. Their composite
system AB is a bipartite system described by the tensor-product space $\mathcal{H}_{AB} := \mathcal{H}_A\otimes \mathcal{H}_B$.
If $\mathcal{H}_1$ and $\mathcal{H}_2$ are single-mode continuous-variable (CV) systems, then their tensor product $\mathcal{H}_{12} = \mathcal{H}_1\otimes \mathcal{H}_2$ is called a two-mode CV system. The annihilation and creation operators acting on $\mathcal{H}_{12}$ satisfy the canonical commutation relations (CCRs):
\begin{eqnarray*}
    \left[\hat{a}_k,\hat{a}^\dagger_l\right]=\delta_{kl}I,\quad\left[\hat{a}_k,\hat{a}_l\right]=\left[\hat{a}^\dagger_k,\hat{a}^\dagger_l\right]=0,\quad k,l=1,2.
\end{eqnarray*}
As is standard in the physics literature, we identify $\hat{a}_1$ with $\hat{a}_1 \otimes I_2$ and $\hat{a}_2$ with $I_1 \otimes \hat{a}_2$, where $I_k$ is the identity operator on $\mathcal{H}_k$, respectively. We adopt this identification whenever no confusion arises. For any positive integer $n$, the $n$-mode CV system is defined similarly.

{\bf Quantum state.} For a quantum system described by Hilbert space $\mathcal{H}$, a quantum state can be
described as a density operator $\rho$, which is a positive operator with unit trace acting on $\mathcal{H}$. $\rho$ is
called a pure state if $\rho$ is a rank-one projection, otherwise, $\rho$ is called a mixed state. As the range
of a pure state $\rho$ is one-dimensional, which is spanned by a unit vector $|\psi\rangle\in\mathcal{H}$, thus $\rho$ can be
written as $\rho=|\psi\rangle\langle\psi|$, the rank-1 projection with range spanned by $|\psi\rangle$ , and we also say that every
unit vector is a pure state. Moreover, every mixed state $\rho$ can be written as a generalized convex
combination of pure states:
$$ \rho=\sum_{j}p_j|\psi_j\rangle\langle\psi_j|, \quad p_j\geq 0, \sum_jp_j=1.
$$
Denote by $\mathbb S(\mathcal{H})$ the set of all density operators on $\mathcal{H}$.
An important case in quantum information theory is the qubit system, where $\mathcal{H}=\mathbb{C}^2$ and its orthonormal basis may be denoted by
$\{|0\rangle,|1\rangle\}$. In particular, when these basis vectors are represented as $|0\rangle=\begin{pmatrix}
    1\\0
\end{pmatrix}$
and
$|1\rangle=\begin{pmatrix}
    0\\1
\end{pmatrix}$,
$\{|0\rangle,|1\rangle\}$ is called the computational basis. If $\dim\mathcal{H}=d$,  $\rho\in\mathbb{S}(\mathcal {H})$ is called a qudit state.

In the case of CV systems, an important subclass in $\mathbb S(\mathcal{H})$ consists of Gaussian states. Consider $n$-mode CV system $\mathcal{H}=\mathcal{H}_{12\dots n}$ with $\hat{x}_j, \hat{p}_j$, the position operator and momentum operator, respectively for $j$-th mode, $j=1,2,\dots,n$.
For a state $\rho\in \mathbb S(\mathcal{H})$, its characteristic function is defined
as $\chi(\xi):=\Tr\left[\mathcal W(\xi)\rho\right]$ for all $2n$-dimensional real vector $\xi\in\mathbb{R}^{2n}$, where ${\mathcal W}(\xi):=\exp{(i\xi^{T}\hat{R})}$ is the {Weyl operator} and
$$\hat{R}=(\hat{R}_1,\hat{R}_2,\hat{R}_3,\hat{R}_4,\dots,\hat{R}_{2n-1},\hat{R}_{2n})^T=(\hat{x}_1,\hat{p}_1,\hat{x}_2,\hat{p}_2,\dots,\hat{x}_n,\hat{p}_n)^T.$$
Here,
$$(f_1,f_2,\dots,f_k)^T=\begin{pmatrix}
    f_1\\f_2\\\vdots\\f_k
\end{pmatrix},\quad \begin{pmatrix}
    f_1\\f_2\\\vdots\\f_k
\end{pmatrix}^T=(f_1,f_2,\dots,f_k).$$
The state $\rho$ is called an $n$-mode Gaussian state if its characteristic function has the form
$$\chi(\xi)=\exp{\left(-\frac{1}{4}\xi^{T}\Gamma\xi+id^{T}\xi\right)},$$
where
$$d=(\langle\hat{R}_1\rangle,\langle\hat{R}_2\rangle,\ldots,\langle\hat{R}_{2n}\rangle)^T\in\mathbb{R}^{2n}$$
is called the mean vector of $\rho$ and $\Gamma=(\gamma_{kl})\in\mathcal{M}_{2n}(\mathbb{R})$, the algebra of all $2n\times 2n$ matrices over the real field $\mathbb{R}$, is the covariance matrix of $\rho$ defined by
$$\gamma_{kl}=\Tr\left[\rho(\Delta\hat{R}_k\Delta\hat{R}_l+\Delta\hat{R}_l\Delta\hat{R}_k)\right]$$
where $\langle\hat{R}_k\rangle=\Tr\left[\rho\hat{R}_k\right]$ and $\Delta\hat{R}_k=\hat{R}_k-\langle\hat{R}_k\rangle$. Note that a matrix $\Gamma\in\mathcal{M}_{2n}(\mathbb{R})$ is a covariance matrix for some $n$-mode Gaussian state if and only if $\Gamma$ is real symmetric and satisfies the condition $\Gamma+iJ_{n}\geq 0$, where $\underbrace{J_n=J\oplus J\oplus\cdots \oplus J}_{n}\in\mathcal{M}_{2n}(\mathbb{R})$ with
$J=\begin{pmatrix}
    0 & 1\\
    -1 & 0
\end{pmatrix}$.
In this work, we {mainly} treat two-mode Gaussian states when discuss CV systems.

\textbf{Quantum operation}. The most general quantum operation $\mathcal{E}$ is a trace-nonincreasing completely positive (CP) linear map transforming one state into another state, which can be described as a probabilistic (stochastic) physical process
\begin{eqnarray*}
    \rho\to \frac{\mathcal{E}(\rho)}{\Tr[\mathcal{E}(\rho)]}.
\end{eqnarray*}
where $$\Tr[\mathcal{E}(\rho)]\leq 1.$$
Such a map can be expressed via the Kraus operator decomposition~\cite{Kraus1983}
\begin{eqnarray}\label{eq-Kraus}
    \mathcal{E}(\rho)=\sum_i V_i\rho V_i^\dagger,
\end{eqnarray}
for some $V_i\in\mathcal{B}(\mathcal{H})$ (the von Neumann algebra of all bounded linear operators acting on $\mathcal{H}$) with the constraint $\sum_i V_i^\dagger V_i\leq I$.
The success probability $p$ of this operation is $\Tr[\mathcal{E}(\rho)]$, which depends on the input state $\rho$.
A CP map is deterministic (i.e., a quantum channel) if and only if it is trace-preserving ($\Tr[\mathcal{E}(\rho)]=1$ for all $\rho$), corresponding to the equality $\sum V_i^\dagger V_i=I$ in Eq.~\eqref{eq-Kraus}.
One example is the quantum gate, which is a unitary operator $U:\mathcal{H}\to\mathcal{H}$ acting on states as
\begin{eqnarray*}
    |\psi\rangle\mapsto U|\psi\rangle \quad {\rm or }\  \rho\mapsto U\rho U^\dagger.
\end{eqnarray*}
Another example is the quantum measurement which is a collection $\left\{M_m\right\}$ of measurement operators satisfying the completeness equation $\sum_m M_m^\dagger M_m=I$. Here, the index $m$ refers to the measurement outcomes that
may occur with the probability $p_m=\Tr\left[M_m\rho M_m^\dagger\right]$ in the experiment. The state of the system after the measurement of outcome $m$ is $\rho_m=p_m^{-1}M_m\rho M_m^\dagger$.

For two spatially separated systems $A$ and $B$, quantum operations between them are confined to LOCCs.
In this bipartite system, the mathematical structure of LOCC maps consist of complicated operations that can be composed out of local operations in the form of $\mathcal{E}_{AB}=\mathcal{E}_A\otimes\mathcal{E}_B$ with local quantum channels $\mathcal{E}_A$ and $\mathcal{E}_B$, and the following adaptive strategies enabled by classical feedback~\cite{LOCC2007}
\begin{eqnarray*}
    \mathcal{E}'(\rho)=\sum_{i}(A_i\otimes I_B\rho_{AB}A_i^\dagger \otimes I_B)\otimes |i\rangle_{B'}\langle i|,
\end{eqnarray*}
and
\begin{eqnarray*}
    \mathcal{E}''(\rho)=\sum_{i}|i\rangle_{A'}\langle i|\otimes (I_A\otimes B_i\rho_{AB} I_A \otimes B_i^\dagger),
\end{eqnarray*}
where $\sum_{i}A_i^\dagger A_i=I_A$, $\sum_{i}B_i^\dagger B_i=I_B$, and $A'$ and $B'$ are auxiliary systems of parties $A$ and $B$, respectively~\cite{Axiomatic_def_measure2009}.

\textbf{Quantum entanglement}. Entanglement is a quantum correlation between subsystems of a composite quantum system.
A bipartite state $\rho\in\mathcal{H}_{AB}$ is separable if it admits a convex decomposition into product states
\begin{eqnarray*}
    \rho_{AB}=\sum_j p_j \rho_A^{(j)} \otimes \rho_B^{(j)},\quad p_j>0,\quad \sum_j p_j=1,
\end{eqnarray*}
where $\rho_A^{(j)}$ and $\rho_B^{(j)}$ are states of $\mathcal{H}_A$ and $\mathcal{H}_B$, respectively, or,  trace-norm limit of such combinations if $\dim \mathcal{H}_{AB}=\infty$.  Otherwise, $\rho_{AB}$ is called {\it entangled}. Particularly, a pure state $|\psi\rangle\in\mathcal{H}_{AB}$ is separable if and only if there are $|\psi_{A}\rangle\in\mathcal{H}_{A}$ and $|\psi_{B}\rangle\in\mathcal{H}_{B}$ such that $|\psi\rangle=|\psi_{A}\psi_{B}\rangle=|\psi_{A}\rangle\otimes|\psi_{B}\rangle$.

\textbf{Schmidt decomposition and majorization~\cite{Matrix_analysis_Bhatia,VG2010,Nielsen1999,MRN}.}\label{sec-majorizaiton}
{For an arbitrary pure state $|\psi\rangle$ in the bipartite quantum system $H_{AB}$, there always exists a product orthonormal sequence of the form $\left\{|i_n j_n\rangle\right\}$ such that the state is expressed in its \emph{Schmidt decomposition}
\begin{eqnarray*}\label{eq-Schmidt_decomposition}
    |\psi\rangle=\sum_{n} \sqrt{\lambda_n} |i_n j_n\rangle
\end{eqnarray*}
with $\lambda_n>0$, $\lambda_1 \geq \lambda_2 \geq \cdots$, and $\sum_n \lambda_n=1$.
Moreover, these Schmidt coefficients $\{\lambda_n\}_n$ correspond to the nonzero elements of the spectrum of either reduced state $\rho_{A}=\Tr_B(|\psi\rangle\langle\psi|)\in\mathbb{S}(\mathcal{H}_A)$ and $\rho_{B}=\Tr_A(|\psi\rangle\langle\psi|)\in\mathbb{S}(\mathcal{H}_B)$. Where the partial trace $\Tr_A$ is the linear super-operator defined by $\Tr_A(A\otimes B)=\Tr(A)B$.
Note that, $|\psi\rangle$ is entangled if and only if its maximal Schmidt coefficient $\lambda_1<1$.
Particularly, the Schmidt decomposition of a two-qubit pure state always takes the Schmidt coefficients $\{\lambda,1-\lambda\}$ with $\lambda\in[1/2,1]$, such that
\begin{eqnarray}\label{eq-2qubit}
    |\psi_{\lambda}\rangle=\sqrt{\lambda}|00\rangle+\sqrt{1-\lambda}|11\rangle.
\end{eqnarray}
In the case $\lambda=1/2$, $|\psi_{\lambda}\rangle$ is called a singlet state or maximally entangled two-qubit state.
Therefore, every bipartite pure state is completely characterized by the vector formed by the descending order of Schmidt coefficients which we term the \emph{Schmidt-value vector}.
Moreover, in the infinite-dimensional bipartite system, {\it a two-mode squeezed vacuum state}  (TMSVS)---common two-mode pure Gaussian states---with a single parameter $r>0$, its Schmidt decomposition is in form of
\begin{eqnarray}\label{eq-TMSVS}
    |\psi^r\rangle=\sqrt{1-\chi^2}\sum_{n=0}^{+\infty}\chi^n|nn\rangle.
\end{eqnarray}
}

Let $\ell^1$ denote the Lebesgue space composed of discrete absolutely summable sequences, $\ell^1=\{(x_1,x_2,\ldots):\ \sum_{n}|x_{n}|<\infty,\ x_{n}\in\mathbb{R}\}$.
For given positive vectors $\vec{x},\ \vec{y}\in(\ell^{1})_{+}$, $\vec{x}$ is said to be {\it majorized} by $\vec{y}$, denoted by $\vec{x}\prec \vec{y}$ or $\vec{y}\succ \vec{x}$, if $\sum_{j=1}^{k}x_{j}^{\downarrow}\leq\sum_{j=1}^{k}y_{j}^{\downarrow}$ for all positive integers $k$ and $\sum_{j}x_{j}^{\downarrow}=\sum_{j}y_{j}^{\downarrow}$. The notation $v_{j}^{\downarrow}$ denotes the $j$-th largest component of the positive vector $\vec{v}\in \ell^1$.
{For two arbitrary bipartite states $|\phi'\rangle$ and $|\psi''\rangle$ with Schmidt-value vectors $\vec{\lambda'}$ and $\vec{\lambda''}$, $|\phi'\rangle$ can be converted into $|\psi''\rangle$ via deterministic LOCC if and only if $\vec{\lambda'}\prec\vec{\lambda''}$~\cite{Nielsen1999,MRN}.
}

\textbf{Singlet conversion probability}.
For an entangled bipartite pure state $|\varphi\rangle_{AB}$ in a finite-dimensional system, denoted by $\vec{\mu}=(\mu
_{1},\mu_{2},\cdots,\mu_{d})^{\mathrm{T}}$ $(\sum_{j}\mu_{j}=1)$ the Schmidt-value vector consisting of Schmidt coefficients of $|\varphi\rangle_{AB}$ in descending order.
Then the singlet conversion probability (SCP) $P_{\mathrm{SCP}}$~\cite{Percolation_CEP2007} of $|\varphi\rangle_{AB}$---the maximum probability to convert the state $|\varphi\rangle_{AB}$ into a singlet state
\begin{eqnarray}\label{eq-singlet}
    |\psi^{+}\rangle=\frac{|00\rangle+|11\rangle}{\sqrt{2}}
\end{eqnarray}
via LOCC---is
\begin{eqnarray*}\label{eq-SCP_pure}
P_{\mathrm{SCP}}(|\varphi\rangle_{AB})=\min\left\{1,2(1-\mu_{1})\right\}
\end{eqnarray*}
which is verified in Ref.~\cite{SCP1999}.

\textbf{Entanglement measure}.
An entanglement measure is a nonnegative function $E$ on bipartite states quantifying the entanglement. It meets two fundamental requirements at least: (i) zero on separable states and (ii) non-increasing under any LOCC $\Phi$, that is, $E(\Phi(\rho))\leq E(\rho)$~\cite{EM2000,measure2007,Axiomatic_def_measure2009}.
Experimentally, for a LOCC $\Phi$, $\Phi(\rho)$ is often given by an ensemble $\{p_j,\rho_j'\}_j$ made by all possible outcome states $\rho_j'$ with corresponding probabilities $p_j$. If the entanglement measure $E$ also satisfies the following condition
\begin{eqnarray*}
\label{eq_monotone}
  E(\rho)\geq\sum_j q_j E(\rho_j')
\end{eqnarray*}
for every $\rho$, then $E$ is called an \emph{entanglement monotone}~\cite{Axiomatic_def_measure2009}.
We introduce three entanglement measures as follows.

\emph{Concurrence.}
For a bipartite pure state $\rho=|\psi\rangle\langle\psi|$ of a DV system with Schmidt coefficients $\lambda_0$, $\lambda_1$,..., $\lambda_{d-1}$, the concurrence~\cite{Hill1997,Concurrence2001,Concurrence_convex_roof2003_2,Concurrence2013} is
\begin{eqnarray}\label{eq-concurrence}
    c(\rho)=\sqrt{2\sum_{k\neq s}\lambda_k^2\lambda_s^2}.
\end{eqnarray}
When $|\psi\rangle=|\psi_{\lambda}\rangle$, the concurrence in Eq.~\eqref{eq-concurrence} becomes $c(\rho)=\sqrt{2\lambda(1-\lambda)}$.
For a mixed state, the concurrence can be defined using the generalized convex roof construction~\cite{Concurrence_convex_roof2003_2}.

{\it $G$-concurrence.}
The \emph{$G$-concurrence}~\cite{G_concurrence2005} of the aforementioned pure state $\rho=|\psi\rangle\langle\psi|$ with Schmidt coefficients $\lambda_0$, $\lambda_1$,..., $\lambda_{d-1}$ is
\begin{eqnarray}\label{G-lebal}
    C_G(\rho)=d(\lambda_0\lambda_1\cdots\lambda_{d-1})^{1/d}.
\end{eqnarray}
The $G$-concurrence also can be extended to mixed states via convex roof construction.
When $d=2$, the $G$-concurrence reduces to the concurrence measure.

{\it $f$-negativity.}
For a bipartite CV system $AB$, a bipartite state $\rho$ of $AB$ expressed in the Fock basis $\{|ks\rangle\}_{k,s}$, its partial transpose $\rho^{T_A}$ on subsystem $A$ is defined as $\langle j,k|\rho^{T_{A}}|l,s\rangle=\langle l,k|\rho|j,s\rangle$.
The trace norm $\|\cdot\|_1$ of an operator $W$ is defined by $\|W\|_{1}:=\Tr\sqrt{W^{\dagger}W}$.
Given a real function $f$ which satisfies two conditions: (i) $f(0)=0$, and (ii) $f$ is strictly increasing in $[0,+\infty)$, the $f$-negativity $f_\mathcal{N}$ is defined as
\begin{eqnarray*}\label{eq-f_negativity}
   f_{\mathcal{N}}(\rho)= f(\mathcal{N}(\rho)),
\end{eqnarray*}
where $\mathcal{N}$ is the entanglement negativity defined by
\begin{eqnarray*}
    \mathcal{N} (\rho)=\frac{\|\rho^{T_A}\|_1-1}{2}.
\end{eqnarray*}
All $f$-negativity measures are entanglement measures.

\emph{Ratio negativity~\cite{Ratio_negativity2024}.}
The ratio negativity $\chi_{\mathcal{N}}$ is an entanglement measure defined by
\begin{eqnarray}\label{eq-ratio_negativity}
    \chi_{\mathcal{N}}(\rho)=\frac{\|{\rho}^{T_{A}}\|_{1}-1}{\|{\rho}^{T_{A}}\|_{1}+1}
\end{eqnarray}
for a state $\rho$ of $AB$.
Particularly, considering two-mode CV system, the ratio negativity of the TMSVS $|\psi^r\rangle$ is $\chi=\tanh r$.

\textbf{Quantum network}. A quantum network (QN) consists of several separated parties and several multipartite quantum states shared by some of the parties. These quantum states in a QN are called
source states. If a QN has $N$ parties, the network state of this QN is the tensor product of all source states in this QN, which is an $N$-partite state regardless of how many parties that every source state shares. In the present paper, we focus on such quantum networks that every source state is shared by two parties, that is, every source state is a bipartite state. In this case, every QN can be represented as an undirected graph with nodes being parties and links being source states

\subsection{Percolation theory and entanglement percolation}

In this subsection, we present the basic principles and notions of percolation theory and introduce the fundamental methods for studying entanglement \yaqi{distribution} in quantum networks using percolation theory.

\textbf{Percolation theory}. Percolation theory~\cite{Classical_percolation1999_2,Classical_percolation2006_2,Percolation_review2021} can be viewed as a theory of connectivity in a graph.
Consider a network connecting two defined boundaries, comprising nodes and links connecting nodes. Here, each link is weighted by its connection probability.
The sponge-crossing probability, denoted $P_{\mathrm{SC}}$, is defined as the sum of path probabilities over all connecting paths between two boundaries. This quantity was central to early studies of bond percolation on 2-dimensional (and higher-dimensional) lattices (Fig.\ref{fig-lattice})~\cite{SykesEssam1964,cross-probab-square_k80,Wierman1981,Classical_percolation1999_2,correlation_length_2000}.
For infinite-size networks with uniform link probability $p$, there exists $p_{\rm th}>0$  such that $P_{\rm SC}=0$ when $p\leq p_{\rm th}$ and $P_{\rm SC}>0$ when $p>p_{\rm th}$. This number $p_{\rm th}$ is called the critical threshold. The thresholds for several regular lattices are summarized in Table~\ref{table_threshold}.

\begin{table}[t!]
\begin{threeparttable}
  \caption{Series-parallel rules for DV DET scheme.}
  \label{table_threshold}
  \begin{tabular}{cc}
    \hline\hline\noalign{\smallskip}
     Lattice type & Classical critical threshold $p_{\mathrm{th}}$ \\
     \noalign{\smallskip}\hline\noalign{\smallskip}
    Square      & $1/2$~\cite{cross-probab-square_k80}\\
     Triangular  & $2\sin{(\pi/18)}$~\cite{SykesEssam1964}\\
     Honeycomb   & $1-2\sin{(\pi/18)}$~\cite{SykesEssam1964}\\
     Bethe (degree $k$)      & $1/(k-1)$~\cite{correlation_length_2000}\\
     \noalign{\smallskip}\hline
  \end{tabular}
  \end{threeparttable}
\end{table}

\begin{figure}[h!]
    \centering
     \includegraphics[width=120pt]{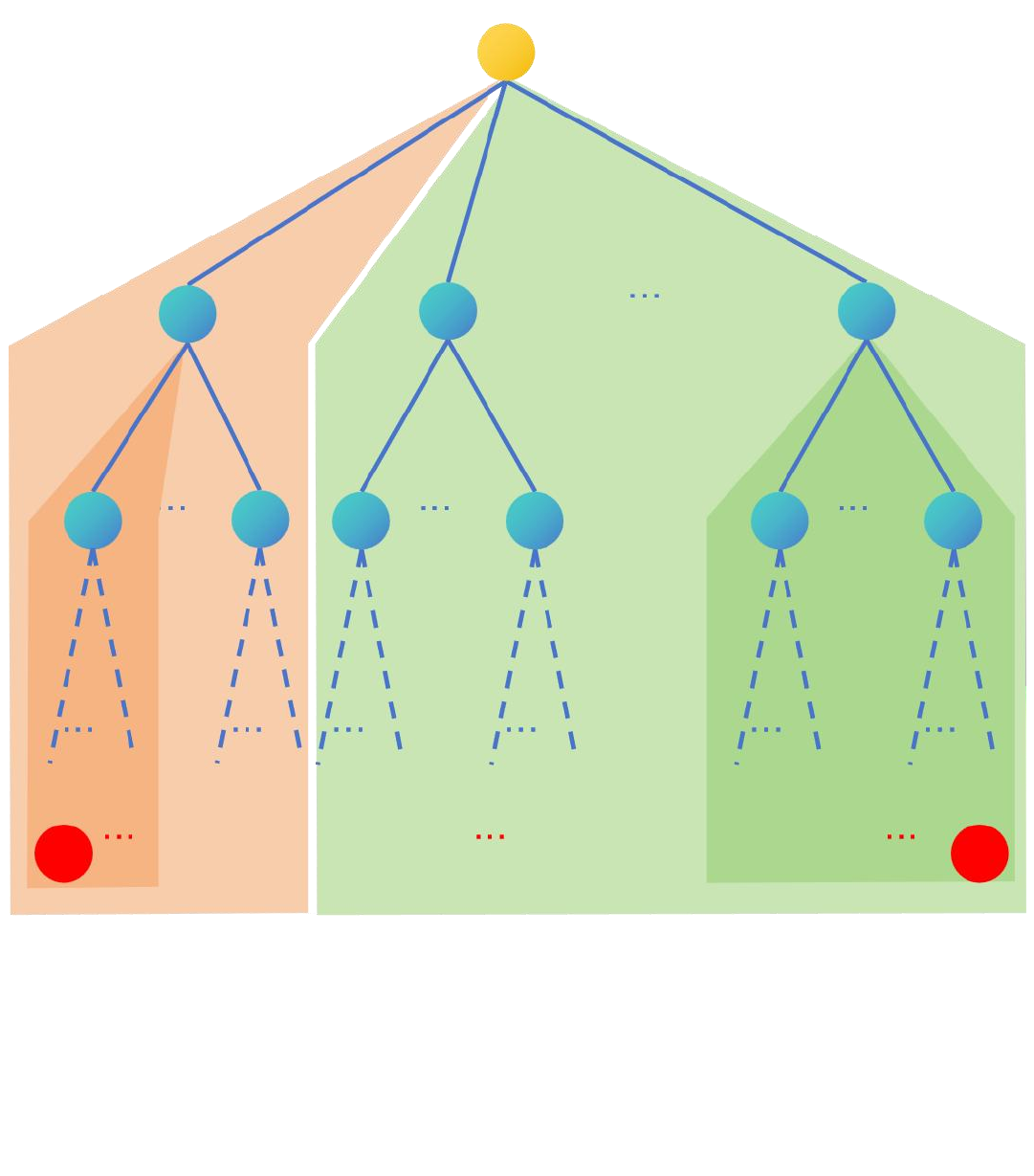}
     \vspace{-30pt}
    \caption{Bethe lattice (infinite homogeneous Cayley tree) of degree $k$ (i.e., each node is incident to $k$ links) and network depth $l$ (the path length from the yellow node to the red nodes).}\label{fig-lattice}
\end{figure}

\textbf{Classical entanglement percolation (CEP)}.
Ac\'{i}n et al. proposed the concept of \emph{classical entanglement percolation} for entanglement \yaqi{distribution} in QNs composed of pure states, as described in Ref.~\cite{Percolation_CEP2007}. In this framework, a QN is represented as a graph with nodes and links, where the nodes correspond to spatially separated parties, and each link represents a pure two-qubit state
\begin{eqnarray*}
    |\psi_{\lambda}\rangle=\sqrt{\lambda}|00\rangle+\sqrt{1-\lambda}|11\rangle,\ \lambda\geq 1/2
\end{eqnarray*}
shared by adjacent nodes connected by this link. The weight of the link is defined as the singlet conversion probability of $|\psi_{\lambda}\rangle$, $p=P_{\rm{SCP}}(|\psi_{\lambda}\rangle)=2(1-\lambda)$.

Given such a QN, the probability of establishing a singlet state between any two nodes $S$ and $T$ can be calculated using the classical bond percolation method by summing the probabilities over all possible paths connecting $S$ and $T$.

\textbf{Quantum entanglement percolation~\cite{Percolation_CEP2007}}.
{The researchers discover that entanglement swapping based {on} the Bell-state measurement (BSM) can modify the topology of the QN, resulting in the quantum entanglement percolation (QEP) combining entanglement swapping and CEP. These processes reduce the entanglement threshold for specific regular lattice structures. Specifically, given a honeycomb lattice where each node is connected by two copies $|\psi\rangle=|\psi_{\lambda}\rangle^{\otimes^2}$ of the same two-qubit state $|\psi_{\lambda}\rangle$ [Eq.~\eqref{eq-2qubit}], resulting in the SCP $P_{\rm{SCP}}(|\psi\rangle)=2(1-\lambda^2)$. Then the threshold of $\lambda$ is $\lambda_{\rm{th}}=\sqrt{1/2+\sin{(\pi/18)}}\approx 0.82$.
The BSM converts this honeycomb lattice into a triangular lattice with each link being connected with a probability $p=P_{\rm{SCP}}(|\psi_{\lambda}\rangle)=2\lambda$. The thresholds of two lattice shown in Table~\ref{table_threshold} yield the thresholds of the concurrence of $|\psi_{\lambda}\rangle$ are $c_{\rm{th}}\approx 0.7671$ for the honeycomb lattice and $c_{\rm{th}}\approx 0.7576$ for the triangular lattice, respectively. This implies the quantum advantage.}

\textbf{Deterministic entanglement percolation}.
In a QN, each link represents an entangled state, with its weight assigned as the entanglement value quantified under a chosen entanglement measure. Deterministic entanglement percolation relies on deterministic entanglement transmission (DET) schemes, primarily comprising entanglement swapping and entanglement concentration protocols.
We introduce two DET schemes for QNs, designed for DV and CV systems, respectively. Each scheme comprises a deterministic entanglement swapping protocol and a deterministic entanglement concentration protocol.

(1) \textbf{DV-based QN of two-qubit states}.
In 2021, Meng et al. proposed a deterministic entanglement transmission (DET) scheme for DV QNs distributing pure two-qubit states. This scheme consists of an entanglement swapping protocol for series networks and an entanglement concentration protocol for parallel networks which are all based on deterministic LOCC.

\emph{DV-based entanglement swapping protocol}. For a 1D chain of $N$ pure two-qubit states $|\psi_{\lambda_1}\rangle, |\psi_{\lambda_2}\rangle,\ldots,|\psi_{\lambda_N}\rangle$ (see Fig.~\ref{fig-seriesN_DV}) shared by $S$ and $R_1$, $R_1$ and $R_2$,..., and $R_{N-1}$ and $T$, respectively.
After employing the optimal deterministic swapping protocol proposed in Ref.~\cite{Percolation2008}, the final state created between $S$ and $T$ is $|\psi_{\lambda}\rangle$ with
\begin{eqnarray*}\label{eq-DET_seriesDV}
    \lambda=\frac{1+\sqrt{1-\prod_{n=1}^{N}c_n^2}}{2}
\end{eqnarray*}
where $c_n=2\sqrt{\lambda_n(1-\lambda_n)}$ is the concurrence of $|\psi_{\lambda_n}\rangle$ ($n=1,2,\ldots,N$), and the concurrence of $|\psi_{\lambda}\rangle$ is exactly $c=c_1c_2\cdots c_N$.

\begin{figure}
    \centering
    \includegraphics[width=3.5in]{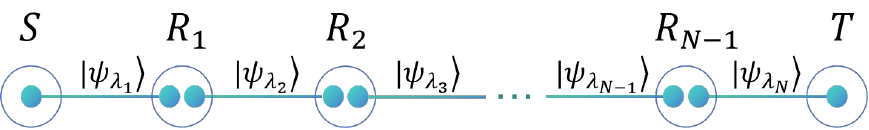}
    \caption{A series DV QN with $N+1$ nodes and $N$ states.}
    \label{fig-seriesN_DV}
\end{figure}

\emph{DV-based entanglement concentration protocol}. For a parallel QN composed {of} two distinct nodes $S$ and $T$ and $K$ pure two-qubit states $|\psi_{\lambda_1}\rangle, |\psi_{\lambda_2}\rangle,\ldots,|\psi_{\lambda_K}\rangle$ (see Fig.~\ref{fig-parallelK_DV}) shared between these nodes, the concentration protocol proposed in Ref.~\cite{Percolation_ConPT2021} converts these $K$ states into one two-qubit state $|\psi_{\lambda}\rangle$ with
\begin{eqnarray*}\label{eq-DET_parallelDV}
   \lambda=\max\left\{\frac{1}{2},\prod_{k=1}^{K}\lambda_{k}\right\}.
\end{eqnarray*}

(2) \textbf{CV-based QN of TMSVSs}.
Zhao et al. proposed a CV DET scheme consisting of an entanglement swapping protocol and an entanglement concentration based on deterministic LOCC~{\cite{NegPT}}.

\emph{CV-based entanglement swapping protocol}. Consider a one-dimensional chain as in Fig.~\ref{fig-seriesN_DV} but with $N$ TMSVSs [Eq.~\eqref{eq-TMSVS}] $|\psi^{r_1}
\rangle,|\psi^{r_2}
\rangle,\ldots,|\psi^{r_N}
\rangle$ as source states. After performing the optimal deterministic swapping protocol proposed in Ref.~\cite{P2002}, the final TMSVS established between the terminals $S$ and $T$ is $|\psi^{r}\rangle$ with
\begin{eqnarray*}\label{eq-DET_seriesCV}
    \tanh r=\prod_{n=1}^{N}\tanh r_n.
\end{eqnarray*}

\emph{CV-based entanglement concentration protocol}. Consider a parallel CV QN as in Fig.~\ref{fig-parallelK_DV} with source states being $K$ TMSVSs $|\psi^{r_1}
\rangle,|\psi^{r_2}
\rangle,\ldots,|\psi^{r_K}
\rangle$ of squeezing parameters $r_1,r_2,\ldots,r_K$ where $r_1\geq r_2\geq \ldots\geq r_K$. The concentration protocol converts these $K$ states into one TMSVS $|\psi^{r}\rangle$ with
\begin{eqnarray*}\label{eq-DET_parallelCV}
    \sinh r=\sinh r_1\prod_{k=2}^{K}\cosh r_k.
\end{eqnarray*}

\begin{figure}
    \centering
    \includegraphics[width=1.5in]{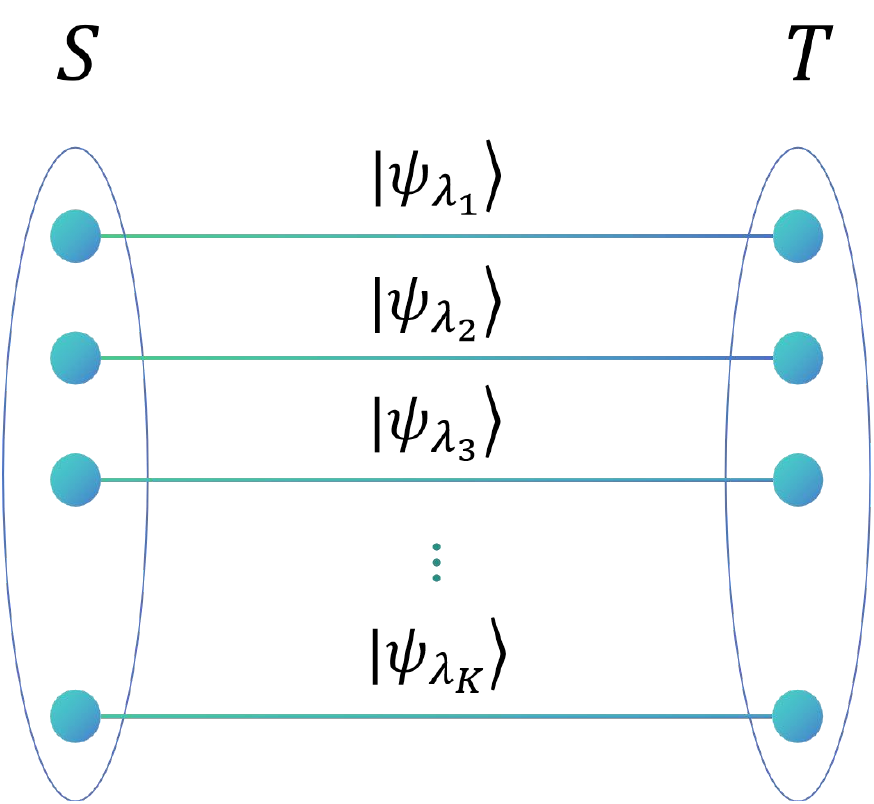}
    \caption{A parallel DV QN with $K$ states.}
    \label{fig-parallelK_DV}
\end{figure}

\section{Series-parallel networks revisited}\label{sec-network}

In this section, we {revisit} series-parallel networks through graph theory. This is a very common network model for studying percolation.

A network is an undirected graph composed of vertices (nodes) and {edges} (links) that connect them.
When two nodes can be connected through multiple links and nodes, they are connected by a path.
If no path exists between two sets of nodes, then these two sets are not considered to belong to the same network.
The length of a self-avoiding path{~\cite{self_avoiding1996}---a path that visits each node at most once---}between two nodes is defined as the number of links constituting the path. The distance between the two nodes is defined as the minimal length among all self-avoiding paths connecting the two nodes.
Two nodes are called \emph{adjacent} if the distance between them is 1.
If a node has $k$ adjacent nodes, then the degree of that node is $k$, meaning it has $k$ neighbors.

Two links $a$ and $b$ in a network are said to be {\it confluent} if there do not exist two distinct circuits $C_1$ and $C_2$ such that $C_1$ meets $a$ and $b$ in the same sense but $C_2$ meets $a$ and $b$ in opposite sense.
Ref.~\cite{Series_parallel1965} gives the necessary and sufficient conditions for series-parallel networks, i.e., a network is of series-parallel type if and only if every pair of links is confluent.
Thus, if there is a pair of links that is not confluent, then this network is not series-parallel. For example, the links between $R_1$ and $R_4$, and between $R_2$ and $R_3$ in the Wheatstone bridge network model shown in Fig.~\ref{bridge} are not confluent since the directions of links $(R_1,R_4)$ and $(R_2,R_3)$ have the same sense relative to the circuit $(R_1,R_2,R_3,R_4,R_1)$ (orange arrows), while the two links have opposite sense relative to the circuit $(R_1,R_3,R_2,R_4,R_1)$ (red arrows). It means that the Wheatstone bridge is not
a series-parallel network. However, when considering the entanglement percolation in quantum networks (QNs), we only need to consider the connected paths after giving two terminals $S$ and $T$. Therefore, we use the definition of two-terminal series-parallel graph.

\begin{figure}[h!]
{\includegraphics[width=120pt]{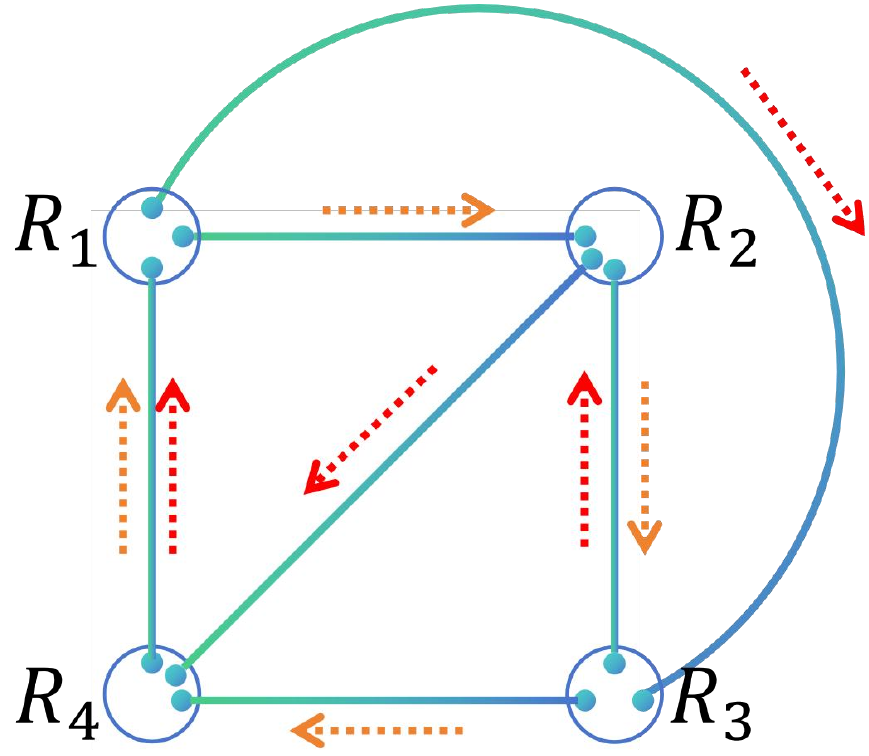}}
\caption{\quad The Wheatstone bridge. }
\label{bridge}
\end{figure}

A graph is a two-terminal series-parallel graph, with terminals $S$ and $T$, if it can be transformed into a single-link graph \yaqi{$\mathcal{K}_{ST}$} by a sequence of the following operations:

(${\rm I}$). {\emph{Simplification for two parallel links.} {If multiple parallel links share the same terminals denoted} by $R'$ and $R''$, perform the simplification: Merge these links into a single link connecting $R'$ and $R''$. Repeat until no such links remain, then proceed to operation (${\rm II}$).}

(${\rm II}$). {\emph{Simplification for two series links.}
If there are three distinct nodes denoted by $R_1$, $R_2$ and $R_3$ where node $R_2$ of degree 2 is exclusively adjacent to $R_1$ and $R_3$, perform the simplification: Replace the two series links connecting $R_2$ to $R_1$ and $R_3$ with a single direct link connecting $R_1$ and $R_3$. This process can equivalently be interpreted as the elimination of node $R_2$ within the network topology. Repeat until no such series links can be further simplified, then return to operation (${\rm I}$).}

If the graph used to represent the structure of a QN is series-parallel, then the QN is referred to as a series-parallel network. To facilitate the demonstration that a network is a series-parallel network, we present the following proposition.

\begin{proposition}
     Given a two-terminal network $\mathcal{N}_{ST}$ with two terminals $S$ and $T$, it is a series-parallel network if and only if no two connected paths from $S$ to $T$ contain a reverse sub-path consisting of two vertices.
\end{proposition}

\begin{proposition}\label{prop_s_p}
     Given a network $\mathcal{N}_{ST}$ with two terminals $S$ and $T$, it is a series-parallel network if and only if {there do not exist two distinct} paths $\mathcal{P}=(S, R_{i_1},R_{i_2},...,T)$ and $\mathcal{Q}=(S, R_{j_1}, R_{j_2},...,T)$ from $S$ to $T$, represented as node sequences, that contain sub-paths $(R_{i_k},...,R_{i_{k+n}})$ in $\mathcal{P}$ and $(R_{j_{m}},...,R_{j_{m+n}})$ in $\mathcal{Q}$ forming reverse-oriented links satisfying $R_{i_k}=R_{j_{m+n}}$ and $R_{i_{k+n}}=R_{j_m}$ for some $k$, $m$, and $n$.
\end{proposition}

\begin{proof}
 Assume there are two paths between $S$ and $T$, ($S$, $R_{i_1}$,$R_{i_2}$,...,$T$) and ($S$, $R_{j_1}$, $R_{j_2}$,...,$T$), their two sub-paths $(R_{i_k},...,R_{i_{k+n}})$ and $(R_{j_{m}},...,R_{j_{m+n}})$ satisfy $R_{i_k}=R_{j_{m+n}}$ and $R_{i_{k+n}}=R_{j_m}$. Then, after applying operations (${\rm I}$) and (${\rm II}$), the network $\mathcal{N}_{ST}$ can always be simplified to a network containing a sub-network as Fig.~\ref{fig-non_series_parallel} with a bridge (link $R_{i_k}R_{i_{k+n}}$) after operations (${\rm I}$) and (${\rm II}$). At this point, no further link or node removal can be done. The graph cannot {be reduced} to a single-link graph $S\text{-}T$ with two terminals $S$ and $T$, indicating that this network is not series-parallel.

\begin{figure}[]
\includegraphics[width=120pt]{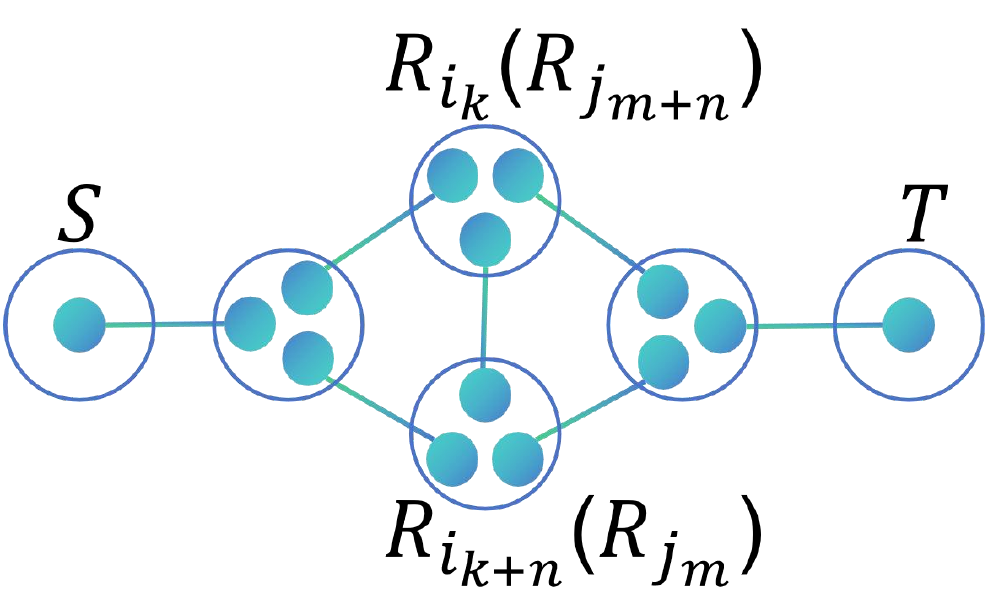}
\caption{\quad  {An example of simplified non-series-parallel network.}}
\label{fig-non_series_parallel}
\end{figure}

 If there do not exist two connected paths from $S$ to $T$ {contains} a reverse sub-path consisting of two nodes, then the network $\mathcal{N}_{ST}$ can be simplified to a network where $S$ and $T$ are connected only by one path of length 1 through sequentially applying operations (${\rm I}$) and (${\rm II}$), indicating that the network is a series-parallel network.

\end{proof}

\begin{figure}[]
    \centering
    \subfigure[]{
       \includegraphics[width=133pt]{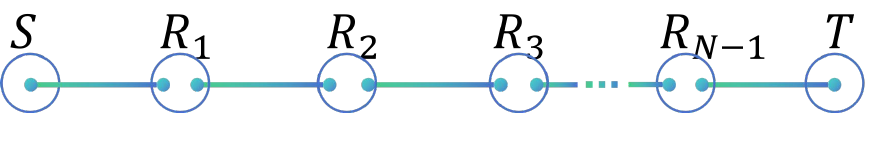}
        \label{fig-s}
    }
    \subfigure[]{
       \includegraphics[width=133pt]{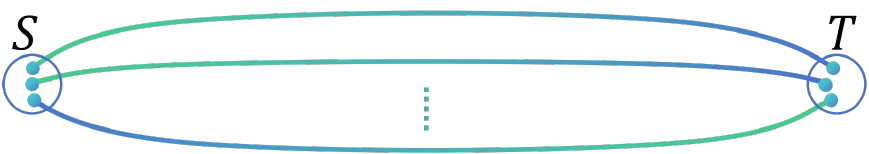}
        \label{fig-p}
    }
    \quad
    \subfigure[]{
       \includegraphics[width=133pt]{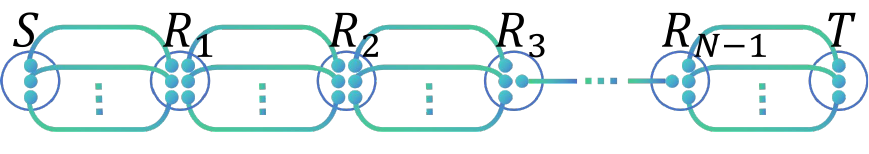}
        \label{fig-p_t_s}
    }
    \subfigure[]{
       \includegraphics[width=133pt]{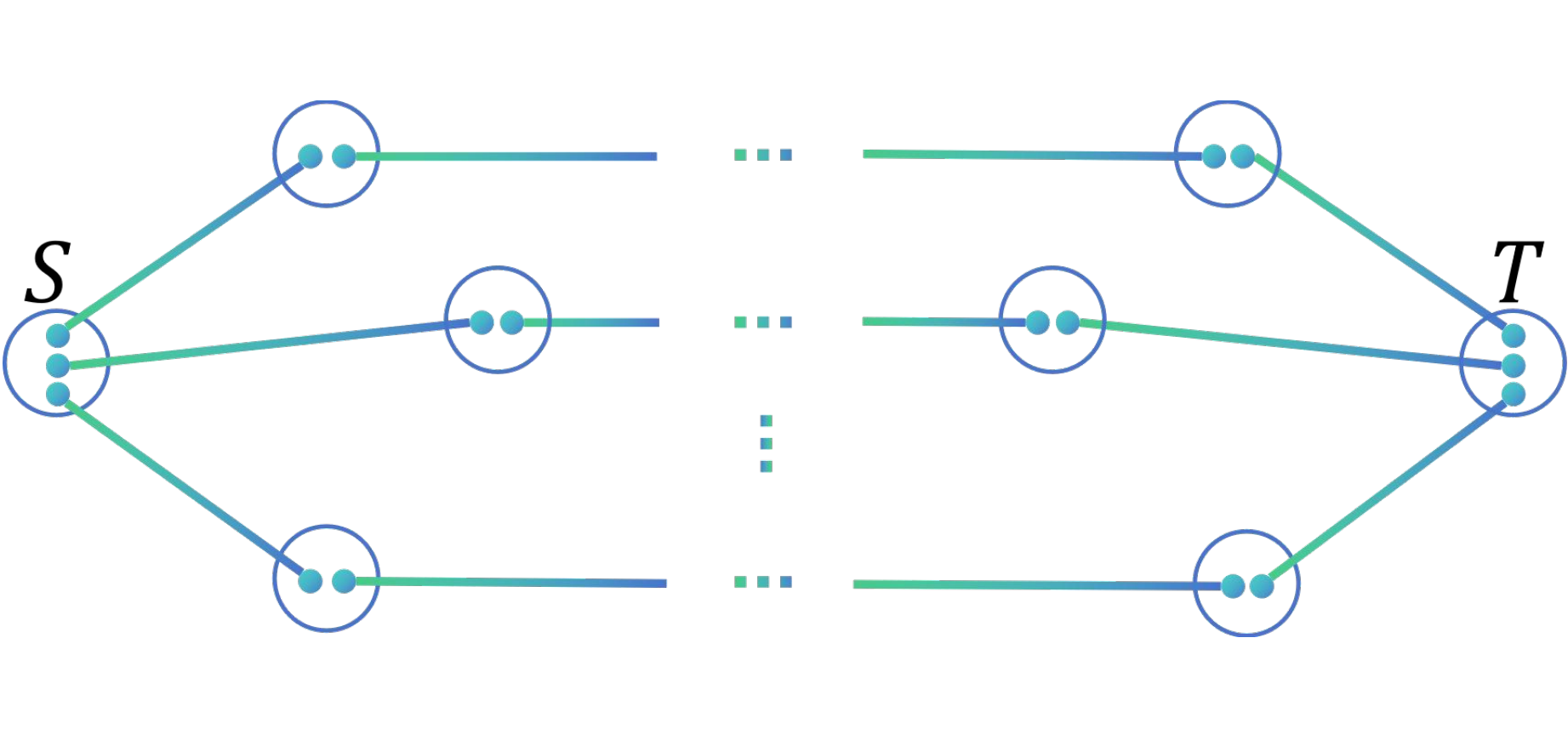}
        \label{fig-s_t_p}
    }
    \quad
    \subfigure[]{
       \includegraphics[width=133pt]{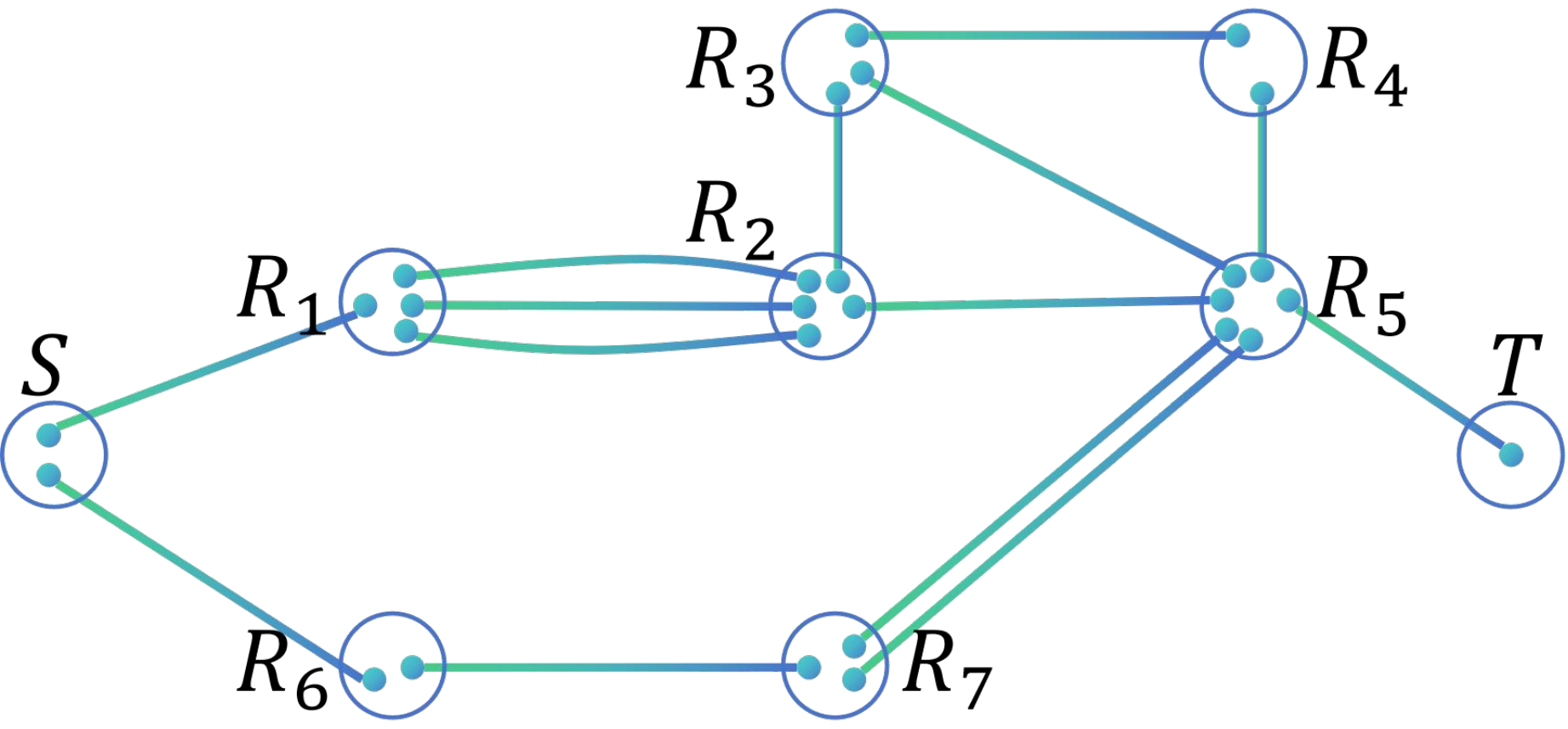}
        \label{fig-s_p}
    }
    \subfigure[]{
       \includegraphics[width=133pt]{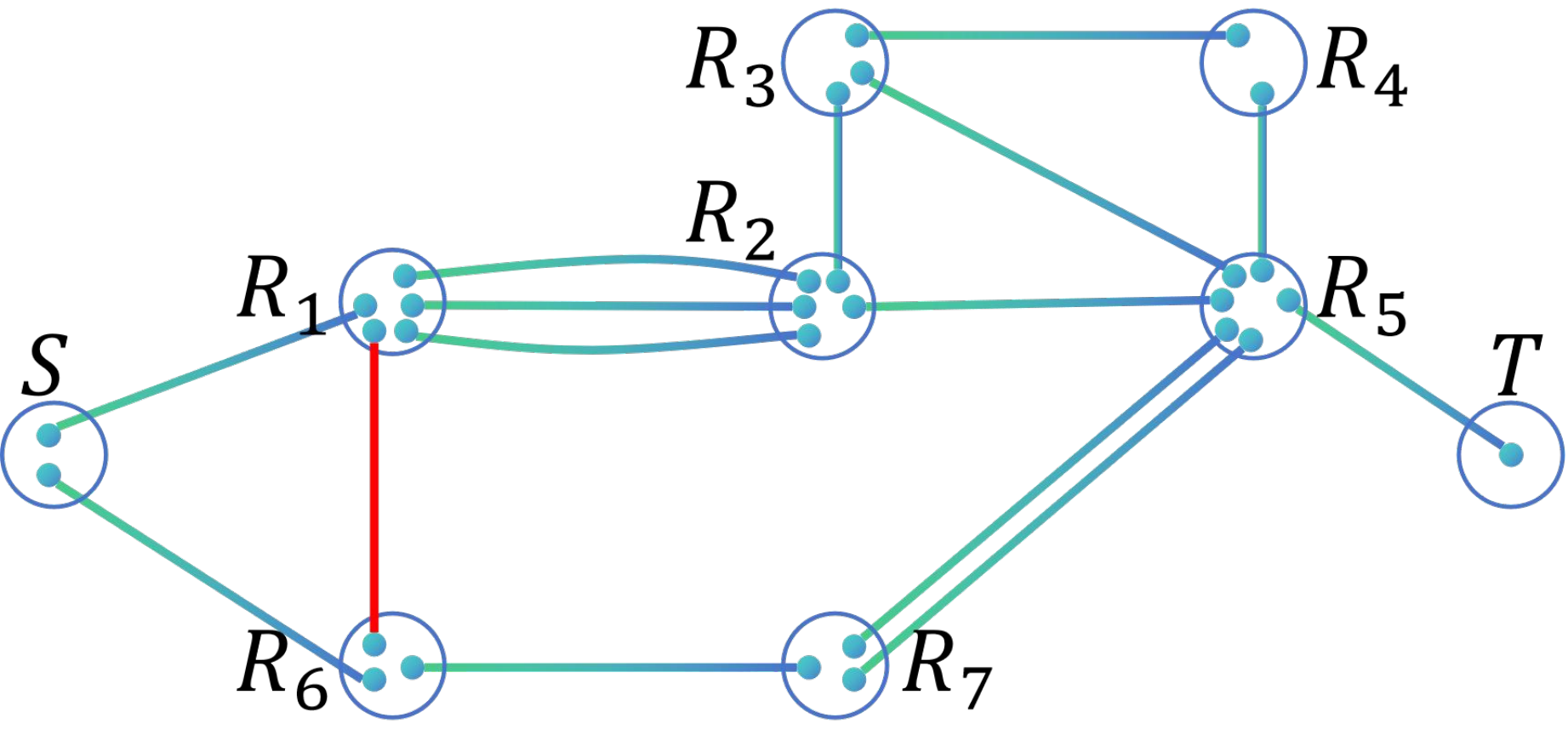}
        \label{fig-general}
    }
	\caption{Different network topologies between two terminals $S$ and $T$.~\subref{fig-s}~Series network. In this configuration, all nodes are pairwise distinct, and there is exactly one link between each pair of adjacent nodes.~\subref{fig-p}~Parallel network. In this configuration, the terminals $S$ and $T$ are connected by more than two parallel links.~\subref{fig-p_t_s}~Parallel-then-series network. This network consists of multiple two-terminal parallel sub-networks connected in series in a head-to-tail fashion.~\subref{fig-s_t_p}~Series-then-parallel network. This network can be reduced to a single link connecting the two terminals $S$ and $T$.~\subref{fig-general}~Non-series-parallel network.}
	\label{fig-series_parallel}
\end{figure}

\begin{figure}
    \centering
    \subfigure[]{
       \includegraphics[width=100pt]{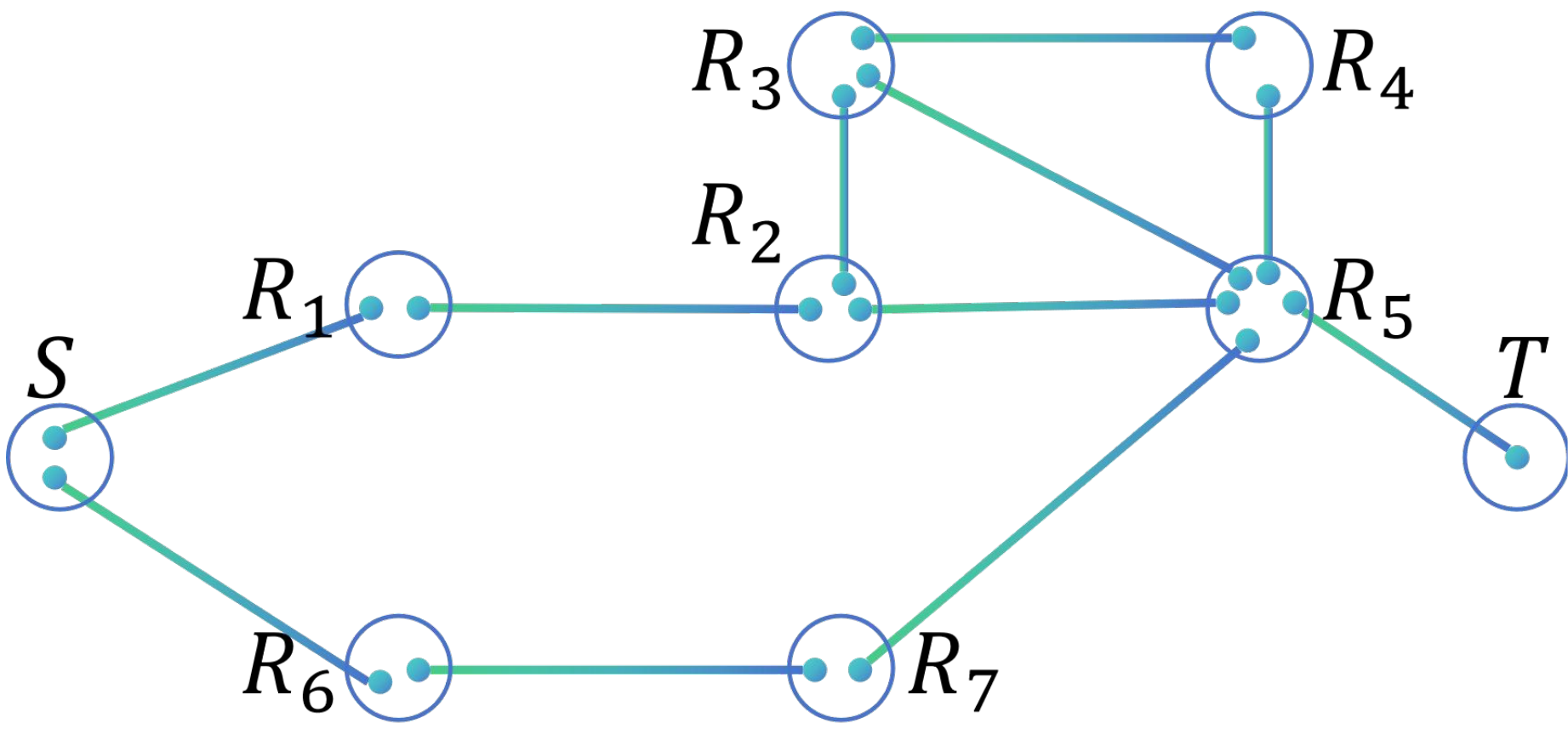}
        \label{fig-step1}
    }
    \subfigure[]{
       \includegraphics[width=100pt]{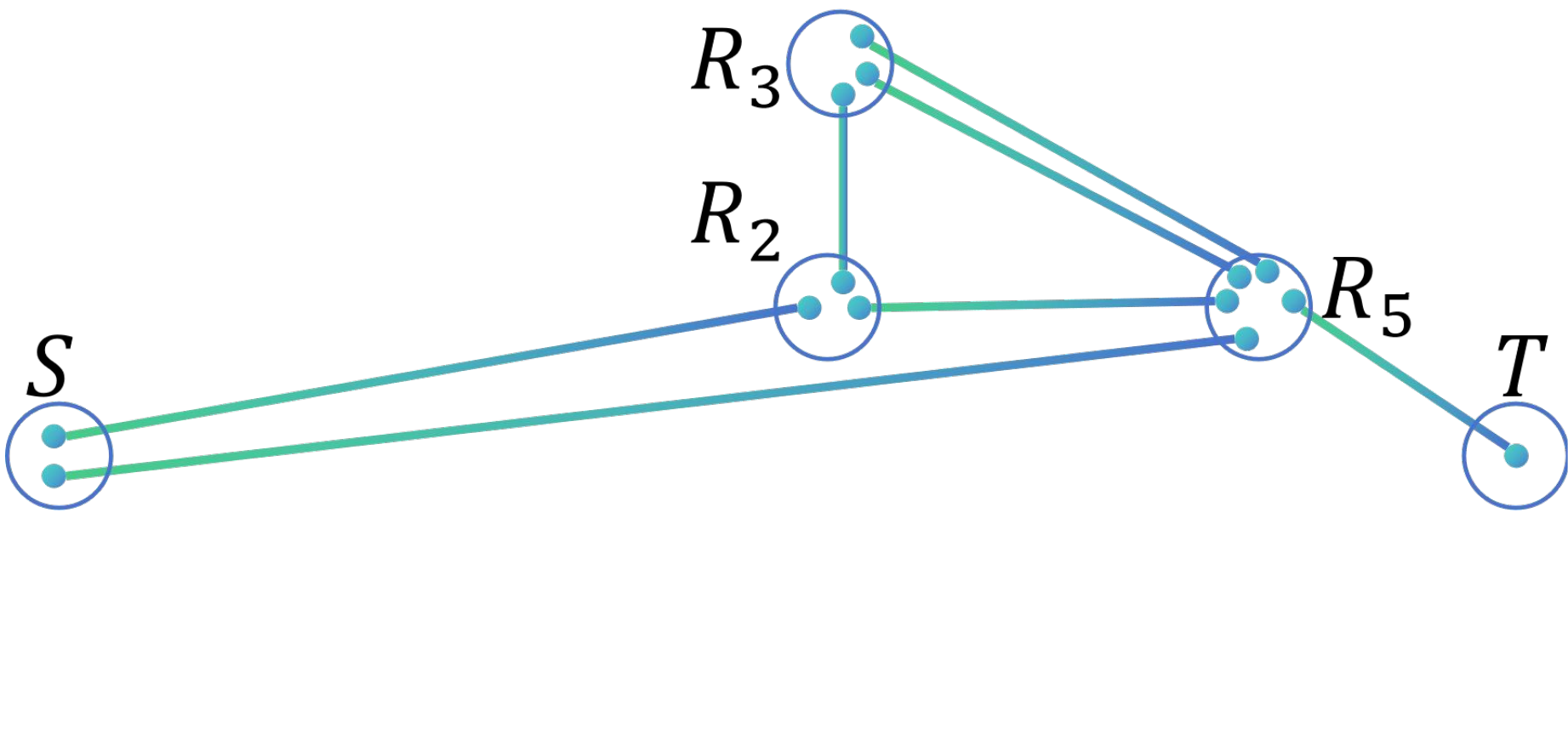}
        \label{fig-step2}
    }
    \subfigure[]{
       \includegraphics[width=100pt]{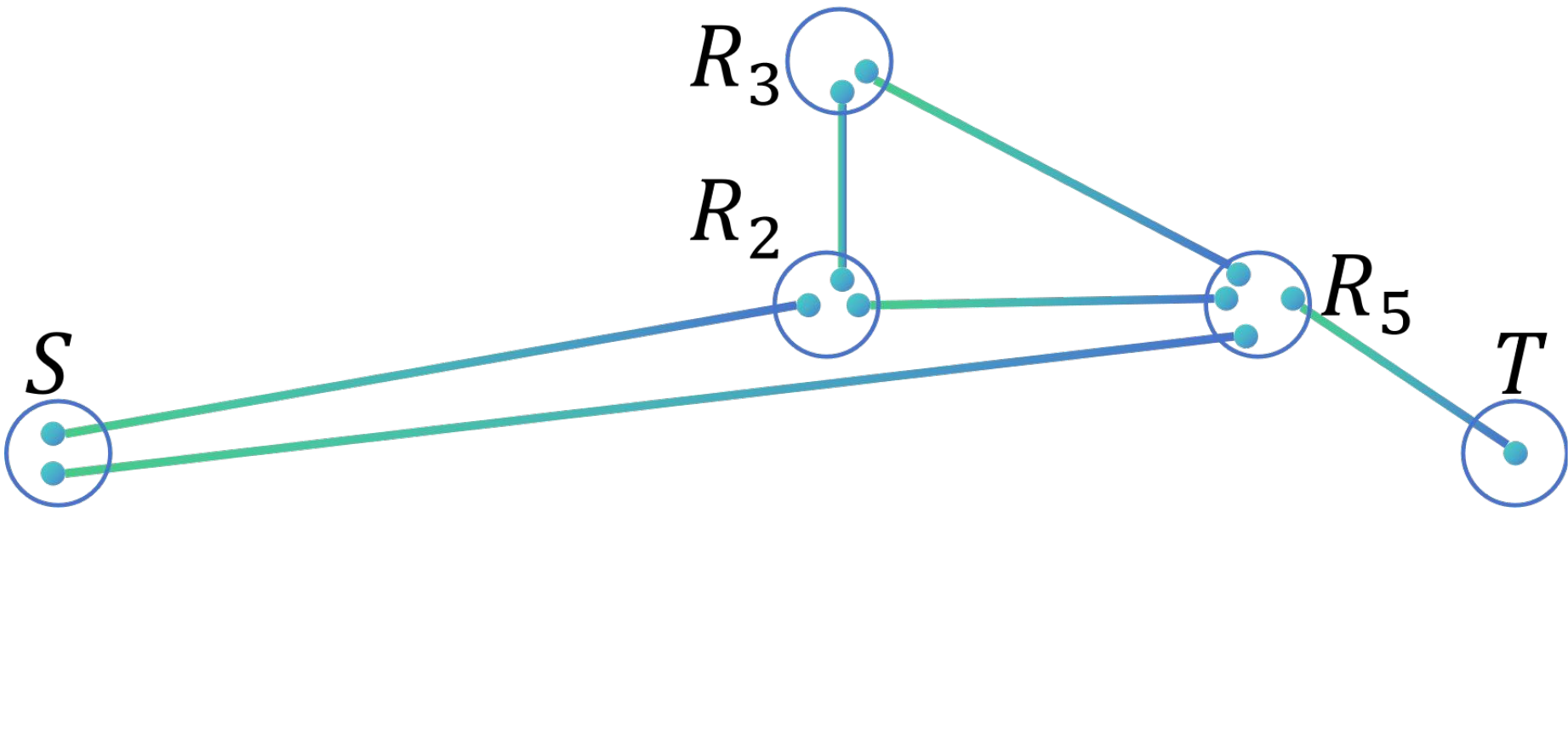}
        \label{fig-step3}
    }
    \subfigure[]{
       \includegraphics[width=100pt]{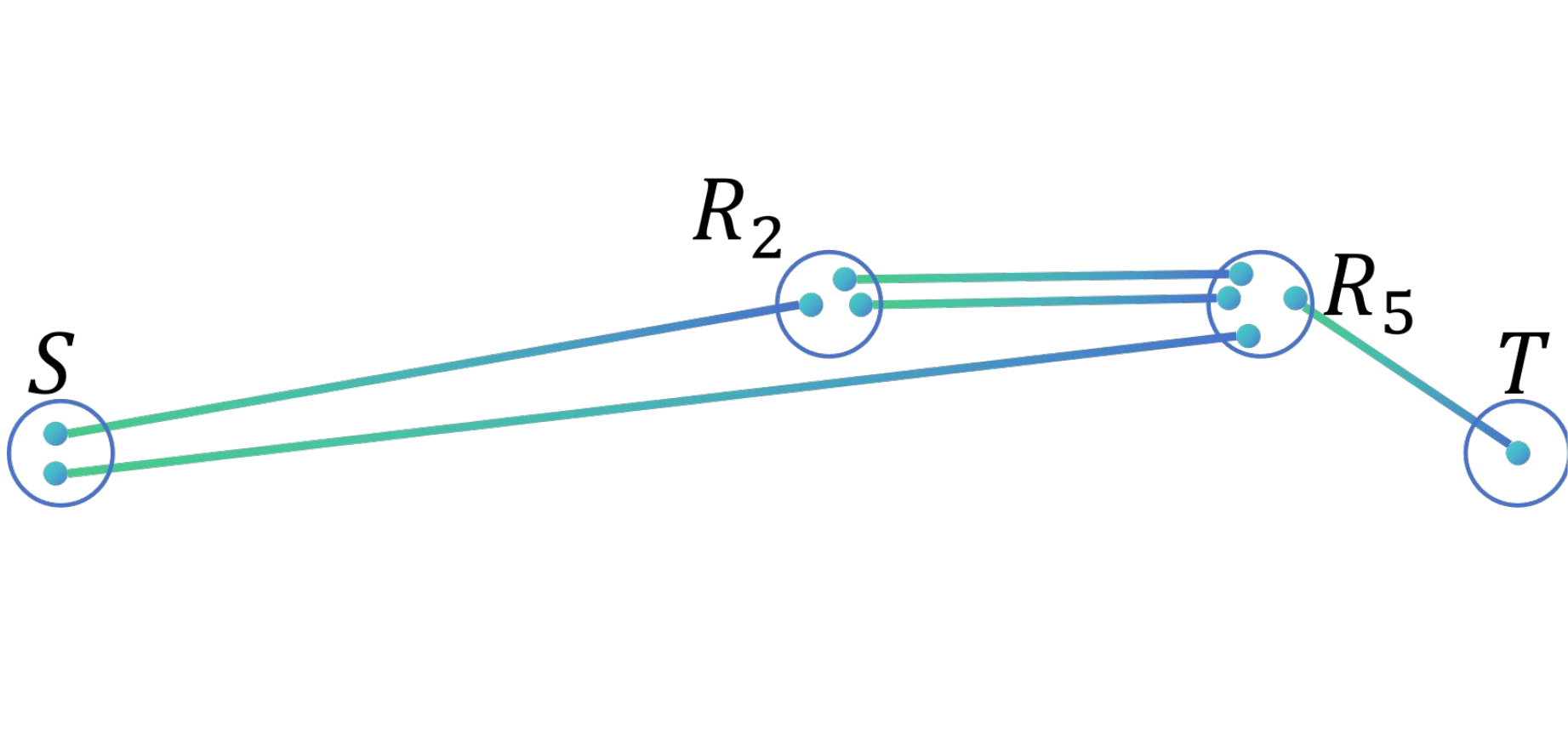}
        \label{fig-step4}
    }\\
    \subfigure[]{
       \includegraphics[width=100pt]{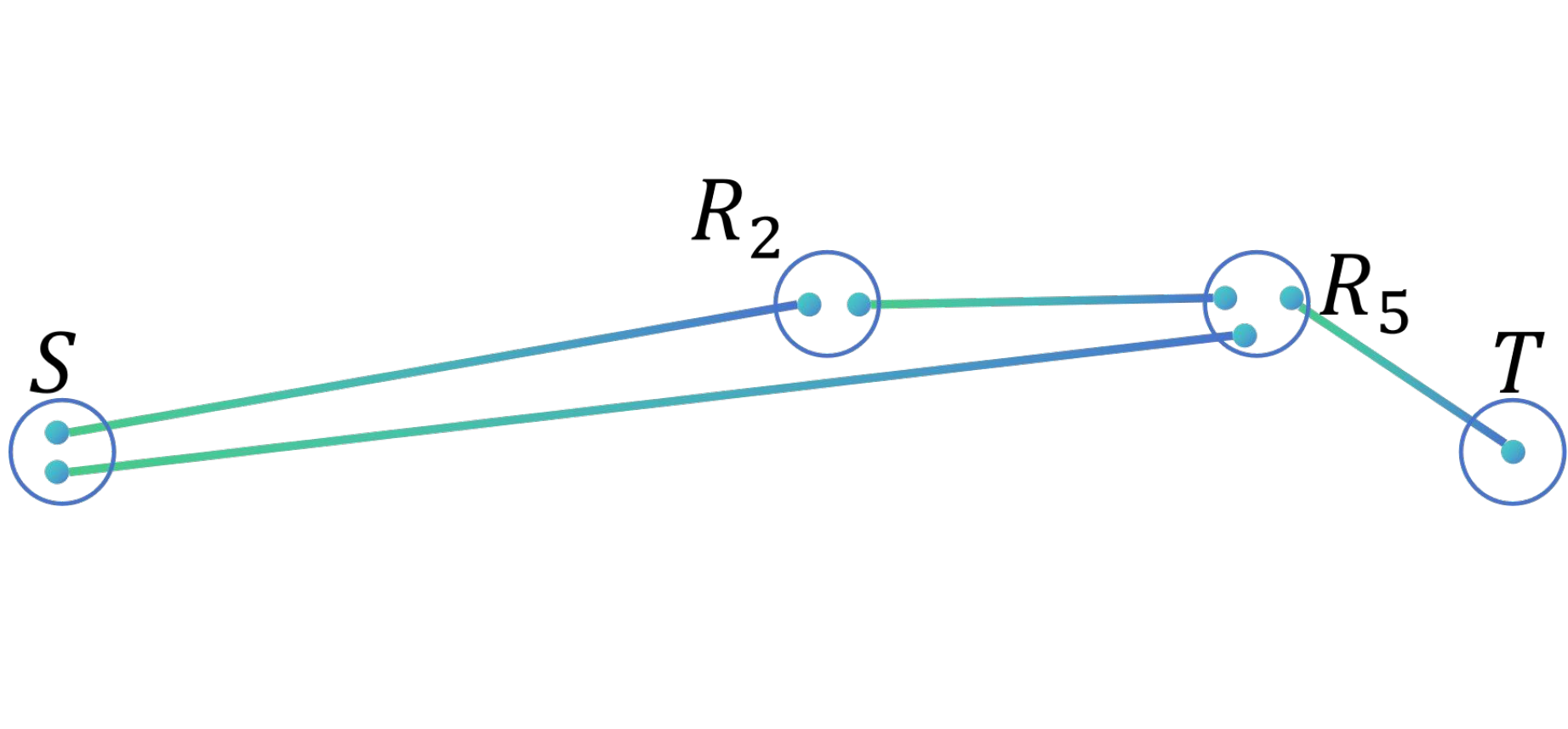}
        \label{fig-step5}
    }
    \subfigure[]{
       \includegraphics[width=100pt]{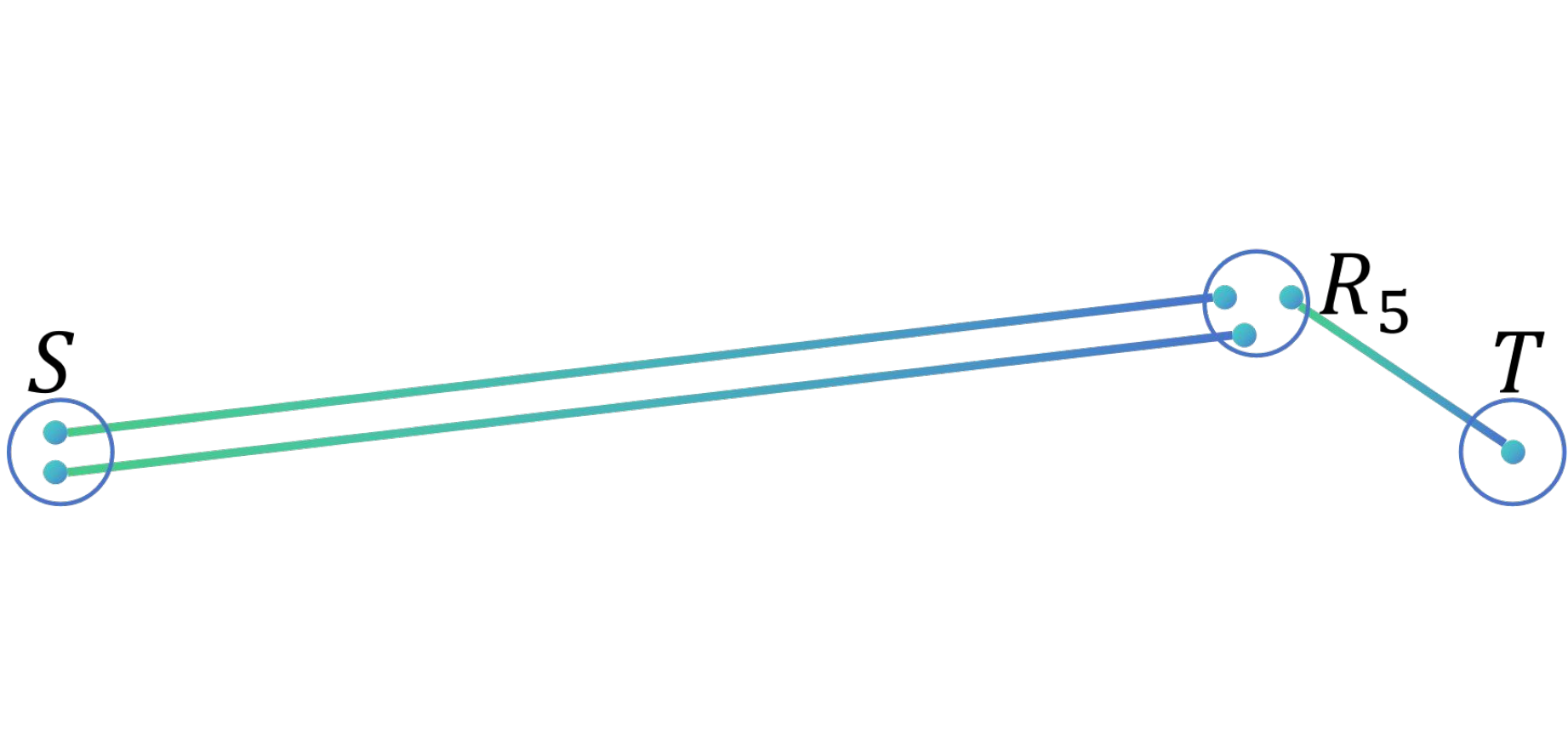}
        \label{fig-step6}
    }
    \subfigure[]{
       \includegraphics[width=100pt]{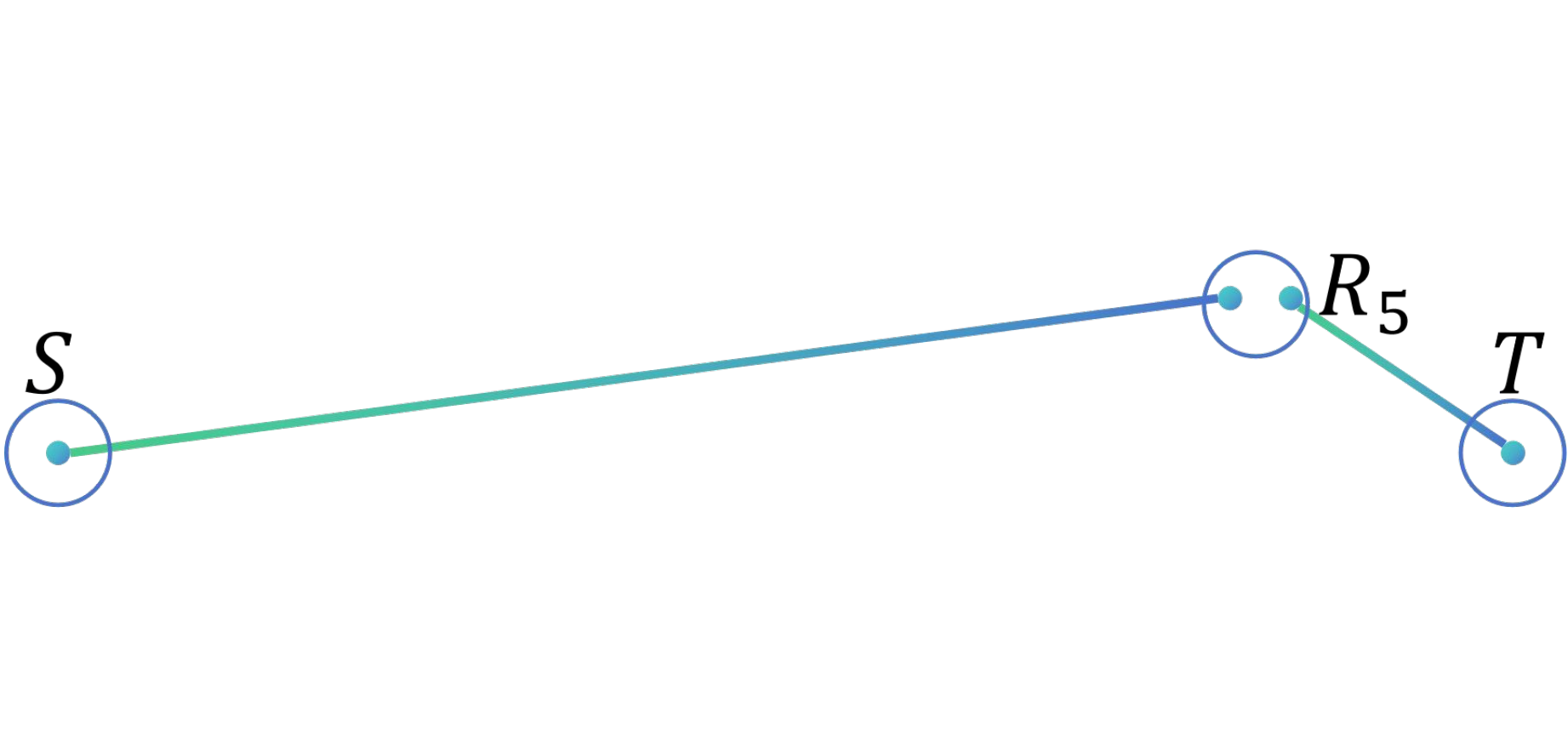}
        \label{fig-step7}
    }
    \subfigure[]{
       \includegraphics[width=100pt]{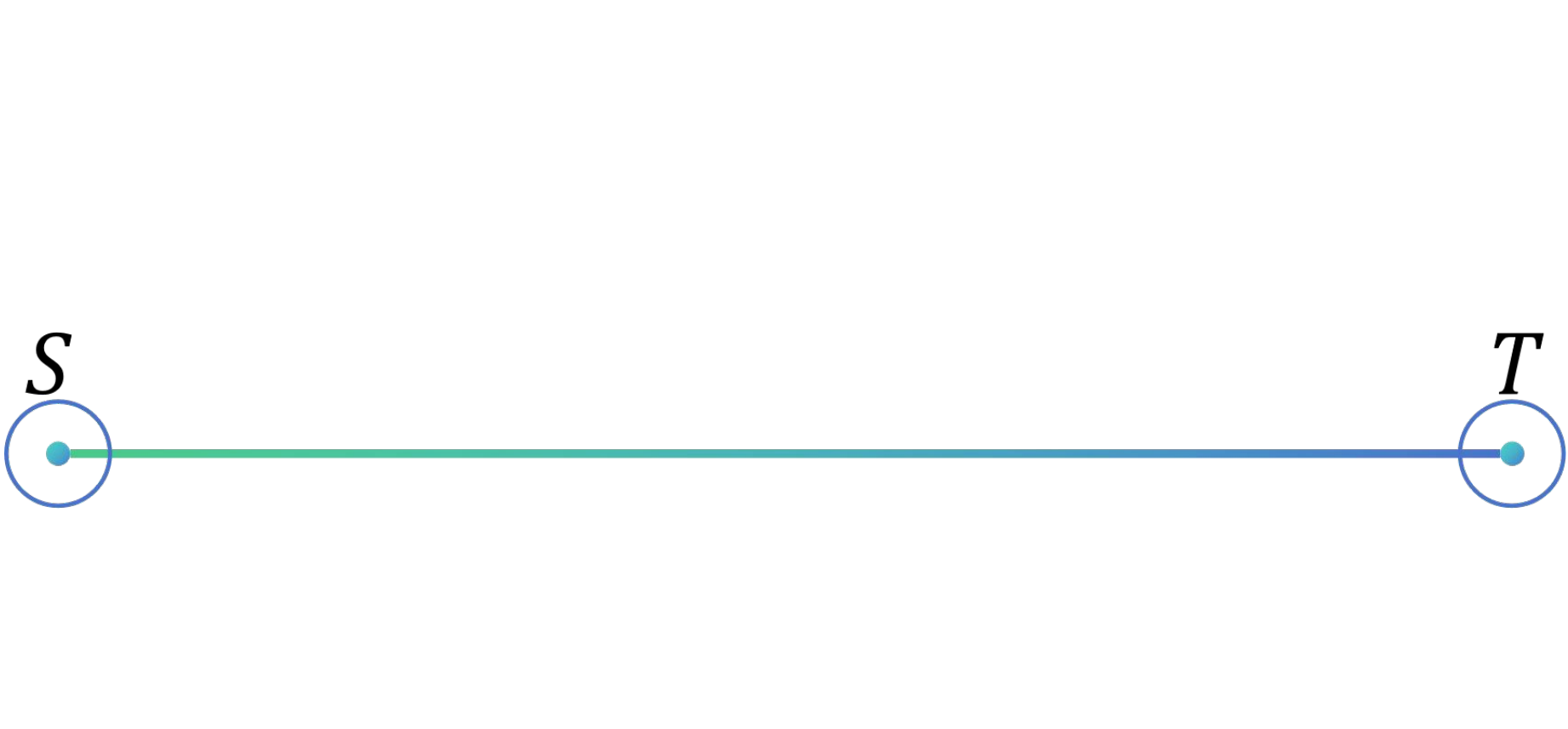}
        \label{fig-step8}
    }
    \caption{Steps of simplifying the network with terminals $S$ and $T$ shown in Fig.~\ref{fig-s_p}. }
    \label{fig-simplify_s_p}
\end{figure}

{Several typical series-parallel and non-series-parallel networks will be presented. The network topology can be characterized by distinct topological categories [Fig.~\ref{fig-series_parallel}]. As established by Proposition~\ref{prop_s_p}, all topologies between $S$ and $T$ in Figs.~\ref{fig-s}--\ref{fig-s_p} are series-parallel. For Fig.~\ref{fig-s_p}, we illustrate its simplification processes under criteria (${\rm I}$) and (${\rm II}$) in Fig.~\ref{fig-simplify_s_p}.
First, apply operation (${\rm I}$) to the parallel links (between $R_1$ and $R_2$, and $R_5$ and $R_7$) in Fig.~\ref{fig-s_p} to obtain Fig.~\ref{fig-step1}, followed by operation (${\rm II}$) on the series links in three paths $(S,R_1,R_2)$, $(S,R_6,R_7,R_5)$ and $(R_3,R_4,R_5)$, respectively, to generate Fig.~\ref{fig-step2}. Subsequently, iteratively execute operation (${\rm I}$) on the parallel links between $R_3$ and $R_5$ and operation (${\rm II}$) on path $(R_2,R_3,R_4)$, producing Figs~\ref{fig-step4}~and~\ref{fig-step5} in sequence.
Proceed by implementing (${\rm I}$) on links connecting $R_2$ and $R_5$ [Fig,~\ref{fig-step5}], then eliminate node $R_2$ via (${\rm II}$) [Fig.~\ref{fig-step6}]. Finally, simplify the parallel links between $S$ and $R_5$ through (${\rm I}$) to form the series network in Fig.~\ref{fig-step7}, and further estimate node $R_5$ to achieve the simplification process of initial network.
}

{The introduction of a `bridge' (red line) to Fig.~\ref{fig-s_p} yields a non-series-parallel topology as shown in Fig.~\ref{fig-general} which can be simplified, via processes (${\rm I}$) and (${\rm II}$), into the form of Fig.~\ref{fig-simplify_general} with $R_{i_{k-1}}=S$, $i_k=1$, $i_{k+1}=6$ and $i_{k+2}=5$.
Notably, certain network topologies inherently contain non-essential components that play no role in the connectivity between two designated terminals. Formally, in a network with terminal nodes $S$ and $T$, a sub-network is said to be \emph{topologically redundant} with respect to the connectivity between $S$ and $T$ if it is not contained in any self-avoiding path between $S$ and $T$. For example, in the network shown in Fig.~\ref{fig-s-p}, the sub-network formed by the orange nodes ($A_1$ and $A_2$) and the corresponding orange links is topologically redundant. This implies that, with respect to $S$-$T$ connectivity, the network is effectively equivalent to the topology shown in Fig.~\ref{fig-s_p}.
}
\begin{figure}
    \centering
    \subfigure[]{
    \includegraphics[width=1.59in]{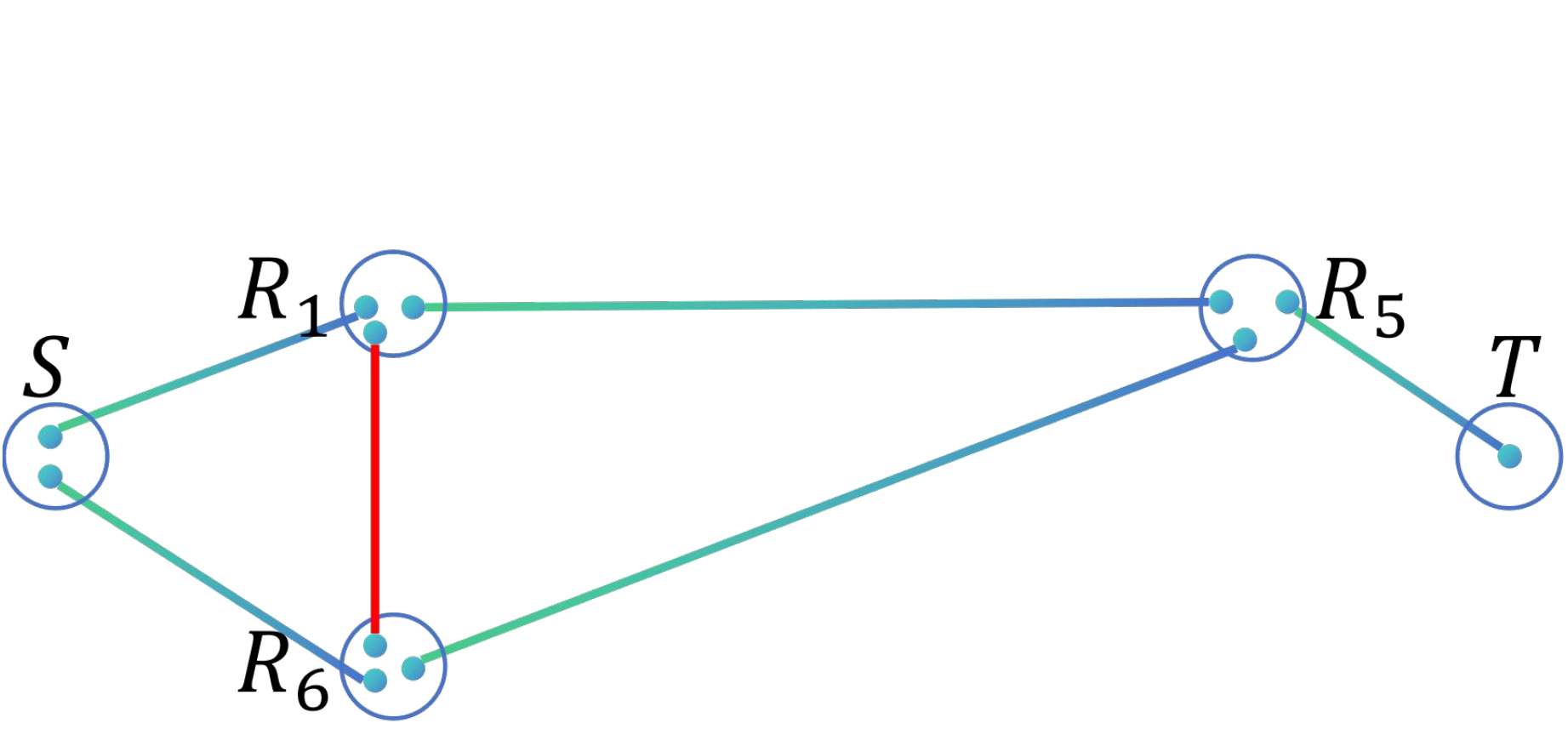}
    \label{fig-simplify_general}
    }
    \hspace{10mm}
    \subfigure[]{
    \includegraphics[width=1.59in]{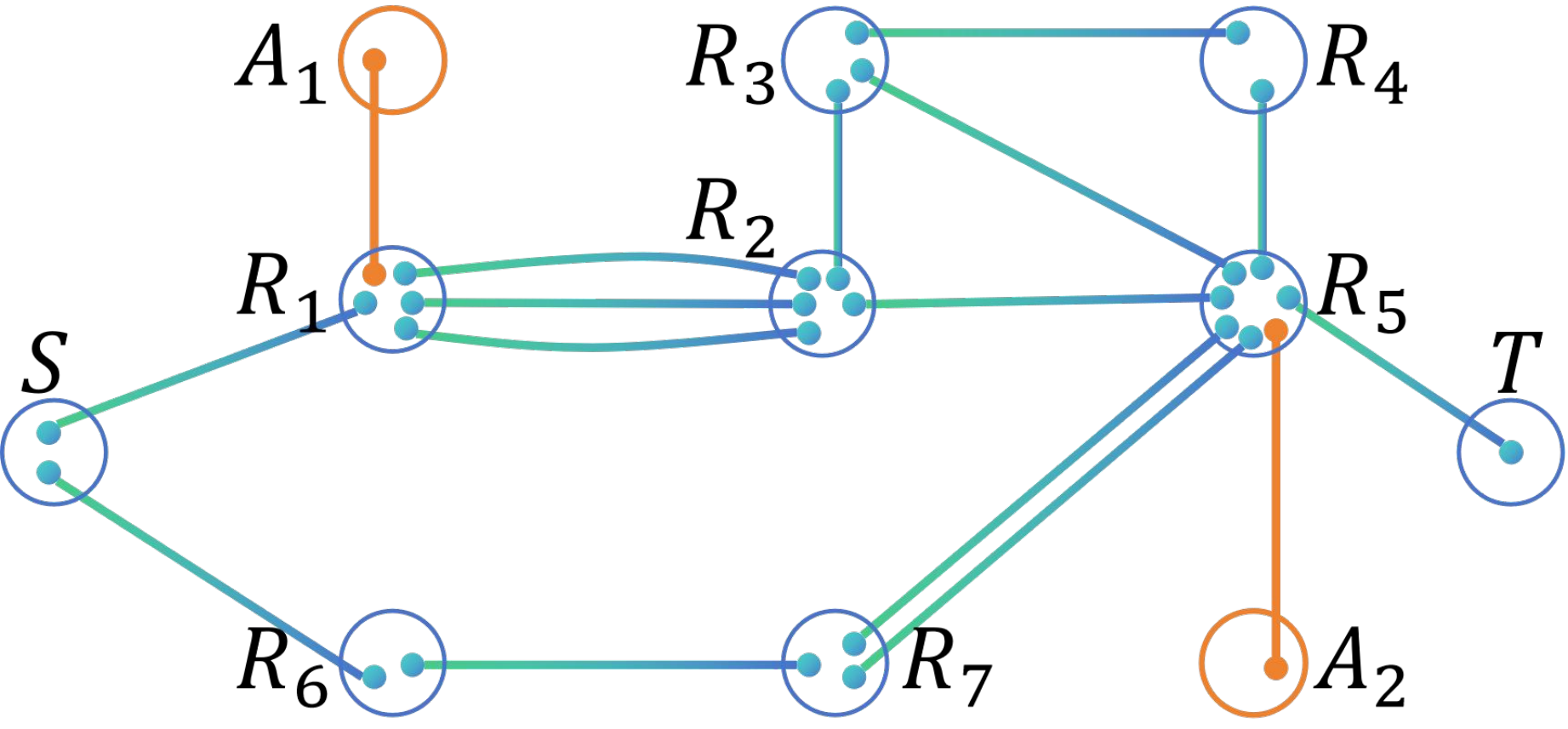}
    \label{fig-s-p}
    }
    \caption{\subref{fig-simplify_general}~Simplified form of network in Fig.~\ref{fig-general}.
    \subref{fig-s-p}~A equivalent network of Fig.~\ref{fig-s_p} with two terminals $S$ and $T$. Here, the sub-network consisting the orange nodes $A_1$ and $A_2$}
\end{figure}

\section{Operator representation of basic operations in entanglement percolation}

In this section, we describe the QN model and the three basic operations needed for entanglement percolation in the language of operator theory. These operations are entanglement swapping, entanglement concentration, and star transformation.

We first formalize a QN based on operator theory. Consider Fig.~\ref{fig-QN6nodes}, where different nodes represent spatially separated systems (for example, nodes $R_{1}$ and $R_{2}$), and each link connecting two nodes (e.g., the cyan line $R_{1}R_{2}$) represents shared quantum source states (such as $\rho_{R_{1}^{(1)}R_{2}^{(1)}}$).
When two nodes $S$ and $T$ share no state, no link exists between them. When the nodes of a QN are identified with subscripted labels, $R_j$, we denote the source state $\rho_{R_{j_1}^{(l_{j_1,j_2})}R_{j_2}^{(l_{j_2,j_1})}}$ shared between $R_{j_1}$ and $R_{j_2}$ ($j_1<j_2$) according to the following rules. Let $K_n$ be the number of subsystems in $R_{j_n}$, and label the nodes adjacent to $R_{j_n}$ in increasing order of their subscripts as $R_{j_{n1}},R_{j_{n2}},..., R_{j_{nK_n}}$ where $j_{n1}\leq j_{n2}\leq \cdots\leq j_{nK_n}$. When $R_{j_2}=R_{j_{1m}}$, we set $l_{j_1,j_2}=m$.
The number $l_{j_2,j_1}$ is obtained similarly.
The state of the entire network is denoted as $\rho_{\rm net}$, which is the tensor product of all source states.
An example of a QN is shown in Fig.~\ref{fig-QN6nodes}.

{In this paper, we focus on QNs with two-terminal series–parallel structures that are free of topological redundancy. For such QNs, we denote the terminals by $S$ and $T$, and label the internal nodes as $R_1, R_2, \dots$ [Figs.~\ref{fig-s}--\ref{fig-s_p}]. When applying the indexing rule $l$ to the source states, the terminals $S$ and $T$ are treated as $R_0$ and $R_{N+1}$, respectively.}

\begin{figure}[h]
    \centering
    \includegraphics[width=150pt]{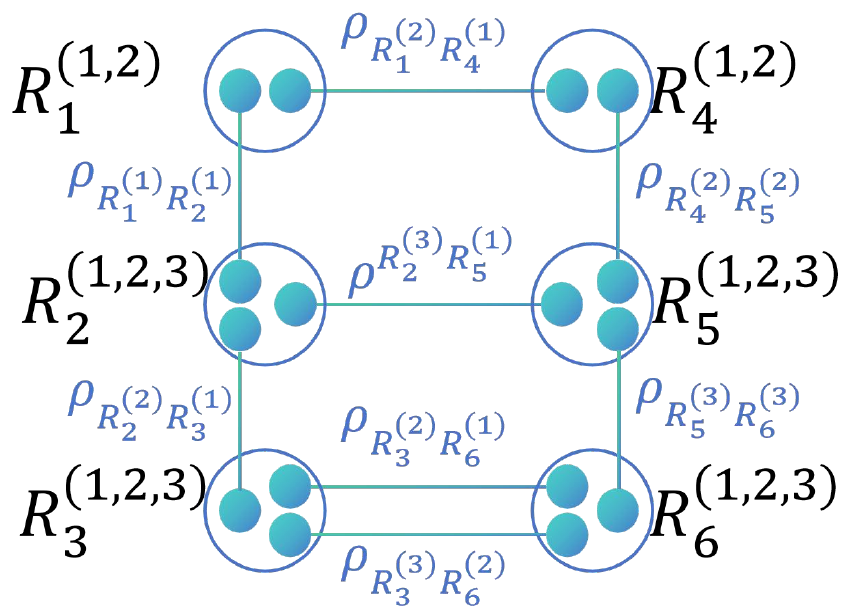}
    \caption{
    A quantum network with six nodes. The nodes are denoted as $R_{1},R_{2},\cdots,R_{6}$, respectively. The notation rules for states are described in the main text.}
    \label{fig-QN6nodes}
\end{figure}

\subsection{Entanglement swapping}\label{sec-swapping}

Entanglement swapping, as a class of local operations and classical communication (LOCC), is the most fundamental operation in entanglement percolation protocols. While prior research on deterministic entanglement swapping has provided concrete operational schemes~\cite{DET2023, P2002}, it has not yet abstracted a unified operator representation. Here, we {aim to provide} a more detailed operator expression for it acting on a QN.

\begin{figure}[h!]
    \centering
    \includegraphics[width=120pt]{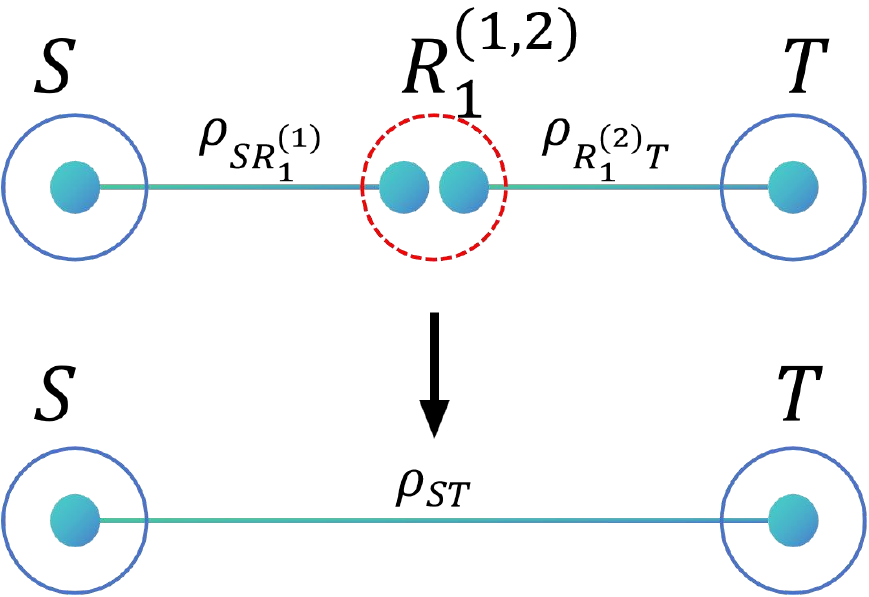}
    \caption{Entanglement swapping model.}
    \label{fig-seri2}
\end{figure}

\subsubsection{Operator  representation for entanglement swapping}

Let us start with the simplest network structure to gradually establish the operator representation of entanglement swapping.
Given a serial QN with only one intermediate node illustrated in Fig.~\ref{fig-seri2}. It consists of three nodes: two terminals $S$ and $T$, and a repeater node $R_1$ (or $R_1^{(1,2)}$).
Two states $\rho_{SR_1^{(1)}},\rho_{R_1^{(2)}T}$ are shared between $S$ and the first subsystem $R_1^{(1)}$ of $R$, and the second subsystem $R_1^{(2)}$ of $R$ and $T$, respectively.
We next introduce the LOCC-based entanglement swapping protocol that establishes long-distance entanglement between $S$ and $T$ via local measurements performed on the intermediate node $R_1$.

Given a measurement operation $\mathcal{M}_{R_1 }$ with locally measurement collection $\left\{M_{R_1 }^{(m)}\right\}_m$ performed on the state $\sigma_{R_1 }$ of bipartite system $R_1 $ defined by
\begin{eqnarray}\label{eq-DV_sum}
    \mathcal{M}_{R_1 }(\sigma_{R_1 })\equiv\sum_{m}M_{R_1 }^{(m)}\sigma_{R_1 }\left(M_{R_1 }^{(m)}\right)^{\dagger}.
\end{eqnarray}
The following LOCC protocol on the network state in Fig. ~\ref{fig-seri2},
\begin{eqnarray}\label{eq-Lambda}
  \mathcal{S}_{R_1 }:=I_{S}\otimes \mathcal{M}_{R_1 }\otimes I_{T}
\end{eqnarray}
yields a new entangled state
\begin{eqnarray*}
    \rho_{ST}={\rm Tr}_{R_1 }\left[ \mathcal{S}_{R_1 }\left(\rho_{SR_1^{(1)}}\otimes\rho_{R_1^{(2)}T}\right)\right]
\end{eqnarray*}
between $S$ and $T$.
The protocol results in the complete elimination of entanglement between $R$ and the remaining nodes $S$ and $T$---effectively rendering $R$ disconnected from the QN---this process is regarded as a {\it removal-node entanglement swapping}
\begin{eqnarray*}
    \Phi:\mathcal{D}_{\mathcal{S}}\otimes\mathcal{D}_{\mathcal{S}}\to\mathcal{D}_{\mathcal{S}}
\end{eqnarray*}
with
\begin{eqnarray}\label{eq-swapping_def}
    \Phi\left(\rho_{SR_1^{(1)}}\otimes\rho_{R_1^{(2)}T}\right)=\rho_{ST}.
\end{eqnarray}
where $\mathcal{D}_{\mathcal{S}}$ is a subset of states that can be chosen as source states in the QN. Thus, explicitly, the removal-node entanglement swapping is defined as
\begin{eqnarray}\label{eq-SwapPhi}
\Phi_{R_1 }^{ST}\left(\rho_{SR_1^{(1)}}\otimes\rho_{R_1^{(2)}T}\right)
    ={\rm Tr}_{R_1 }
    \left[ \mathcal{S}_{R_1 }\left(\rho_{SR_1^{(1)}}\otimes\rho_{R_1^{(2)}T}\right)\right],
\end{eqnarray}
where the subscript $R_1 $ indicates the node being measured. Note that, $\Phi_{R_1}^{ST}$ in Eq.~\eqref{eq-SwapPhi} is linear.
In certain protocols, local unitary operations $U_S^{(m)}$, $U_{R_1 }^{(m)}$ and $U_T^{(m)}$ depend on both the QN state $\rho_{SR_1^{(1)}}\otimes\rho_{R_1^{(2)}T}$ and the measurement outcome $m$, and are applied to the three nodes following the operation $M_{R_1 }^{m}$. In such implementations, the overall swapping operation takes the form with replacing $\mathcal{S}_{R_1 }$ in Eq.~\eqref{eq-Lambda} by
\begin{eqnarray}\label{eq-Lambda2}
    \mathcal{S}_{R_1 }'(\rho_{SR_1^{(1)}}\otimes\rho_{R_1^{(2)}T})=\sum\limits_{(m)}
    Q^{(m)}_{R_1 }
    \rho_{SR_1^{(1)}}\otimes\rho_{R_1^{(2)}T}
    \left(Q^{(m)}_{R_1 }\right)^{\dagger},
\end{eqnarray}
with $Q^{(m)}_{R_1 }=U_S^{(m)}\otimes U_{R_1 }^{(m)}M_{R_1 }^{(m)}\otimes U_T^{(m)}$ where the unitaries $U_S^{(m)}$ and $U_{R_1 }^{(m)}$ depend on the QN state $\rho_{SR_1^{(1)}}\otimes\rho_{R_1^{(2)}T}$.
In this case, the removal-node entanglement swapping becomes
\begin{eqnarray}\label{eq_swap2_unitary}
\Phi_{R_1 }^{ST}(\rho_{SR_1^{(1)}}\otimes\rho_{R_1^{(2)}T})
    ={\rm Tr}_{R_1 }
    \left[ \mathcal{S}_{R_1 }'\left(\rho_{SR_1^{(1)}}\otimes\rho_{R_1^{(2)}T}\right)\right].
\end{eqnarray}
Note that, as the unitary operators may be related to the source state, $\Phi_{R_1}^{ST}$ in Eq.~\eqref{eq_swap2_unitary} is in general nonlinear.
The map $\Phi_{R_1 }^{ST}$ given in Eq.~\eqref{eq-SwapPhi} is just the special case of the general form in Eq.~\eqref{eq_swap2_unitary} obtained by taking every unitary to be the identity.
The above description presents the operator representation of entanglement swapping on the simplest network structure. For more complex network structures with additional nodes, the protocol is successively applied to establish entanglement between distant nodes.

It is also mentioned that when $R_1 $ is a CV quantum system, the local operation $\mathcal{M}$ becomes an integral in continuous form as
\begin{eqnarray}\label{eq-cV_int}
    \mathcal{M}_{R_1 }(\sigma_{R_1 })\equiv\int_{\alpha}M_{R_1 }^{(\alpha)}\sigma_{R_1 }\left(M_{R_1 }^{(\alpha)}\right)^{\dagger}d\alpha.
\end{eqnarray}
In the subsequent discussion, we focus on the discrete form of the entanglement swapping operation, as shown in Eq.~\eqref{eq-DV_sum}, since the continuous integral form in Eq.~\eqref{eq-cV_int} can be generalized straightforwardly.

Now let us return to the issue of entanglement percolation. Recall that the purpose of the entanglement percolation protocol is to establish entanglement between two distant nodes. We use the entanglement measure to quantify the presence of entanglement. With the mathematical representation of the entanglement swapping protocol established above, we are able to formulate the entanglement percolation for the simple case of a series network.
Given an entanglement measure $E$ for bipartite states, in the series  QN of Fig.~\ref{fig-seri2}, to determine whether the entanglement between adjacent nodes can be {transmitted} to that between two distant nodes $S$ and $T$, we check if the entanglement measure between them is nonzero. The new established entanglement ${\rm Seri}_{\Phi,E}$ is defined by
\begin{eqnarray}\label{eq-Seri_2}
    {\rm Seri}_{\Phi,E}\left[\rho_{SR_1^{(1)}},\rho_{R_1^{(2)}T}\right]:=E\left[\Phi_{R_1 }^{ST}\left(\rho_{SR_1^{(1)}}\otimes\rho_{R_1^{(2)}T}\right)\right].
\end{eqnarray}
We call ${\rm Seri}_{\Phi,E}$ the \emph{series rule} for the swapping $\Phi$ on state set $\mathcal{D}_{\mathcal{S}}$ under $E$.
In the next subsection, we focus on the properties and behaviors of entanglement swapping in percolation.

\subsubsection{Properties and Examples}\label{sec-swapping_examples}

{For the entanglement swapping mapping $\Phi: \mathcal{D}_S \otimes \mathcal{D}_S \to \mathcal{D}_S$, we consider the question under what conditions on entanglement measure $E$ and two distinct pairs of quantum states ${\rho_{SR_1^{(1)}}, \rho_{R_1^{(2)}T}}$ and ${\rho_{SR_1^{(1)}}', \rho_{R_1^{(2)}T}'}$ in $\mathcal{D}_S$ that ensure the equality
\begin{eqnarray}\label{eq-preserving}
         {\rm Seri}_{\Phi,E}\left[\rho_{SR_1^{(1)}},\rho_{R_1^{(2)}T}\right]
         ={\rm Seri}_{\Phi,E}\left[\rho_{SR_1^{(1)}}',\rho_{R_1^{(2)}T}'\right]
    \end{eqnarray}
holds. This question is important for our purpose. In general, Eq.~\eqref{eq-preserving} does not hold true. However, by imposing certain constraints, it is possible to satisfy this equality under specific circumstances.}

{We find that for certain specific entanglement swapping $\Phi$ and entanglement measure $E$, if the two pairs of states in $\mathcal{D}_S$ satisfy
\begin{eqnarray}\label{eq-preserving1}
E\left(\rho_{SR_1^{(1)}}\right)=E\left(\rho_{SR_1^{(1)}}'\right),\quad E\left(\rho_{R_1^{(2)}T}\right)=E\left(\rho_{R_1^{(2)}T}'\right),
\end{eqnarray}
or
\begin{eqnarray}\label{eq-preserving2}
E\left(\rho_{SR_1^{(1)}}\right)=E\left(\rho_{R_1^{(2)}T}'\right),\quad E\left(\rho_{R_1^{(2)}T}\right)=E\left(\rho_{SR_1^{(1)}}'\right),
\end{eqnarray}
Eq.~\eqref{eq-preserving} holds. In {the} following, we give several examples that satisfy this preservation property. }

\textbf{Example 4.1} (\textbf{DV-based QNs of two-qudit states}) Let $E$ be the $G$-concurrence $C_G$ for DV-based states in Eq. (\ref{G-lebal}) and let $\mathcal{D}_{\mathcal{S}}$ be the set of all pure two-qudit (the dimension of each system is $d$) entangled states.
In a QN with its sources from  $\mathcal{D}_{\mathcal{S}}$.
Ref.~\cite{DET2023} proposes an entanglement swapping protocol.
Denoted by $|e_{1}\rangle$, $|e_{2}\rangle$,...,$|e_{d}\rangle$ the computational basis of a $d$-dimensional quantum system.
Given $d^2$ two-qudit states as
\begin{eqnarray*}
    |M^{(m)}\rangle=d^{-1}\sum_{\mu,\nu}M_{m,\mu,\nu}|e_{\mu}\rangle\otimes |e_{\nu}\rangle,\ m=1,2,\ldots,d^2,
\end{eqnarray*}
where $M_{m,\mu,\nu}=\exp\left[{-m(d\mu+\nu)2\pi i/d^2-2\pi i \mu\nu/d}\right]$.
Then the collection $\{{M}_{R_1 }^{(m)}\}_m$ of projections
\begin{eqnarray}\label{eq-measure_qudit}
    {M}^{(m)}=|M^{(m)}\rangle \langle M^{(m)}|
\end{eqnarray}
constitutes a projective measurement.

In the simplest 1D network shown in Fig.~\ref{fig-seri2}, we have the following observation of
a symmetric series rule in the form of product:
 \begin{eqnarray}\label{eq-seri_CG}
    {\rm Seri}_{\Phi,C_G}\left[\rho_{SR_1^{(1)}}, \rho_{R_1^{(2)}T}\right]
    =C_G\left(\rho_{SR_1^{(1)}}\right) C_G\left(\rho_{R_1^{(2)}T}\right).
\end{eqnarray}

\begin{proof}
Consider the scenario that $S$ and $R_1^{(1)}$ share a state $\rho_{SR_1^{(1)}}=|a\rangle\langle a|\in\mathcal{D}_{\mathcal{S}}$ with
\begin{eqnarray*}
    |a\rangle=\sum_{j=1}^{d}\sqrt{a_{j}}|e_{j}\rangle_{S}\otimes |e_{j}\rangle_{R_1^{(1)}}
\end{eqnarray*}
and $R_1^{(2)}$ and $T$ share another bipartite pure state $\rho_{R_1^{(2)}T}=|b\rangle\langle b|\in\mathcal{D}_{\mathcal{S}}$ with
\begin{eqnarray*}
    |b\rangle=\sum_{j=1}^{d}\sqrt{b_{j}}|e_{j}\rangle_{R_1^{(2)}}\otimes |e_{j}\rangle_{T},
\end{eqnarray*}
where $a_j,b_j\in(0,1)$, $\sum_{j=1}^{d}a_j=\sum_{j=1}^{d}b_j=1$ and the coefficient vectors $\vec{a}=(a_1,a_2,\ldots,a_d)$ and $\vec{b}=(b_1,b_2,\ldots,b_d)$ of the two states are both arranged in descending order.
Performing the entanglement swapping protocol with projective measurement $\{{M}_m\}_m$ on $R_1^{(1)}$ and $R_1^{(2)}$, then we can obtain the output pure state on nodes $S$ and $T$ for outcome $m$ is
\begin{eqnarray*}
    |\eta^{(m)}\rangle_{S T}=d^{-1} \sum_{k=1,s=1}^{d} \sqrt{a_{k} b_{s}} \exp\left[-(mdk+ms+dks) 2 \pi i/ d^2\right]|e_k\rangle_{S} \otimes |e_s\rangle_{T}
\end{eqnarray*}
with Schmidt coefficients $f_{sw}(\vec{a},\vec{b})$.
Here, $f_{sw}(\vec{x},\vec{y})$ is a swapping function proposed in Ref.~\cite{DET2023} of vectors $\vec{x}$ and $\vec{y}$,
\begin{eqnarray*}
    f_{sw}(\vec{x},\vec{y})=d\times \sigma^2\left(\rm{diag}(\vec{x})^{1/2} V \rm{diag}(\vec{y})^{1/2}\right)
\end{eqnarray*}
where $\sigma^2{(A)}$ denotes the entry-wise square of the singular values $\sigma{(A)}$ of operator $A$ arranged in descending order, and the matrix $V$ is constant and unitary with elements $V_{\mu,\nu}=d^{-1/2}\exp(-2\pi i \mu\nu/d)$, $\mu,\nu=1,2,\ldots,d$.
{Let $\rho_{S}^{(m)}=\Tr_{T}\left(|\eta^{(m)}\rangle\langle\eta^{(m)}|_{S T}\right)$, $\rho_{T}^{(m)}=\Tr_{S}\left(|\eta^{(m)}\rangle\langle\eta^{(m)}|_{S T}\right)$ and $\vec{\lambda}_{\rm{out}}=\left(\lambda_1,\lambda_2,\ldots,\lambda_d\right)^{\rm{T}}=f_{sw}(\vec{a},\vec{b})$.
Correspondingly, we denote eigenvectors of $\rho_{S}$ and $\rho_{T}$ by $|\alpha_1^{(m)}\rangle$, $|\alpha_2^{(m)}\rangle$, $\ldots$, $|\alpha_{d}^{(m)}\rangle$ and
$|\beta_1^{(m)}\rangle$, $|\beta_2^{(m)}\rangle$, $\ldots$, $|\beta_{d}^{(m)}\rangle$, respectively, where
$|\alpha_k^{(m)}\rangle=\sum_{s=1}^{d}\alpha_{k,s}^{(m)}|e_s\rangle$ and
$|\beta_k^{(m)}\rangle=\sum_{s=1}^{d}\beta_{k,s}^{(m)}|e_s\rangle$.
We then establish two local unitary operators
\begin{eqnarray*}\label{eq-local_US_qudit}
    U_{S}^{(m)}=\sum_{k=1}^{d}\sum_{s=1}^{d}\alpha_{k,s}^{(m)}|e_k\rangle\langle e_k|
\end{eqnarray*}
and
\begin{eqnarray*}\label{eq-local_UT_qudit}
    U_{T}^{(m)}=\sum_{k=1}^{d}\sum_{s=0}^{d-1}\beta_{k,s}^{(m)}|e_k\rangle\langle e_k|
\end{eqnarray*}
for systems of $S$ and $T$, respectively.}
Therefore, for each output state $|\eta^{(m)}\rangle_{S T}$ for outcome $m$, there exist a local rotation
\begin{eqnarray}\label{eq-local_UST_qudit}
    U_{ST}^{(m)}=U_{S}^{(m)}\otimes U_{T}^{(m)},
\end{eqnarray}
which is a unitary, such that $U_{ST}^{(m)}|\eta^{(m)}\rangle_{S T}=|\eta^{\rm{out}}\rangle_{S T}$
with
\begin{eqnarray*}
    |\eta^{\rm{out}}\rangle_{S T}
    =\sum\limits_{j=0}^{d-1}\sqrt{\lambda_{j}}|e_j\rangle_{S}\otimes|e_j\rangle_{T}
\end{eqnarray*}
which is independent with outcome $m$.
In a word, executing the entanglement swapping map $\Phi$ in Eq.~\eqref{eq_swap2_unitary} consisting of the measurement defined in Eq.~\eqref{eq-measure_qudit} and local unitary transformations defined in Eq.~\eqref{eq-local_UST_qudit}, one can deterministically obtain
\begin{eqnarray*}
    \Phi_{R_1 }^{ST}\left(\rho_{SR_1^{(1)}}\otimes\rho_{R_1^{(2)}T}\right)=|\eta^{\rm{out}}\rangle\langle\eta^{\rm{out}}|_{S T}
\end{eqnarray*}
from the QN state $\rho_{SR_1^{(1)}}\otimes\rho_{R_1^{(2)}T}$.
Then under the characterization of the $G$-concurrence $C_G$, we obtain the equation
\begin{eqnarray*}
    {\rm Seri}_{\Phi,C_G}\left(\rho_{SR_1^{(1)}}, \rho_{R_1^{(2)}T}\right)
    =C_G\left(\rho_{SR_1^{(1)}}\right) C_G\left(\rho_{R_1^{(2)}T}\right).
\end{eqnarray*}
\end{proof}

It is evident that the product-form series rule given in Eq.~\eqref{eq-seri_CG} satisfies the preservation property---that is, the validity of either Eq.~\eqref{eq-preserving1} or Eq.~\eqref{eq-preserving2} implies the validity of Eq.~\eqref{eq-preserving}.
Furthermore, product-form series rules also exist for two-qubit states and TMSVSs, suggesting that these systems similarly possess the aforementioned preservation property. The details are presented below.

In qubit systems ($d=2$), we have a more detailed calculation. Let $\{|ks\rangle\}_{k,s=0}^{1}$ be the computational basis. For the pure states  $|a\rangle=|b\rangle=\sqrt{\lambda}|00\rangle+\sqrt{1-\lambda}|11\rangle$,
we obtain
\begin{eqnarray*}
    |M^{(m)}\rangle=\frac{1}{2}\left(|00\rangle+e^{-\frac{m\pi i}{2}}|01\rangle+e^{-m\pi i}|10\rangle-e^{-\frac{3m\pi i}{2}}|11\rangle\right)
\end{eqnarray*}
and thereby the output state between $S$ and $T$ of outcome $m$ is
\begin{eqnarray*}
    |\eta^{(m)}\rangle_{ST}=\lambda|00\rangle+e^{\frac{m\pi i}{2}}|01\rangle+e^{m\pi i}\sqrt{\lambda(1-\lambda)}|10\rangle-e^{\frac{3m\pi i}{2}}(1-\lambda)|11\rangle
\end{eqnarray*}
 with probability of the outcome $m$ $p^{(m)}=1/4$.
The Schmidt coefficients of $|\eta^{(m)}\rangle_{ST}$ are
\begin{eqnarray*}
    \lambda_1=\frac{1+\sqrt{1-16\lambda^2(1-\lambda)^2}}{2},\lambda_2=\frac{1-\sqrt{1-16\lambda^2(1-\lambda)^2}}{2}.
\end{eqnarray*}
Then the eigenstates of $\rho_{S}^{(m)}$ are
\begin{eqnarray*}
    |\alpha_1\rangle
    =\frac{a}{\sqrt{a^2+(\lambda-\lambda_1)^2}}|0\rangle
    -\frac{\lambda-\lambda_1}{e^{-m\pi i}\sqrt{a^2+(\lambda-\lambda_1)^2}}|1\rangle
\end{eqnarray*}
and
\begin{eqnarray*}
    |\alpha_2\rangle
    =\frac{a}{\sqrt{a^2+(\lambda-\lambda_2)^2}}|0\rangle
    -\frac{\lambda-\lambda_2}{e^{-m\pi i}\sqrt{a^2+(\lambda-\lambda_2)^2}}|1\rangle,
\end{eqnarray*}
where $a=(2\lambda-1)\sqrt{\lambda(1-\lambda)}$.
Similarly, the eigenstates of $\rho_{T}^{(m)}$ are
\begin{eqnarray*}
    |\beta_1\rangle
    =\frac{a}{\sqrt{a^2+(\lambda-\lambda_1)^2}}|0\rangle
    -\frac{\lambda-\lambda_1}{e^{-\frac{m\pi i}{2}}\sqrt{a^2+(\lambda-\lambda_1)^2}}|1\rangle
\end{eqnarray*}
and
\begin{eqnarray*}
    |\beta_2\rangle
    =\frac{a}{\sqrt{a^2+(\lambda-\lambda_2)^2}}|0\rangle
    -\frac{\lambda-\lambda_2}{e^{-\frac{m\pi i}{2}}\sqrt{a^2+(\lambda-\lambda_2)^2}}|1\rangle.
\end{eqnarray*}
Then two unitary operators in Eqs.~\eqref{eq-local_UST_qudit} in qubit systems are
\begin{eqnarray*}
    U_{S}^{(m)}=&\frac{a}{\sqrt{a^2+(\lambda-\lambda_1)^2}}|0\rangle\langle 0|
    -\frac{\lambda-\lambda_2}{e^{-m\pi i}\sqrt{a^2+(\lambda-\lambda_1)^2}}|0\rangle\langle 1|&\\
    +&\frac{a}{\sqrt{a^2+(\lambda-\lambda_2)^2}}|1\rangle\langle 0|
    -\frac{\lambda-\lambda_2}{e^{-m\pi i}\sqrt{a^2+(\lambda-\lambda_2)^2}}|1\rangle\langle 1|&.
\end{eqnarray*}
and
\begin{eqnarray*}
    U_{T}^{(m)}=&\frac{a}{\sqrt{a^2+(\lambda-\lambda_1)^2}}|0\rangle\langle 0|
    -\frac{\lambda-\lambda_1}{e^{-\frac{m\pi i}{2}}\sqrt{a^2+(\lambda-\lambda_1)^2}}|0\rangle\langle 1|&\\
    +&\frac{a}{\sqrt{a^2+(\lambda-\lambda_2)^2}}|1\rangle\langle 0|
    -\frac{\lambda-\lambda_1}{e^{-\frac{m\pi i}{2}}\sqrt{a^2+(\lambda-\lambda_2)^2}}|1\rangle\langle 1|&.
\end{eqnarray*}
Finally, we have that
\begin{eqnarray}\label{eq-phi_2qubit}
    \Phi_{R_1 }^{ST}\left(\rho_{SR_1^{(1)}}\otimes\rho_{R_1^{(2)}T}\right)=|\eta^{\rm{out}}\rangle\langle\eta^{\rm{out}}|=|\psi_{\lambda}\rangle\langle\psi_{\lambda}|
\end{eqnarray}
with series rule
\begin{eqnarray}\label{eq-seri_c}
    {\rm Seri}_{\Phi,c}\left[\rho_{SR_1^{(1)}}, \rho_{R_1^{(2)}T}\right]
    =c(|\psi_{\lambda}\rangle\langle\psi_{\lambda}|)=c\left(\rho_{SR_1^{(1)}}\right)c\left(\rho_{R_1^{(2)}T}\right)
\end{eqnarray}
for the concurrence $c$.

\textbf{Example 4.2} (\textbf{CV-based QNs of TMSVSs}) Here we consider the QN with CV states as its sources. Let $\mathcal{D}_{\mathcal{S}}$ be a set of all TMSVSs [Eq.~\eqref{eq-TMSVS}].
Then two source states can be written as $\rho_{SR_1^{(1)}}=|\psi^{r_1}\rangle\langle\psi^{r_1}|$ and $\rho_{R_1^{(2)}T}=|\psi^{r_2}\rangle\langle\psi^{r_2}|$.
For this situation, the operator $\Phi_{R_1 }^{ST}$ is a CV-based entanglement swapping proposed in {Ref.~{\cite{P2002}}}, implemented in two steps:
(1) CV Bell measurement at node $R$: Optically, the two modes $R^{(1)}$ and $R^{(2)}$ at node $R$ are combined on a 50:50 beam splitter. Subsequently, the quadratures
$\hat{x}_u=(\hat{x}_1-\hat{x}_2)/\sqrt{2}$, $\hat{p}_v=(\hat{p}_1+\hat{p}_2)/\sqrt{2}$ are detected, where $\hat{x}_k$ and $\hat{p}_k$ are position operator and momentum operator of $R^{(k)}$, $k=1,2$; (2) Local displacements at node $T$: Local displacements are applied to $T$ according to:
\begin{eqnarray*}
    \hat{x}_T\to \hat{x}_T'=\hat{x}_T+g_{\rm swap}\sqrt{2}\hat{x}_u,\quad\hat{p}_T\to \hat{p}_T'=\hat{p}_T+g_{\rm swap}\sqrt{2}\hat{p}_v
\end{eqnarray*}
with the gain
\begin{eqnarray*}
    g_{\mathrm{swap}} = \dfrac{\sinh 2r_1 + \sinh 2r_2}{\cosh 2r_1 + \cosh 2r_2 -1}.
\end{eqnarray*}
Then we obtain
\begin{eqnarray*}
    \Phi\left(\rho_{SR_1^{(1)}}\otimes\rho_{R_1^{(2)}T}\right)=|\psi^r\rangle\langle\psi^r|_{ST}
\end{eqnarray*}
where $r$ is given by $\tanh r=\tanh r_1 \tanh r_2$, yielding a series rule for the ratio negativity $\chi_{\mathcal{N}}$ [Eq.~\eqref{eq-ratio_negativity}]~\cite{Ratio_negativity2024}:
\begin{eqnarray}\label{eq-seri_chi}
    {\rm Seri}_{\Phi,\chi_{\mathcal{N}}}\left[\rho_{SR_1^{(1)}}, \rho_{R_1^{(2)}T}\right]
    =\chi_{\mathcal{N}}\left(\rho_{SR_1^{(1)}}\right) \chi_{\mathcal{N}}\left(\rho_{R_1^{(2)}T}\right)
\end{eqnarray}
for all $\rho_{SR_1^{(1)}}$, $\rho_{R_1^{(2)}T}\in\mathcal{D}_{\mathcal{S}}$.

{
For general cases, the requirements of symmetric series rule in the form of product shown in Example 4.1 and 4.2 may be too stringent, so we propose a more general order-preserving property: the inequality
\begin{eqnarray*}\label{t-eq-phi_order_preserving}
    {\rm Seri}_{\Phi,E}\left[\rho_{SR_1^{(1)}},\rho_{R_1^{(2)}T}\right]\geq{\rm Seri}_{\Phi,E}\left[\rho_{SR_1^{(1)}}',\rho_{R_1^{(2)}T}'\right]
\end{eqnarray*}
holds for all states pairs in $\mathcal{D}_S$ which satisfy
\begin{eqnarray*}\label{c-eq-phi_order_preserving}
    E\left(\rho_{SR_1^{(1)}}\right)\geq E\left(\rho_{SR_1^{(1)}}'\right),\ E\left(\rho_{R_1^{(2)}T}\right)\geq E\left(\rho_{R_1^{(2)}T}'\right),
\end{eqnarray*}
or
\begin{eqnarray*}\label{c-eq-phi_order_preserving}
    E\left(\rho_{SR_1^{(1)}}\right)\geq E\left(\rho_{R_1^{(2)}T}'\right),\ E\left(\rho_{R_1^{(2)}T}\right)\geq E\left(\rho_{SR_1^{(1)}}'\right).
\end{eqnarray*}
It is easily checked that, if an entanglement swapping protocol $\Phi$ and an entanglement measure $E$ satisfy this order-preserving property, then Eq.~\eqref{eq-preserving} holds whenever two pairs of states in $\mathcal{D}_S$ satisfy Eq.~\eqref{eq-preserving1} or Eq.~\eqref{eq-preserving2}.
\hfill $\square$
}

As evident from the three expressions shown in Eqs.~\eqref{eq-seri_CG},~\eqref{eq-seri_c}~and~\eqref{eq-seri_chi}, all of the series rules in the three case manifest a multiplicative form, despite being quantified by distinct entanglement measures. As we shall see later, this perfect product form is crucial for establishing the entanglement percolation protocol. Unfortunately, apart from the three cases mentioned above, we still do not know whether this serial-product rule holds for other scenarios—even in networks with pure-state sources. Exploring this problem remains highly challenging (\cite{Ratio_negativity2024}). Here we derive and present a more general conclusion.

\begin{theorem}
If $E$ is an entanglement monotone, then
    \begin{eqnarray}\label{ineq-Seri}
        {\rm Seri}_{\Phi,E}\left[\rho_{SR_1^{(1)}},\rho_{R_1^{(2)}T}\right]
    \leq \min\left\{E\left(\rho_{SR_1^{(1)}}\right), E\left(\rho_{R_1^{(2)}T}\right)\right\}
    \end{eqnarray}
    for all pure bipartite states $\rho_{SR_1^{(1)}}$ and $ \rho_{R_1^{(2)}T}$.
\end{theorem}
\begin{proof}
    Let $\rho_{SR_1^{(1)}}=|\phi\rangle\langle\phi|,\rho_{R_1^{(2)}T}=|\psi\rangle\langle\psi|$.
    Denote by $|\phi\rangle=\sum_{\alpha}\lambda_{\alpha}|\xi_\alpha\rangle_{S}|\eta_\alpha\rangle_{R_1^{(1)}}$ the Schmidt decomposition~\cite{Nielsen2010} of $|\phi\rangle$.
    Then $|\phi\rangle|\psi\rangle$ has the Schmidt decomposition $|\phi\rangle|\psi\rangle=\sum_{\alpha}\lambda_{\alpha}|\xi_\alpha\rangle_{S}|\eta_\alpha'\rangle_{R_1 T}$
    between two parties $S$ and $R_1 T$, where $|\eta_\alpha'\rangle=|\eta_\alpha\rangle|\psi\rangle$.
    Since every entanglement measure of a bipartite pure state depends solely on its Schmidt coefficients, any two such states with identical Schmidt coefficients will have the same value under every entanglement measure.
    Hence, the entanglement of $$\rho_{SR_1^{(1)}}\otimes\rho_{R_1^{(2)}T}=|\phi\rangle|\psi\rangle\langle\phi|\langle\psi|$$
    between parts $S$ and $R_1 T$ is equal to the entanglement of $\rho_{SR_1^{(1)}}$,
    \begin{eqnarray}\label{eq-1}
        E\left(\rho_{SR_1^{(1)}}\right)
        =E^{S|R_1 T}\left(\rho_{SR_1^{(1)}}\otimes\rho_{R_1^{(2)}T}\right).
    \end{eqnarray}
    Since $E$ is an entanglement monotone, one then obtains that
    \begin{eqnarray}\label{eq-2}
        E^{S|R_1 T}\left(\rho_{SR_1^{(1)}}\otimes\rho_{R_1^{(2)}T}\right)
        &\geq&
        E^{S|R_1 T}\left[\mathcal{S}_{R_1}'\left(\rho_{SR_1^{(1)}}\otimes\rho_{R_1^{(2)}T}\right)\right]\nonumber\\
        &\geq&
        E^{S|T}\left[{\rm Tr}_{R_1 }
    \left[ \mathcal{S}_{R_1 }'\left(\rho_{SR_1^{(1)}}\otimes\rho_{R_1^{(2)}T}\right)\right]\right]\nonumber\\
        &=&E\left[\Phi_{R_1 }^{ST}\left(\rho_{SR_1^{(1)}}\otimes\rho_{R_1^{(2)}T}\right)\right]\nonumber\\
        &=&{\rm Seri}_{\Phi,E}\left[\rho_{SR_1^{(1)}},\rho_{R_1^{(2)}T}\right]
    \end{eqnarray}
    for the LOCC maps $\mathcal{S}_{R_1}'$ [Eq.~\eqref{eq-Lambda2}] and ${\rm Tr}_{R_1 }$.
    From Ineqs.~\eqref{eq-1}~and~\eqref{eq-2}, we thereby have
    \begin{eqnarray*}
        E\left(\rho_{SR_1^{(1)}}\right)
        \geq
        {\rm Seri}_{\Phi,E}\left[\rho_{SR_1^{(1)}},\rho_{R_1^{(2)}T}\right].
    \end{eqnarray*}
    Similarly, we obtain that the inequality
    \begin{eqnarray*}
        E\left(\rho_{R_1^{(2)}T}\right)
        \geq
        {\rm Seri}_{\Phi,E}\left[\rho_{SR_1^{(1)}},\rho_{R_1^{(2)}T}\right]
    \end{eqnarray*}
    holds true.
    Then Ineq.~\eqref{ineq-Seri} holds true for $E$.
\end{proof}

Since the $G$-concurrence $C_G$ and the ratio negativity $\chi_{\mathcal{N}}$ are entanglement monotones~\cite{G_concurrence2005,Ratio_negativity2024}, the series-rule inequality~\eqref{ineq-Seri} holds in both the case $E=C_G$ with $\mathcal{D}_{\mathcal{S}}$ being a set of pure two-qudit states, and the case $E=\chi_{\mathcal{N}}$ with $\mathcal{D}_{\mathcal{S}}$ being a set of TMSVSs for CV systems.

\subsection{Entanglement concentration}\label{sec-concentration}

\begin{figure}[]
    \centering
    \includegraphics[width=280pt]{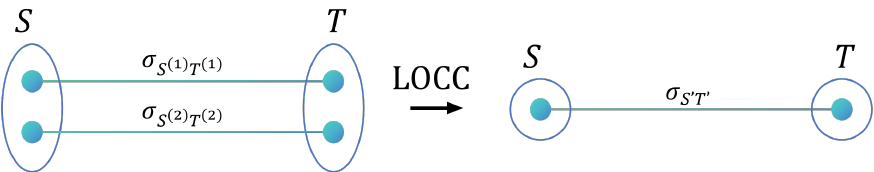}
    \caption{The simplest parallel quantum network.}
    \label{fig-parallel2}
\end{figure}

In physical experiments, when two nodes share weakly entangled states with respect to an entanglement measure $E$, {they can converse these states} into a single, more entangled state using an entanglement concentration protocol {based on LOCC and local selective operations~\cite{concentration1996_1,concentration1996_2}}.
{In quantum information theory, existing models of entanglement concentration can be formulated in the language of operator theory~\cite{concentration1996_1,concentration1996_2,X_state2021,Percolation_ConPT2021,NegPT}. Specifically, for a given set of states $\mathcal{D}_{\mathcal{S}}$, suppose that two distant parties $S$ and $T$ share two entangled states $\sigma_{S^{(1)}T^{(1)}}$ and $\sigma_{S^{(2)}T^{(2)}}$, with $\sigma_{S^{(1)}T^{(1)}}, \sigma_{S^{(2)}T^{(2)}} \in \mathcal{D}_{\mathcal{C}}$. The entanglement concentration process can then be formalized as a map} $\mathcal{C}_{ST}:\mathcal{D}_{\mathcal{C}}\otimes \mathcal{D}_{\mathcal{C}}\to \mathcal{D}_{\mathcal{C}}\otimes \mathcal{D}_{\mathcal{C}}$
{for which there exist {two states} $\rho_{S^{(1)}T^{(1)}},\rho_{S^{(2)}T^{(2)}}\in \mathcal{D}_{\mathcal{C}}$ such that}
\begin{eqnarray}\label{eq-symmetry_C}
    {\mathcal{C}_{ST}\left(\sigma_{S^{(1)}T^{(1)}}\otimes\sigma_{S^{(2)}T^{(2)}}\right)
    =\mathcal{C}_{ST}\left(\sigma_{S^{(2)}T^{(2)}}\otimes\sigma_{S^{(1)}T^{(1)}}\right)=\rho_{S^{(1)}T^{(1)}}\otimes\rho_{S^{(2)}T^{(2)}},}
\end{eqnarray}
{where $S^{(1)}$ and $S^{(2)}$ denote two subsystems of $S$, and $T^{(1)}$ and $T^{(2)}$ denote two subsystems obtained from $T$}. In this process, $\rho_{S^{(1)}T^{(1)}}$ is an entangled state shared by $S^{(1)}$ and $T^{(1)}$ whose degree of entanglement is higher than that of either initial state, while $\rho_{S^{(2)}T^{(2)}}$ is another bipartite state shared by $S^{(2)}$ and $T^{(2)}$ whose entanglement is lower than that of both initial states.
This reduces a \emph{normalized entanglement concentration map}
\begin{eqnarray}\label{eq-concen2}
\Psi_{ST}\left(\sigma_{S^{(1)}T^{(1)}}\otimes\sigma_{S^{(2)}T^{(2)}}\right):
=\Tr_{S^{(2)}T^{(2)}}\left[\mathcal{C}_{ST}\left(\sigma_{S^{(1)}T^{(1)}}\otimes\sigma_{S^{(2)}T^{(2)}}\right)\right],
\end{eqnarray}
that is, $\Psi_{ST}\left(\sigma_{S^{(1)}T^{(1)}}\otimes\sigma_{S^{(2)}T^{(2)}}\right)=\rho_{S^{(1)}T^{(1)}}$.
{More generally, when considering an entanglement concentration protocol $\mathcal{C}_{ST}:\mathcal{D}_{\mathcal{C}}\otimes \mathcal{D}_{\mathcal{C}}\to \mathcal{D}_{\mathcal{C}}\otimes \mathcal{D}_{\mathcal{C}}$, the resulting output state is not a product state as in Eq.~\eqref{eq-symmetry_C}. For such cases, we also define $\Psi_{ST}$ as in Eq.~\eqref{eq-concen2}. The map $\Psi_{ST}$ is then termed a normalized entanglement concentration map if it satisfies the following condition: the entanglement of $\Psi_{ST}\big(\sigma_{S^{(1)}T^{(1)}}\otimes\sigma_{S^{(2)}T^{(2)}}\big)$ is greater than that of either initial state,
$$E\left[\Psi_{ST}\big(\sigma_{S^{(1)}T^{(1)}}\otimes\sigma_{S^{(2)}T^{(2)}}\big)\right]\geq \max\left\{E\left[\sigma_{S^{(1)}T^{(1)}}\right],E\left[\sigma_{S^{(2)}T^{(2)}}\right]\right\}.$$}

Let $E$ be an entanglement measure. A {\it parallel rule of entanglement concentration} ${\rm Para}_{\Psi,E}$ is defined as
\begin{eqnarray*}
    {\rm Para}_{\Psi,E}\left[\sigma_{S^{(1)}T^{(1)}},\sigma_{S^{(2)}T^{(2)}}\right]:=E\left[\Psi_{ST}\left(\sigma_{S^{(1)}T^{(1)}}\otimes\sigma_{S^{(2)}T^{(2)}}\right)\right].
\end{eqnarray*}
A common calculation from Eq.~\eqref{eq-symmetry_C} implies a very key symmetric property of the entanglement concentration map as follows,
\begin{eqnarray*}
    \Psi_{ST}(\sigma_{S^{(1)}T^{(1)}}\otimes\sigma_{S^{(2)}T^{(2)}})=\Psi_{ST}(\sigma_{S^{(2)}T^{(2)}}\otimes\sigma_{S^{(1)}T^{(1)}}).
\end{eqnarray*}
Next, analogous to entanglement swapping, we investigate when the equality
\begin{eqnarray*}
        {\rm Para}_{\Psi,E}\left[\sigma_{S^{(1)}T^{(1)}},\sigma_{S^{(2)}T^{(2)}}\right]=
        {\rm Para}_{\Psi,E}\left[\sigma_{S^{(1)}T^{(1)}}',\sigma_{S^{(2)}T^{(2)}}'\right]
\end{eqnarray*}
holds for some entanglement measure $E$, where $\sigma_{S^{(1)}T^{(1)}}$, $\sigma_{S^{(1)}T^{(1)}}'$, $\sigma_{S^{(2)}T^{(2)}}$, and $\sigma_{S^{(2)}T^{(2)}}' \in \mathcal{D}_{\mathcal{C}}$.

{\begin{definition}
(\textbf{$E$-preserving and $E$-order-preserving concentration})
An entanglement concentration map $\Psi$ is said to be \emph{$E$-preserving} on $\mathcal{D}_{\mathcal{C}}$ if, for any $\sigma_{S^{(1)}T^{(1)}}$, $\sigma_{S^{(1)}T^{(1)}}'$, $\sigma_{S^{(2)}T^{(2)}}$, and $\sigma_{S^{(2)}T^{(2)}}' \in \mathcal{D}_{\mathcal{C}}$, the equation
    \begin{eqnarray*}
        {\rm Para}_{\Psi,E}\left[\sigma_{S^{(1)}T^{(1)}},\sigma_{S^{(2)}T^{(2)}}\right]=
        {\rm Para}_{\Psi,E}\left[\sigma_{S^{(1)}T^{(1)}}',\sigma_{S^{(2)}T^{(2)}}'\right]
    \end{eqnarray*}
    holds whenever   $E\left(\sigma_{S^{(1)}T^{(1)}}\right)=E\left(\sigma_{S^{(1)}T^{(1)}}'\right)$ and $E\left(\sigma_{S^{(2)}T^{(2)}}\right)=E\left(\sigma_{S^{(2)}T^{(2)}}'\right)$.
The map $\Psi$ is called \emph{$E$-order-preserving} if the inequalities
   \begin{eqnarray*}
       E\left(\sigma_{S^{(1)}T^{(1)}}\right) \geq E\left(\sigma_{S^{(1)}T^{(1)}}'\right) \quad \text{and} \quad E\left(\sigma_{S^{(2)}T^{(2)}}\right) \geq E\left(\sigma_{S^{(2)}T^{(2)}}'\right)
   \end{eqnarray*}
   imply
   \begin{eqnarray}\label{eq-phi_order_preserving}
      {\rm Para}_{\Psi,E}\left[\sigma_{S^{(1)}T^{(1)}},\sigma_{S^{(2)}T^{(2)}}\right] \geq
      {\rm Para}_{\Psi,E}\left[\sigma_{S^{(1)}T^{(1)}}',\sigma_{S^{(2)}T^{(2)}}'\right].
   \end{eqnarray}
Furthermore, $\Psi$ is said to be \emph{strictly $E$-order-preserving} if $\Psi$ further satisfies that the equality in
Ineq.~\eqref{eq-phi_order_preserving} holds if and only if
\begin{eqnarray*}
  E\left(\sigma_{S^{(1)}T^{(1)}}\right) = E\left(\sigma_{S^{(1)}T^{(1)}}'\right) \quad \text{and} \quad E\left(\sigma_{S^{(2)}T^{(2)}}\right) = E\left(\sigma_{S^{(2)}T^{(2)}}'\right).
  \end{eqnarray*}
\end{definition}}

{Next we show the entanglement concentration models for pure two-qudit states, pure two-qubit states, and TMSVSs, and their properties in preserving specific entanglement measures (and the corresponding order).}

\textbf{Example 4.3} (\textbf{DV-based QNs}) Consider the DV case where $\mathcal{D}_{\mathcal C}$ is the set of all pure two-qubit states. Let the entanglement measure $E$ be the concurrence $c$, $\sigma_{S^{(1)}T^{(1)}},\sigma_{S^{(2)}T^{(2)}}\in\mathcal{D}_{\mathcal C}$, and $\Psi$ be the entanglement concentration protocol proposed in Ref.~\cite{Percolation_ConPT2021}. Then the output state $\rho_{S^{(1)}T^{(1)}}$ can be written as
\begin{eqnarray*}\label{eq-psi_2qubit}
    \rho_{S^{(1)}T^{(1)}}=\Psi_{ST}(\sigma_{S^{(1)}T^{(1)}}\otimes\sigma_{S^{(2)}T^{(2)}})=|\psi_{\lambda}\rangle\langle\psi_{\lambda}|
\end{eqnarray*}
The corresponding parallel rule
${\rm Para}_{\Psi,c}\left[c(\sigma_{S^{(1)}T^{(1)}}),c(\sigma_{S^{(2)}T^{(2)}})\right]=c(\rho_{S^{(1)}T^{(1)}})$
under the concurrence $c$ satisfies
\begin{eqnarray}\label{eq-para_c}
    \frac{1+\sqrt{1-{\rm Para}_{\Psi,c}^2\left[\sigma_{S^{(1)}T^{(1)}},\sigma_{S^{(2)}T^{(2)}}\right]}}{2}
    =\max\left\{\frac{1}{2},\prod_{k=1}^{2}\frac{1+\sqrt{1-c^2\left(\sigma_{S^{(k)}T^{(k)}}\right)}}{2}\right\}.
\end{eqnarray}
It follows that the entanglement concentration map $\Psi$ is $c$-preserving. Moreover, $\Psi$ is strictly $c$-order-preserving when
\begin{eqnarray*}
\prod_{k=1}^{2}\frac{1+\sqrt{1-c^2(\sigma_{S^{(k)}T^{(k)}})}}{2}
> \frac{1}{2}.
\end{eqnarray*}
However, when
\begin{eqnarray*}
    \prod_{k=1}^{2}\frac{1+\sqrt{1-c^2\left(\sigma_{S^{(k)}T^{(k)}}\right)}}{2}
    \leq\frac{1}{2},
\end{eqnarray*}
the equality ${\rm Para}_{\Psi,c}\left[c(\sigma_{S^{(1)}T^{(1)}}),c(\sigma_{S^{(2)}T^{(2)}})\right]=1$ always holds, which implies that $\Psi$ is $c$-order-preserving but not strictly $c$-order-preserving.

{Consider the more general DV case for $\mathcal{D}_{\mathcal{C}}$ being the set of all pure two-qudit states. Let the nonzero-Schmidt-value vectors of two-qudit states $\sigma_{S^{(1)}T^{(1)}}$ and $\sigma_{S^{(2)}T^{(2)}}$ be} $\vec{\lambda}=(\lambda_0,...,\lambda_{d_1-1})$ and $ \vec{\mu}=(\mu_0,...,\mu_{d_2-1})$, respectively.
By the majorization theory for pure-state conversion (Sec.~\ref{sec-majorizaiton}), one obtains the entanglement concentration $\Psi_{ST}$ that converts $\sigma_{S^{(1)}T^{(1)}}\otimes\sigma_{S^{(2)}T^{(2)}}$ into a pure two-qudit state $|\phi\rangle$ with its number $l$ of nonzero Schmidt values no more than $d_1d_2$.
Specifically, the performance of $\Psi_{ST}$ can be written as
\begin{eqnarray*}
\Psi_{ST}(\sigma_{S^{(1)}T^{(1)}}\otimes\sigma_{S^{(2)}T^{(2)}})=|\phi\rangle
\langle\phi|_{S^{(1)}T^{(1)}}.
\end{eqnarray*}
where $|\phi\rangle$ has Schmidt decomposition
 \begin{eqnarray}\label{eq-qudit_concen}
|\phi\rangle_{S^{(1)}T^{(1)}}=\sum_{k=0}^{l-1}\sqrt{\nu_k}|k_{S^{(1)}}k_{T^{(1)}}\rangle
 \end{eqnarray}
with $\nu_0=\lambda_0\mu_0$, $\nu_n=\max\left\{\max\limits_{a\in {\{\lambda_i\mu_j\}_{{i,j}\geq 0}\backslash\{\nu_1,\nu_2,\cdots, \nu_{n-1}\}}}a,{(1-\sum_{m=1}^{n-1}\nu_m)}/{(l-n+1)}\right\}$.
{However, to date, finding a suitable entanglement measure that characterizes the parallel rule for $\Psi_{ST}$ in terms of the entanglement of the initial states remains an open problem.
\hfill $\square$}

\textbf{Example 4.4} (\textbf{CV-based QNs})
{We now consider the CV case where $\mathcal{D}_{\mathcal{C}}$ is the set of all TMSVSs and the entanglement measure $E$ is chosen as the ratio negativity $\chi_\mathcal{N}$.}
Let $\sigma_{S^{(1)}T^{(1)}}=|\psi^{r_1}\rangle\langle\psi^{r_1}|$ and $\sigma_{S^{(2)}T^{(2)}}=|\psi^{r_2}\rangle\psi^{r_2}|$ with $r_1\geq r_2$, and let map $\Psi$ be the normalized entanglement concentration for TMSVSs proposed in Ref.~{\cite{NegPT}}. Then the output state is a new TMSVS
\begin{eqnarray*}\label{eq-psi_TMSVS}
    \rho_{S^{(1)}T^{(1)}}=\Psi_{ST}(\sigma_{S^{(1)}T^{(1)}}\otimes\sigma_{S^{(2)}T^{(2)}})=|\psi^r\rangle\langle\psi^r|
\end{eqnarray*}
with
\begin{eqnarray*}
    \sinh r=\sinh r_1 \cosh r_2.
\end{eqnarray*}
This concentration process yields a parallel rule ${\rm Para}_{\Psi,\chi_{\mathcal{N}}}\left[\sigma_{S^{(1)}T^{(1)}},\sigma_{S^{(2)}T^{(2)}})\right]=\chi_\mathcal{N}(\rho_{S^{(1)}T^{(1)}})$ with
\begin{eqnarray}\label{eq-para_chi}
    \frac{{\rm Para}_{\Psi,\chi_{\mathcal{N}}}^2\left[\sigma_{S^{(1)}T^{(1)}},\sigma_{S^{(2)}T^{(2)}}\right]}{1-{\rm Para}_{\Psi,\chi_{\mathcal{N}}}^2\left[\sigma_{S^{(1)}T^{(1)}},\sigma_{S^{(2)}T^{(2)}}\right]}
    =\frac{\chi_{\mathcal{N}}^2\left(\sigma_{S^{(1)}T^{(1)}}\right)}{\prod_{k=1}^{2} \left[1-\chi_{\mathcal{N}}^2\left(\sigma_{S^{(k)}T^{(k)}}\right)\right]}.
\end{eqnarray}
It follows that $\Psi$ is both $\chi_{\mathcal{N}}$-preserving and strictly $\chi_{\mathcal{N}}$-order-preserving for all TMSVSs.
\hfill $\square$

{However, whether these properties hold for general entangled states and entanglement measures remains an open question.}

\section{Operator model of entanglement percolation processing}

{In this section, we establish an operator-theoretic framework for entanglement percolation. Within this framework, entanglement percolation in a series-parallel QN is formulated as a sequential composition of super-operators. We begin our discussion with the basic building blocks: series and parallel networks.}

\subsection{Series networks}\label{sec-series_QNs}

We now investigate entanglement swapping protocols on a series QN [Fig.~\ref{fig-seriN}]. {These protocols are implemented by successively applying the removal-node entanglement swapping map $\Phi$ (see Eq.~\eqref{eq-swapping_def}) to eliminate intermediate nodes one after another. By examining different orders of node removal, we find that distinct operation sequences generally lead to different final states.}

\begin{figure}[h!]
    \centering
    \includegraphics[width=250pt]{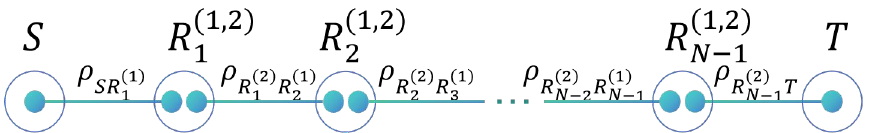}
    \caption{Series QN of $N$ states.}
    \label{fig-seriN}
\end{figure}

Consider the series QN consisting of $N+1$ nodes---$S$, $R_1 $, $R_2 $,$\cdots$, $R_{N-1} $, $T$---as illustrated in Fig.~\ref{fig-seriN}. {These nodes are connected by $N$ states in $\mathcal{D}_{\mathcal{S}}$, denoted from left to right as $\rho_{SR_1^{(1)}}, \rho_{R_1^{(2)}R_2^{(1)}}, \dots, \rho_{R_{N-1}^{(2)}T}$.}
We now introduce an operator-theoretic framework for entanglement \yaqi{distribution} based on entanglement swapping maps along this series topology.
To distribute entanglement across the QN, one must perform entanglement swapping as shown in Eq.~\eqref{eq-SwapPhi} for every intermediate node.
Denoted the entire QN state by
\begin{eqnarray*}\label{eq-rho_net}
    \rho_{{\rm net},N}=\rho_{SR_1^{(1)}}\otimes\rho_{R_1^{(2)}R_2^{(1)}}\otimes\cdots\otimes\rho_{R_{N-1}^{(2)}T}.
\end{eqnarray*}
Let the vector
\begin{eqnarray}\label{eq-order}
    \pi_{N-1}=(\pi_{N-1}(1),\pi_{N-1}(2),...,\pi_{N-1}(N-1))
\end{eqnarray}
be a permutation of $(1, 2,..., N-1)$ specifying the sequence in which the swapping maps occur at nodes $R_{\pi_{N-1}(1)}$, $R_{\pi_{N-1}(2)}$,...,$R_{\pi_{N-1}(N-1)}$.
Then the overall entanglement swapping protocol is given by a composition
\begin{eqnarray}\label{eq-1D_swapPhi}
    \Phi^{\pi_{N-1}}\left(\rho_{{\rm net},N}\right)
    =\Phi_{R_{\pi_{N-1}{(N-1)}}}^{\pi_{N-1}}\circ\Phi_{R_{\pi_{N-1}{(N-2)}}}^{\pi_{N-1}}\circ\cdots\circ\Phi_{R_{\pi_{N-1}{(1)}}}^{\pi_{N-1}}\left(\rho_{{\rm net},N}\right),
\end{eqnarray}
where, the notation $\Phi_{R_{\pi_{N-1}{(n)}}}^{\pi_{N-1}}$ means that, at the $n$-th step of executing the entanglement swapping following the sequence $\pi_{N-1}$, we perform the entanglement swapping [Eq.~\eqref{eq-SwapPhi}] only to the two states adjacent to the intermediate node $R_{\pi_{N-1}{(n)}}$, leaving all other states unchanged.
If local unitary operations follow each swapping, then each swapping map is replaced by the operator form given in Eq.~\eqref{eq_swap2_unitary}.

In summary, the ultimate goal of entanglement percolation is to establish entanglement between two distant nodes.  We quantify this using the entanglement measure $E$, where a value greater than zero certifies the presence of entanglement.
For this series QN, we define the {\it series rule} reduced from the swapping protocol $\Phi^{\pi_{N-1}}$ as
\begin{eqnarray*}\label{eq-seri_rule_N}
    {\rm Seri}_{\Phi,E}^{\pi_{N-1}}\left[\rho_{SR_1^{(1)}},\rho_{R_1^{(2)}R_2^{(1)}},\cdots,\rho_{R_{N-1}^{(2)}T}\right]:
    =E\left[\Phi^{\pi_{N-1}}\left(\rho_{{\rm net},N}\right)\right].
\end{eqnarray*}
{As shown above, different choices of the permutation vector $\pi_{N-1}$ result in different operation sequences for performing entanglement swapping on the intermediate nodes. Consequently, we denote $\pi_{N-1}$ as the \textbf{swapping order vector} associated with the $N$ initial states.}

This raises the question of identifying the operation sequence that maximizes the final entanglement, referred to as the optimal order. This challenging problem will be addressed in the next section. Before proceeding, we present several examples of entanglement percolation in series QNs to illustrate the concepts discussed above.

{\bf Example 5.1} (\textbf{DV-based QNs})
Let $\mathcal{D}_{\mathcal{S}}$ be the set of all pure two-qubit entangled states, and $\Phi$ be the entanglement swapping proposed in Ref.~\cite{Percolation_ConPT2021}.
Consider the series QN is the DV-based QN of $N$ two-qubit states $\rho_{SR_1^{(1)}}=|\psi_{\lambda_{1}}\rangle\langle\psi_{\lambda_{1}}|$, $\rho_{R_n^{(2)}R_{n+1}^{(1)}}=|\psi_{\lambda_{n+1}}\rangle\langle\psi_{\lambda_{n+1}}|$ ($n=1,2,...,N-2$), and $\rho_{R_{N-1}^{(2)}T}=|\psi_{\lambda_{N}}\rangle\langle\psi_{\lambda_{N}}|$.
Let the order of entanglement swapping be given by an arbitrary permutation $\pi_{N-1}$ of $(1,2,...,N-1)$.
By the series rule ${\rm Seri}_c$ [Eq.~\eqref{eq-seri_c}], the final state shared between $S$ and $T$ is
\begin{eqnarray*}
    \rho_{ST}=\Phi^{\pi_{N-1}}\left(\rho_{{\rm net},N}\right)=|\psi_{\lambda}\rangle\langle\psi_{\lambda}|
\end{eqnarray*}
with
\begin{eqnarray*}\label{eq-lambda_1D}
    \lambda=\frac{1+\sqrt{1-\prod_{n=1}^{N}c_n^2}}{2}
\end{eqnarray*}
where $c_n=2\sqrt{\lambda_n(1-\lambda_n)}$ denotes the concurrence of the $n$-th source state $\rho_n$.
That is the final concurrence $c=2\sqrt{\lambda(1-\lambda)}$ satisfies the series rule
\begin{eqnarray}\label{eq-swap_1D_concurrence}
    {\rm Seri}_{\Phi,c}\left[\rho_{SR_1^{(1)}},\rho_{R_1^{(2)}R_2^{(1)}},\cdots,\rho_{R_{N-1}^{(2)}T}\right]=c=\prod_{n=1}^N c_n.
\end{eqnarray}

More generally, we show the case where $\mathcal{D}_{\mathcal{S}}$ is the set of all pure two-qudit entangled states, and $\Phi$ is the entanglement swapping on two-qudit states.
Consider the series QN is the DV-based QN of $N$ two-qudit states $\rho_{SR_1^{(1)}}=|\phi_{1}\rangle\langle\phi_{1}|$, $\rho_{R_n^{(2)}R_{n+1}^{(1)}}=|\phi_{n+1}\rangle\langle\phi_{n+1}|$ ($n=1,2,...,N-2$), and $\rho_{R_{N-1}^{(2)}T}=|\phi_{N}\rangle\langle\phi_{N}|$.
From the series rule ${\rm{Seri}}_{C_G}$ [Eq.~\eqref{eq-seri_CG}], performing the $\Phi$-based entanglement swapping in any order $\pi_{N-1}$, the final two-qudit state $\rho_{ST}$ yields the series rule
\begin{eqnarray*}
    {\rm Seri}_{\Phi,C_G}^{\pi_{N-1}}\left[\rho_{SR_1^{(1)}},\rho_{R_1^{(2)}R_2^{(1)}},\cdots,\rho_{R_{N-1}^{(2)}T}\right]:
    =\prod_n C_G\left(|\phi_{n}\rangle\langle\phi_{n}|\right),
\end{eqnarray*}
independently of the choice of $\pi_{N-1}$.
\hfill $\square$

{\bf Example 5.2} (\textbf{CV-based QNs})
Here let $\mathcal{D}_{\mathcal{S}}$ be the set of all TMSVSs, and $\Phi$ be the entanglement swapping proposed in {Ref.~{\cite{NegPT}}}.
Consider a $N$-state series QN of TMSVSs $\rho_{SR_1^{(1)}}=|\psi^{r_{1}}\rangle\langle\psi^{r_{1}}|$, $\rho_{R_n^{(2)}R_{n+1}^{(1)}}=|\psi^{r_{n-1}}\rangle\langle\psi^{r_{n-1}}|$ ($n=1,2,...,N-2$), and $\rho_{R_{N-1}^{(2)}T}=|\psi^{r_{N}}\rangle\langle\psi^{r_{N}}|$.
Under the series rule ${\rm Seri}_{\chi_{\mathcal{N}}}$ [Eq.~\eqref{eq-seri_chi}], the output TMSVS is
\begin{eqnarray*}
    \rho_{ST}=\Phi^{\pi_{N-1}}\left(\rho_{{\rm net},N}\right)=|\psi^r\rangle\langle\psi^r|
\end{eqnarray*}
whose ratio negativity $\chi={\rm Seri}_{\Phi,\chi_{\mathcal{N}}}^{\pi_{N-1}}\left[\rho_{SR_1^{(1)}},\rho_{R_1^{(2)}R_2^{(1)}},\cdots,\rho_{R_{N-1}^{(2)}T}\right]$ satisfies the parallel rule
\begin{eqnarray}\label{eq-seri_chi_N}
   {\rm Seri}_{\Phi,\chi_{\mathcal{N}}}^{\pi_{N-1}}\left[\rho_{SR_1^{(1)}},\rho_{R_1^{(2)}R_2^{(1)}},\cdots,\rho_{R_{N-1}^{(2)}T}\right]=\prod_{n=1}^{N}\chi_{\mathcal{N}}\left(|\psi^{r_{n}}\rangle\langle\psi^{r_{n}}|\right)=\prod_{n=1}^{N}\tanh r_{n}
\end{eqnarray}
for any permutation $\pi_{N-1}$.
Thus the effective squeezing $r$ is unaffected by the permutation $\pi_{N-1}$.
\hfill $\square$

{In both examples above, the output state does not depend on the swapping order. Hence, any order is optimal.}

\subsection{Parallel networks}\label{sec-parallel_QN}

When two distant nodes $S$ and $T$ share multiple states, a global entanglement concentration map over all states can be achieved by iteratively applying the normalized concentration map $\Psi$.

\begin{figure}
    \centering
    \includegraphics[width=150pt]{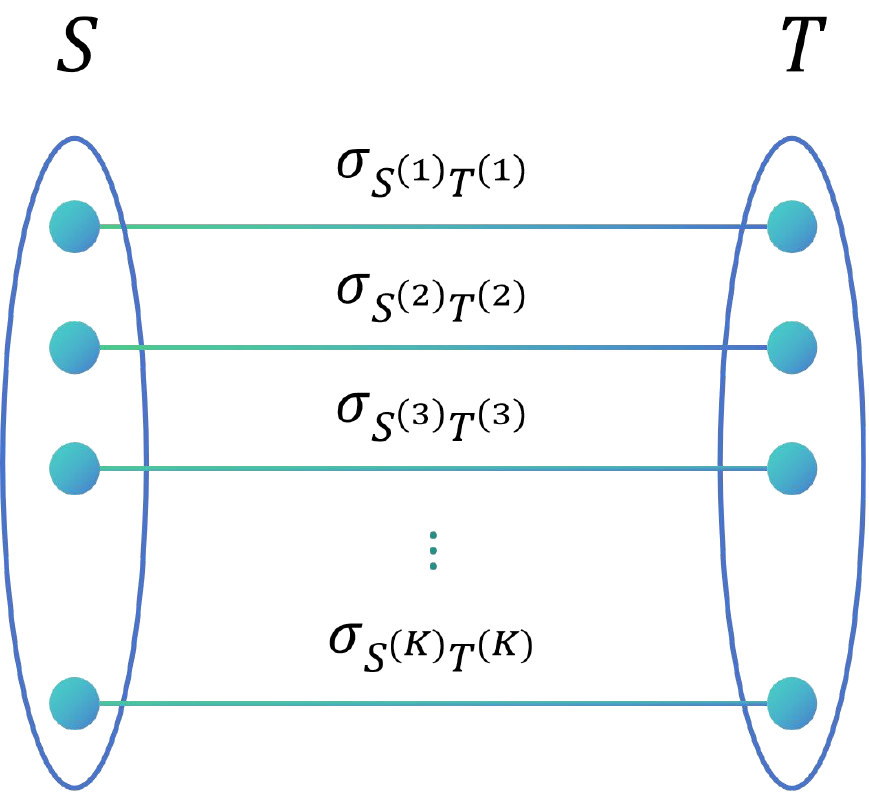}
    \caption{Parallel QN with $K$ source states.}
    \label{fig-parallel_K}
\end{figure}

Consider a parallel QN as shown in Fig.~\ref{fig-parallel_K} where two adjacent nodes $S$ and $T$ share $K$ states $\sigma_{S^{(1)}T^{(1)}}$, $\sigma_{S^{(2)}T^{(2)}}$,..., $\sigma_{S^{(K)}T^{(K)}}$ in $\mathcal{D}_{\mathcal{C}}$ where $K\geq 3$.
Entanglement concentration on these states can be achieved through sequential pairwise operations.
Here, to facilitate later discussions on how the order of concentrated operations affects the percolation outcome, we arrange the states in a order denoted by the permutation (or order)
\begin{eqnarray*}\label{eq-concentration_order}
    \pi_K=\left(\pi_K(1),\pi_K(2), ...,\pi_K(K)\right)
\end{eqnarray*}
of $(1,2,...,K)$, yielding the concentration sequence $\sigma_{S^{(\pi_K(1))}T^{(\pi_K(1))}}$, $\sigma_{S^{(\pi_K(2))}T^{(\pi_K(2))}}$,..., $\sigma_{S^{(\pi_K(K))}T^{(\pi_K(K))}}$ of states.
By the symmetry shown in Eq.~\eqref{eq-symmetry_C}, we assume $\pi_K(1)<\pi_K(2)$.
Let
\begin{eqnarray*}\label{eq-omega}
    \Omega_{k}=\left\{\pi_K(1),\pi_K(2),...,\pi_K(k+1)\right\}
\end{eqnarray*}
with $k=1,...,K-1$, then we construct the concentration protocol corresponding to the order $\pi_K$ where the $k$-th step ($k=1,2,...,K-1$) is designed as:

(1) If $k=1$, concentrate $\sigma_{S^{\left(\pi_K(1)\right)}T^{\left(\pi_K(1)\right)}}$ and $\sigma_{S^{\left(\pi_K(2)\right)}T^{\left(\pi_K(2)\right)}}$ to get output state
\begin{eqnarray*}
    \sigma_{ST}^{\Omega_1}=\Psi_{ST}\left(\sigma_{S^{\left(\pi_K(1)\right)}T^{\left(\pi_K(1)\right)}}\otimes \sigma_{S^{\left(\pi_K(2)\right)}T^{\left(\pi_K(2)\right)}}\right);
\end{eqnarray*}

(2) If $k>1$, concentrate the state $\sigma_{ST}^{\Omega_{k-1}}$ and $\sigma_{\pi_K(k+1)}$ to get output state
\begin{eqnarray*}
    \sigma_{ST}^{\Omega_{k}}
    =\Psi_{ST}\left(\sigma_{ST}^{\Omega_{k-1}}\otimes \sigma_{S^{\left(\pi_K(k+1)\right)}T^{\left(\pi_K(k+1)\right)}}\right).
\end{eqnarray*}
Then we obtain the final output state
\begin{eqnarray*}
    \sigma_{ST}'
    &=&\sigma_{ST}^{\Omega_{K-1}}\nonumber\\
    &=&\Psi_{ST}\left(\sigma_{ST}^{\Omega_{K-2}}\otimes \sigma_{S^{\left(\pi_K(K)\right)}T^{\left(\pi_K(K)\right)}}\right)\nonumber\\
    &=&\Psi_{ST}\left[\Psi_{ST}\left(\sigma_{ST}^{\Omega_{K-3}}\otimes \sigma_{S^{\left({\pi_K(K-1)}\right)}T^{\left(\pi_K(k+1)\right)}}\right)\otimes \sigma_{\pi_K(K)}\right]\nonumber\\
    &=&\cdots\nonumber\\
    &=&\Psi_{ST}\Big\{\Psi_{ST}\Big[\cdots\Psi_{ST}\Big(\Psi_{ST}(\sigma_{S^{\left({\pi_K(1)}\right)}T^{\left(\pi_K(1)\right)}}\otimes\sigma_{S^{\left(\pi_K(2)\right)}T^{\left(\pi_K(2)\right)}})\nonumber\\
    &&\ \ \ \ \ \ \ \ \ \ \ \ \ \ \ \ \ \ \ \ \otimes\sigma_{S^{\left(\pi_K(3)\right)}T^{\left(\pi_K(3)\right)}}\Big)\Big]\otimes\sigma_{S^{\left(\pi_K(K)\right)}T^{\left(\pi_K(K)\right)}}\Big\}.
\end{eqnarray*}
For simplicity, we denote the operation above as
\begin{eqnarray*}                             \sigma_{ST}'=\Psi_{ST}^{\pi_K}\left(\sigma_{\rm{net},K}\right)
\end{eqnarray*}
where $\sigma_{\rm{net},K}=\otimes_{k=1}^{K}\sigma_{S^{(k)}T^{(k)}}$ and the global entanglement concentration defined as
\begin{eqnarray}\label{eq-concentration_1D}
    \Psi_{ST}^{\pi_K}\left(\sigma_{\rm{net},K}\right):=\Psi_{ST}^{\Omega_{K-1}}\circ\Psi_{ST}^{\Omega_{K-2}}\circ\cdots\circ\Psi_{ST}^{\Omega_{1}}\left(\sigma_{\rm{net},K}\right).
\end{eqnarray}
Here, $\Psi_{ST}^{\Omega_{k}}=\Psi_{ST}\otimes I$ represents the operation performing the concentration map $\Psi_{ST}$ on states $\sigma_{ST}^{\Omega_{k-1}}$ and $\sigma_{S^{(\pi_K(k+1))}T^{(\pi_K(k+1))}}$, while leaving remaining states unchanged.
We call $\pi_K$ the \textbf{concentration order} of the $K$ initial states.
For this parallel QN, we define the \textbf{parallel rule} reduced from $\Psi_{ST}^{\pi_{K}}$ as
\begin{eqnarray*}\label{eq-parallel_rule_K}
    {\rm Para}_{\Psi,E}^{\pi_{K}}\left[\sigma_{S^{(1)}T^{(1)}},\sigma_{S^{(2)}T^{(2)}},\cdots,\sigma_{S^{(K)}T^{(K)}}\right]:
    =E\left[\Psi_{ST}^{\pi_K}\left(\sigma_{\rm{net},K}\right)\right].
\end{eqnarray*}
The question of optimal operational order in parallel networks will also be addressed in the next section. Next we show two examples for the entanglement percolation in parallel QNs.

{\bf Example 5.3} (\textbf{DV-based QNs})
Now consider the parallel DV-based QN of $K$ two-qubit states $\sigma_{S^{(k)}T^{(k)}}=|\psi_{\mu_k}\rangle\langle\psi_{\mu_k}|$ for $k=1,2,...,K$.
By the parallel rule ${\rm Para}_{\Psi,c}$ [Eq.~\eqref{eq-para_c}], the state resulting from concentrating these $K$ states is
\begin{eqnarray*}
    \sigma_{ST}=\Psi_{ST}^{\pi_K}\left(\sigma_{\rm{net},K}\right)=|\psi_{\mu}\rangle\langle\psi_{\mu}|
\end{eqnarray*}
with
\begin{eqnarray*}
    \mu=\max\left\{\frac{1}{2},\prod_{k=1}^{K}\mu_{k}\right\}
\end{eqnarray*}
for any permutation $\pi_K$ of $(1,2,...,K-1)$,
implying the final concurrence $c$ satisfies the parallel rule
\begin{eqnarray}\label{eq-para_c_K}
      {\rm Para}_{\Psi,E}^{\pi_{K}}\left[\sigma_{S^{(1)}T^{(1)}},\sigma_{S^{(2)}T^{(2)}},\cdots,\sigma_{S^{(K)}T^{(K)}}\right]=c=2\sqrt{\mu(1-\mu)}.
\end{eqnarray}
Again, the final concentrated state is independent of the concentration order $\pi_K$.
\hfill $\square$

{\bf Example 5.4} (\textbf{CV-based QNs})
Let $\mathcal{D}_{\mathcal{C}}$ be the set of all TMSVSs and $\sigma_{{S}^{(k)}T^{(k)}}=|\psi^{r_k}\rangle\langle \psi^{r_k}|\in\mathcal{D}_{\mathcal{C}}$ with $r_1\geq r_k$ for all $k=2,3,...,K$. Furthermore, let $\Psi$ be the entanglement concentration described in Example 4.4.
{In contrast to the order-independent entanglement swapping for the series CV-based QN shown in Sec.~\ref{sec-concentration}, the entanglement concentration $\Psi_{ST}^{\pi_K}$ for TMSVSs yields output states that depend explicitly on the order $\pi_K$.}
{When $\pi_K(1)=1$, under the ratio negativity $\chi_{\mathcal{N}}$, the parallel rule ${\rm Para}_{\Psi,\chi_{\mathcal{N}}}^{\pi_K}$ for TMSVSs, $\sigma_{{S}^{(1)}T^{(1)}},\sigma_{{S}^{(2)}T^{(2)}},...,\sigma_{{S}^{(K)}T^{(K)}}$, attains the maximum~\cite{NegPT}
$${\rm Para}_{\Psi,\chi_{\mathcal{N}}}^{\pi_K}\left[\sigma_{S^{(1)}T^{(1)}},\sigma_{S^{(2)}T^{(2)}},\cdots,\sigma_{S^{(K)}T^{(K)}}\right]=\sinh r_1\prod_{k=2}^{K}\cosh r_k.$$}
\hfill $\square$

\subsection{Series-parallel networks}\label{sec-s_p_operator}

{The entanglement percolation process in series-parallel QNs is a process of network reduction via iteratively reducing the network through entanglement swapping and concentration operations, ultimately establishing end-to-end entanglement.}

Given a series-parallel QN with source states in a set $\mathcal{D}$ of bipartite states, both removal-node entanglement swapping $\Phi$ [Eq.~\eqref{eq-SwapPhi}] and the normalized entanglement concentration $\Psi$ [Eq.~\eqref{eq-concen2}] are maps from $\mathcal{D}\otimes \mathcal{D}$ to $\mathcal{D}$.
{{In analogy with the simplification operations} (${\rm I}$) and (${\rm II}$) of networks shown in Sec.~\ref{sec-network}, the series-parallel QN can be iteratively {simplified through successive applications of} $\Phi$ and $\Psi$.
This {simplification} process enables entanglement percolation---the \yaqi{distribution} of entanglement across the network. Specifically, the quantum process proceeds {via the following iterative steps}.
}

{(${\rm I^1}$) \emph{Simplification for two parallel sources.} When there exist two states $\sigma_{R_{j_1}^{(l_{j_1,j_2})}R_{j_2}^{(l_{j_2,j_1})}}$ and $\sigma_{R_{j_1}^{(l_{j_1,j_2}+1)}R_{j_2}^{(l_{j_2,j_1}+1)}}$ are shared by two adjacent nodes $R_{j_1}$ and $R_{j_2}$, perform the entanglement concentration $\Psi$ to concentrate these state into a new state
	\begin{eqnarray*}
		\sigma_{{R_{j_1}R_{j_2}}}=\Psi_{R_{j_1}R_{j_2}}
		\left(\sigma_{R_{j_1}^{(l_{j_1,j_2})}R_{j_2}^{(l_{j_2,j_1})}}\otimes\sigma_{R_{j_1}^{(l_{j_1,j_2}+1)}R_{j_2}^{(l_{j_2,j_1}+1)}}\right)
	\end{eqnarray*}
	shared between $R_{j_1}$ and $R_{j_2}$.
	Repeat this until there are no such states can be further concentrated, then proceed to the next step.
}

{(${\rm {II}^1}$) \emph{Simplification for two series states.} Consider three nodes denoted by $R_{n_1}$, $R_{n_2}$ and $R_{n_3}$ that lie consecutively on a self-avoiding path between $S$ and $T$.
	Suppose that $R_{n_1}$ and $R_{n_2}$ share only a state $\rho_{R_{n_1}^{(l_{n_1,n_2})}R_{n_2}^{(l_{n_2,n_1})}}$, $R_{n_2}$ and $R_{n_3}$ share only a state $\rho_{R_{n_2}^{(l_{n_2,n_3})}R_{n_3}^{(l_{n_3,n_2})}}$, and $R_1$ and $R_3$ are all adjacent nodes of $R_{2}$.
	Then perform entanglement swapping $\Phi$ to convert the two states into a single state
	\begin{eqnarray*}
		\rho_{R_{n_1}^{(l_{n_1,n_2})}R_{n_3}^{(l_{n_3,n_2})}} =\Phi_{R_{n_2}}^{R_{n_1}R_{n_3}}\left(\rho_{R_{n_1}^{(l_{n_1,n_2})}R_{n_2}^{(1)}}\otimes\rho_{R_{n_2}^{(2)}R_{n_3}^{(l_{n_3,n_2})}}\right)
	\end{eqnarray*}
	between $R_{n_1}$ and $R_{n_3}$. Repeat this until no such series states can be further swapped, then return to Step (${\rm I^\prime}$)} if necessary.

\begin{figure}
    \centering
    \includegraphics[width=120pt]{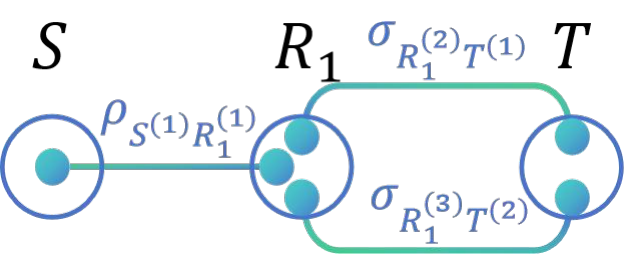}
    \caption{A series-parallel QN.}
    \label{fig-s_p_1}
\end{figure}

For any series-parallel QN, let $\rho_{\rm net}$ be the network state. The entanglement \yaqi{distribution} over the QN, via the aforementioned simplification process, can be written as a map given by
\begin{eqnarray*}
    \Lambda(\rho_{\rm net})=\Lambda_m\circ\Lambda_{m-1}\circ...\circ\Lambda_1(\rho_{\rm net}),\quad m\in\{1,2,...,+\infty\}
\end{eqnarray*}
where each ${\rm \Lambda}_i$ is either an entanglement swapping or an entanglement concentration map.
For instance, in the QN shown in Fig.~\ref{fig-s_p_1} where $S$ and $R_1$ share a state $\rho_{SR_1^{(1)}}$, $R_1$ and $T$ share two states $\sigma_{R_1^{(2)}T^{(1)}}$ and $\sigma_{R_1^{(3)}T^{(2)}}$, the entanglement \yaqi{distribution} process constitutes a composite map $\Lambda$ integrating an entanglement concentration map and an entanglement swapping map, denoted as
\begin{eqnarray*}
    \Lambda=\Lambda_2\circ\Lambda_1
\end{eqnarray*}
where $\Lambda_1=I_{S}\otimes\Psi_{R_1T}$ denotes the application of entanglement concentration to the states shared between $R_1$ and $T$, and $\Lambda_2=\Phi_{R_1}^{ST}$ represents the application of entanglement swapping to the states $\rho_{SR_1^{(1)}}$ and $\Psi_{R_1T}\Big(\sigma_{R_1^{(2)}T^{(1)}} \otimes \sigma_{R_1^{(3)}T^{(2)}}\Big)$.

{In non-series-parallel configurations, entanglement \yaqi{distribution} can be accomplished by combining entanglement swapping and concentration operations with the star-mesh transform technique~\cite{SM_transform1970}, which is itself based entirely on series-parallel rules.}

\section{Optimal degeneration order in entanglement percolation}

Intrinsically, entanglement percolation is a process of QN reduction. Although we have already formulated an operator-theoretic description of the percolation process in the preceding section, a crucial question remains: the order in which operations are executed can influence the final entanglement of the percolation outcome, and an unfavorable sequence {causes inefficient entanglement percolation}. Consequently, in this section we investigate the optimal operation order in entanglement percolation.

\subsection{Series networks}\label{sec-series_order}

For both DV-based series QNs of two-qudit states and CV-based series QNs of TMSVSs shown in Sec.~\ref{sec-series_QNs}, the final entanglement remains invariant under all swapping orders.
However, for general series QNs, different swapping orders produce distinct entanglement values.
Now we define the optimal swapping order [Eq.~\eqref{eq-order}]:
Given the $N$-state series QN with QN state $\rho_{{\rm net},N}=\rho_{SR_1^{(1)}}\otimes\rho_{R_1^{(2)}R_2^{(1)}}\otimes\cdots\otimes\rho_{R_{N-1}^{(2)}T}$ and the entanglement measure $E$ quantifying the entanglement, if the final entanglement swapping order vector $\pi_{N-1}$ yields maximum among all possible swapping orders---i.e., $E\left[\Phi^{\pi_{N-1}}(\rho_{{\rm net},N})\right] \geq E\left[\Phi^{\pi_{N-1}'}(\rho_{{\rm net},N})\right]$ holds for any other swapping order $\pi_{N-1}'$---we call $\pi_{N-1}$ the {\bf $\Phi$-based optimal swapping order}. Similarly, the map $\Phi^{\pi_{N-1}}$ is called the {\bf $\Phi$-based optimal entanglement swapping protocol}.

The order preservation of entanglement measures under $\Phi$ is crucial for establishing the sufficient condition for order-independent entanglement in series QNs. We now formally define the order preservation property for entanglement swapping protocols.

In the following, we investigate how the operation order of entanglement-swapping protocols in series QNs influences the output of entanglement percolation.
The most fundamental and arguably the first question regarding how operational order affects entanglement percolation is: under what conditions does exchanging the sequence of two swapping operations leave the percolation outcome unchanged? For this, we consider the series QN with three source states $\rho_{SR_1^{(1)}}$, $\rho_{R_1^{(2)}R_2^{(1)}}$, $\rho_{R_2^{(2)}T}$.
We say that $\Phi$ is {\bf $E$-based order-independent} on the state set $\mathcal{D}_{\mathcal{S}}$ if, for all $\rho_{SR_1^{(1)}},\rho_{R_1^{(2)}R_2^{(1)}},\rho_{R_2^{(2)}T}\in\mathcal{D}_{\mathcal{S}}$, {the equation}
    \begin{eqnarray}\label{eq-phi_commutative}
      E\left[\Phi^{(2,1)}\left(\rho_{\rm{net},3}\right)\right]
      =E\left[\Phi^{(1,2)}\left(\rho_{\rm{net},3}\right)\right]
    \end{eqnarray}
{holds true}, where $\rho_{\rm{net},3}=\rho_{SR_1^{(1)}}\otimes\rho_{R_1^{(2)}R_2^{(1)}}\otimes\rho_{R_2^{(2)}T}$.

Recall the definition from Eq. \eqref{eq-Seri_2} that the equation in (\ref{eq-phi_commutative}) is equivalent to
\begin{eqnarray*}
	{\rm Seri}_{\Phi,E}^{(2,1)} \left[\rho_{SR_1^{(1)}},\rho_{R_1^{(2)}R_2^{(1)}},\rho_{R_{N-1}^{(2)}T}\right]
	={\rm Seri}_{\Phi,E}^{(1,2)} \left[\rho_{SR_1^{(1)}},\rho_{R_1^{(2)}R_2^{(1)}},\rho_{R_{N-1}^{(2)}T}\right].
\end{eqnarray*}
We now generalize this property to all series QNs.
We say that $\Phi$ has {\bf generalized $E$-based order independence} on $\mathcal{D}_{\mathcal{S}}$ if, for any finite-size series QN  where all source states belong to $\mathcal{D}_{\mathcal{S}}$, all swapping orders [Eq.~\eqref{eq-1D_swapPhi}] yield identical final entanglement. However, this generalized property cannot be directly derived from the $E$-based order independence for $\Phi$.
We now bridge the gap with the following theorem.
\begin{theorem}\label{thm-commutative}
	If $\Phi$ is both $E$-based order-independent and $E$-preserving, then $\Phi$ has generalized $E$-based order independence.
\end{theorem}
Details for proving this result are shown in Appendix~\ref{SI-commutative}.

In other words, when $\Phi$ exhibits generalized $E$-based order independence on $\mathcal{D}_{\mathcal{S}}$,
the final entanglement between $S$ and $T$---quantified by $E$---is independent of the swapping order $\pi_{N-1}$.

\subsection{Parallel networks}\label{sec-parallel_order}

As demonstrated for parallel QNs in Sec.~\ref{sec-parallel_QN}, Eq.~\eqref{eq-para_c_K} (Example 5.3) shows that the DV case achieves final concentrated entanglement independent of the concentration order [Eq.~\eqref{eq-concentration_order}]. Conversely, for the CV case, the non-symmetric parallel rule shown in Eq.~\eqref{eq-para_chi} results in final established entanglement that depends on the concentration order.

To illustrate, consider a parallel QN where nodes $S$ and $T$ share three TMSVSs $\sigma_{S^{(1)}T^{(1)}}=|\psi^{r_1}\rangle\langle\psi^{r_1}|$, $\sigma_{S^{(12)}T^{(2)}}=|\psi^{r_2}\rangle\langle\psi^{r_2}|$, and $\sigma_{S^{(3)}T^{(3)}}=|\psi^{r_3}\rangle\langle\psi^{r_3}|$ with squeezing parameters satisfying $r_1>r_2>r_3>0$ and $\sinh r_1<\sinh r_2\cosh r_3$. Defining $\chi_k = \tanh r_k$ for $k=1,2,3$, we observe: the concentration order $\pi_3 = (1,2,3)$ yields final ratio negativity
    \begin{eqnarray*}
        \chi = \frac{\chi_1}{\sqrt{\chi_1^2 + \prod_{k=1}^{3}(1-\chi_k^2)}},
    \end{eqnarray*}
while another concentration order $\pi_3' = \{2,3,1\}$ produces a different value:
    \begin{eqnarray*}
        \chi' = \frac{\chi_2}{\sqrt{\chi_2^2 + \prod_{k=1}^{3}(1-\chi_k^2)}} < \chi,
    \end{eqnarray*}
confirming the operation order dependence.

Given the parallel QNs where $S$ and $T$ share $K$ states $\sigma_{S^{(k)}T^{(k)}}\in\mathcal{D}_{\mathcal{C}}$ ($k=1,2,...,K$) [Fig.~\ref{fig-parallel_K}] and its normalized concentration $\Psi$. Let $E$ be an entanglement measure which can quantify the entanglement of each state in $\mathcal{D}_{\mathcal{C}}$. We now introduce the optimal concentration order for this QN. If, among all possible concentration orders, $\pi_{K}$ yields an output state $\Psi^{\pi_{K}}\left(\otimes_{k=1}^{K}\sigma_{S^{(k)}T^{(k)}}\right)$ [Eq.~\eqref{eq-concentration_1D}] whose entanglement quantified by $E$ is maximal—--that is, for any other concentration order $\pi_{K}'$ one has  $E\left[\Psi^{\pi_{K}}\left(\otimes_{k=1}^{K}\sigma_{S^{(k)}T^{(k)}}\right)\right] \geq E\left[\Psi^{\pi_{K}'}\left(\otimes_{k=1}^{K}\sigma_{S^{(k)}T^{(k)}}\right)\right]$, then $\pi_{K}$ is called the \textbf{optimal concentration order} for these $K$ states. The corresponding map $\Psi^{\pi_{K}}$ is then referred to as the \textbf{$\Psi$-based optimal entanglement concentration} for the $K$ states.

Consider CV case where $\sigma_{S^{(k)}T^{(k)}}=|\psi^{r_k}\rangle\langle\psi^{r_k}|$ ($k=1,2,...,K$) satisfying $r_1\geq r_k$ for all $k$. As shown in {Ref.~{\cite{NegPT}}}, the maximal output entanglement is achieved by any concentration order $\pi_K$ for $\pi_K(1)=1$.
Correspondingly, the optimal entanglement concentration map $\Psi_{ST}^{\rm{opt}}$ of these $K$ states results in
\begin{eqnarray*}
	\Psi_{ST}^{\rm{opt}}\left(\otimes_{k=1}^{K}\sigma_{S^{(k)}T^{(k)}}\right)
	=\Psi_{ST}^{\pi_K}\left(\otimes_{k=1}^{K}\sigma_{S^{(k)}T^{(k)}}\right)=|\psi^{r}\rangle\langle\psi^{r}|
\end{eqnarray*}
with
\begin{eqnarray*}
	\sinh r=\sinh r_1 \prod_{k=2}^{K}\cosh r_k,
\end{eqnarray*}
relying on the order $\pi_K$.
It follows that
\begin{eqnarray}\label{eq-para_chi_K}
	{\rm Para}_{\Psi,\chi_{\mathcal{N}}}^{\pi_K}\left[\sigma_{S^{(1)}T^{(1)}},...,\sigma_{S^{(K)}T^{(K)}}\right] =\frac{\chi_{\mathcal{N}}\left(\sigma_{S^{(1)}T^{(1)}}\right)}{\sqrt{\chi_{\mathcal{N}}^2\left(\sigma_{S^{(1)}T^{(1)}}\right)+\prod_{k}{\left[1-\chi_{\mathcal{N}}^2\left(\sigma_{S^{(k)}T^{(k)}}\right)\right]}}}.
\end{eqnarray}

In the following, we introduce the order properties of the entanglement concentration protocol in parallel QNs.

We call $\Psi$ is {\bf $E$-based order-independent} for $\mathcal{D}_{\mathcal{C}}$ if, for all $\sigma_{S^{(1)}T^{(1)}}$, $\sigma_{S^{(2)}T^{(2)}}$, $\sigma_{S^{(3)}T^{(3)}}\in\mathcal{D}_{\mathcal{C}}$, all concentration orders yield the same final entanglement---that is, the equality
\begin{eqnarray*}
    {\rm Para}_{\Psi,E}^{\pi_3}\left[\sigma_{S^{(1)}T^{(1)}},\sigma_{S^{(2)}T^{(2)}},\sigma_{S^{(3)}T^{(3)}}\right]
    ={\rm Para}_{\Psi,E}^{\pi_3'}\left[\sigma_{S^{(1)}T^{(1)}},\sigma_{S^{(2)}T^{(2)}},\sigma_{S^{(3)}T^{(3)}}\right]\nonumber\\
\end{eqnarray*}
holds for any two different concentration orders $\pi_3$ and $\pi_3'$.
We call $\Phi$ has the {\bf generalized $E$-based order independence} if, for any finite-size parallel QN, all concentration orders produce the same final entanglement quantified by $E$.

Similar to Theorem~\ref{thm-commutative} for swapping order, we obtain the following result:

\begin{theorem}
    If $\Psi$ is both $E$-based order-independent and $E$-preserving on $\mathcal{D}_{\mathcal{C}}$, then $\Psi$ has generalized $E$-based order independence.
\end{theorem}

While we have provided an example of an entanglement-measure-based order-dependent entanglement concentration map, we have so far been unable to identify a corresponding case for entanglement swapping map.

\subsection{Series-parallel networks}\label{sec-spQN}

As shown in Sec.~\ref{sec-s_p_operator}, the entire entanglement percolation procedure in series-parallel QNs is a continual process of network reduction that repeatedly employs entanglement swapping and concentration operations.
Notably, our previous analysis of operation orders in series and parallel QNs suggests that each simplification step is intimately tied to order independence of operations.
Specifically, when there exist entanglement concentrations that is not generalized $E$-order-independent, the manipulation of parallel states in simplification Step (${\rm I^1}$) requires consideration of order, as illustrated by the three states shared between nodes $R_1$ and $R_2$ in Fig.~\ref{fig-s_p}. Similarly, the same applies to Step (${\rm {II}^1}$). Consequently, Steps (${\rm I^1}$) and (${\rm {II}^1}$) fundamentally involve operations on sub-networks.

We will next partition the entire series-parallel network into modules based on the topological structure of its sub-networks, where each module itself is a series [Fig.~\ref{fig-s}] or parallel network [Fig.~\ref{fig-p}]. Before introducing this modular decomposition, we first define several canonical types of series-parallel sub-network structures.

Given a two-terminal network $\mathcal{N}$ with terminals $S$ and $T$ that is free of topological redundancy, suppose that $\mathcal{N}$ is neither a series network [Fig.~\ref{fig-s}] nor a parallel network [Fig.~\ref{fig-p}]. We now formally define the \emph{maximal series sub-networks} and the \emph{maximal parallel sub-networks} of $\mathcal{N}$.
\begin{definition}
    In the network $\mathcal{N}$, a sub-network $\mathcal{N}'$ is called a \textbf{series sub-network} if (i) $\mathcal{N}'$ itself is a series network of length greater than $1$, and (ii) every internal node (i.e., nodes other than the terminals) of $\mathcal{N}'$ is distinct from $S$ and $T$ and has degree $2$. Furthermore, if $\mathcal{N}'$ is not properly contained in any other series sub-network of $\mathcal{N}$, it is called a \textbf{maximal series sub-network}.
\end{definition}
\begin{definition}
    A sub-network $\mathcal{N}''$ of $\mathcal{N}$ is called a \textbf{parallel sub-network} if $\mathcal{N}''$ itself is a parallel network comprising at least two links. Furthermore, $\mathcal{N}''$ is called a \textbf{maximal parallel sub-network} if it is not properly contained in any other parallel sub-network of $\mathcal{N}$.
\end{definition}

For instance, in the QN of Fig.~\ref{fig-s_p}, there are two maximal series sub-networks $(R_3,R_4,R_5)$ and $(S,R_6,R_7)$, and two maximal parallel sub-networks: one formed by nodes $R_1$ and $R_2$ with all three states between them, and another by $R_7$ and $R_8$ with every states between them. By contrast, for examples, neither $(R_1,S,R_6)$ nor $(R_2,R_3,R_5,T)$ qualifies as a series sub-network of this QN; $(S,R_1)$ is neither a series sub-network nor a parallel sub-network of this QN; Nodes $R_1$ and $R_2$ and any two of the three states between them yield a parallel sub-network that is not maximal of this QN.

Under aforementioned concepts, for any series-parallel QN $\mathcal{Q}$ with two terminals $S$ and $T$, the network can be uniquely decomposed into a collection of maximal series sub-networks and maximal parallel sub-networks.
This implies that the two simplification steps, (${\rm I^1}$) and (${\rm {II}^1}$), then correspond respectively to applying entanglement concentration protocols on the maximal parallel sub-networks and entanglement swapping protocols on the maximal series sub-networks.

Based on the foregoing discussion, the simplification operations (${\rm I^1}$) and (${\rm {II}^1}$) for series-parallel QNs can be optimized as follows:

(${\rm {I}^2}$) {\bf{Parallel Simplification.}} For each maximal parallel sub-network, apply the corresponding optimal entanglement concentration protocol. Once no maximal parallel sub-network remains, then proceed to (${\rm {II}^2}$).

(${\rm {II}^2}$) {\bf{Series Simplification.}} For each maximal series sub-network, apply the corresponding optimal entanglement swapping protocol. Once no maximal series sub-network remains, then return to (${\rm {I}^2}$).

This raises the question of optimization sequence: which simplification procedure should be prioritized---parallel simplification ($\rm I^2$) or series simplification ($\rm II^2$)---to achieve higher final entanglement? We find that if both $\Phi$ and $\Psi$ exhibit generalized $E$-based order independence on $\mathcal{D}_{\mathcal{S}}$, then the sequence of executing (${\rm {I}^2}$) and (${\rm {II}^2}$) does not affect the final result.
However, if either $\Phi$ or $\Psi$ does not commute on $\mathcal{D}_{\mathcal{S}}$, then the execution order of (${\rm {I}^2}$) and (${\rm {II}^2}$) may impact the final outcome (detailed in Appendix~\ref{SI-simplification_sequence}).

These findings show that when entanglement swapping or entanglement concentration is not $E$-based order-independent for the state set $\mathcal{D}_{\mathcal{S}}$, constructing series–parallel QNs with states in $\mathcal{D}_{\mathcal{S}}$ requires not only selecting the operation order within each maximal serial or parallel sub-network, but also determining the sequence of the simplification operations (${\rm {I}^2}$) and (${\rm {II}^2}$) between sub-networks, which adds an extra layer of complexity to QN design.

\section{{Analysis of entanglement percolation in series-parallel networks}}

{Although QNs in practice are of finite scale, due to the arbitrariness of scale, when discussing entanglement percolation properties, they are often considered within the context of infinite-scale networks. For example, what patterns does entanglement percolation exhibit as the network scale gradually increases? What strategies for expanding an existing network can maintain high-quality entanglement percolation effects? Such questions can only be thoroughly studied within the framework of infinite networks. Therefore, in this section, we employ the operator-theoretic framework developed earlier to construct a qualitative framework for analyzing entanglement percolation, thereby deriving precise mathematical conditions that determine whether percolation succeeds in a scalable QN.}

To ensure effective entanglement \yaqi{distribution} (i.e., entanglement percolation) in QNs under asymptotic expansion, we focus our analysis on series QNs rather than parallel configurations. This prioritization stems from two fundamental considerations:
\begin{enumerate}[label=(\arabic*),itemsep=0pt,topsep=0pt,parsep=0pt]
  \item Resource Perspective: Entanglement swapping consumes entanglement sources while entanglement concentration enhances them. Consequently, series networks may exhibit complete entanglement degradation during infinite expansion, whereas parallel configurations inherently preserve entanglement.
  \item Network Design Constraint: Direct connections between infinitely distant nodes are physically unrealizable. Instead, QNs interconnecting such nodes deploy multiple parallelized infinite-length series paths. Successful entanglement percolation therefore requires independent percolation along every constituent series path.
\end{enumerate}

\subsection{Series networks}\label{sec-series}

In practical implementations, QNs are constructed by connecting local sub-networks end-to-end to form a chain structure.
If each sub-network is reduced--via entanglement \yaqi{distribution}--to two terminals sharing a single source state, the entire network simplifies to a series QN [Fig.~\ref{fig-seriN}]. We now investigate the properties under which such an infinite chain exhibits percolation, i.e., a nonzero long-range entanglement established across arbitrarily many nodes.

Let $\Phi:\mathcal{D}_{\mathcal{S}}\otimes\mathcal{D}_{\mathcal{S}}\to\mathcal{D}_{\mathcal{S}}$ be a removal-node entanglement swapping, $\Psi$ a normalized entanglement concentration, and $E$ a suitable entanglement measure to quantify states in $\mathcal{D}_{\mathcal{S}}$ which takes values in $[0,1]$, such as the $G$-concurrence $C_G$ for DV-based states and the ratio negativity $\chi_{\mathcal{N}}$ for CV-based states.
Consider a series QN of $N$ source states $\rho_{SR_1^{(1)}}$, $\rho_{R_1^{(2)}R_2^{(1)}}$,..., $\rho_{R_{N-1}^{(2)}T}\in\mathcal{D}_{\mathcal{S}}$ whose individual entanglement are $E_1$, $E_2$,..., $E_{N}$, respectively.
Denote by $\Phi^{\pi_{N-1}^{\rm{opt}}}$ the $\Phi$-based optimal entanglement swapping map on these $N$ states, then the entanglement of the final output state---termed the sponge-crossing entanglement---is
\begin{eqnarray*}
    E_{\rm{SC}}^{(N)}=E\left[\Phi^{\pi_{N-1}^{\rm{opt}}}\left(\rho_{SR_1^{(1)}}\otimes\rho_{R_1^{(2)}R_2^{(1)}}\otimes\cdots\otimes\rho_{R_{N-1}^{(2)}T}\right)\right].
\end{eqnarray*}
{In the limit $N\to\infty$, if the the sponge-crossing entanglement $E_{\rm{SC}}^{(N)}$ satisfies
\begin{eqnarray*}\label{eq-inferior}
    \varliminf _{N\to\infty}E_{\rm{SC}}^{(N)}:=\lim_{N\to\infty}\inf_{n\geq N}E_{\rm{SC}}^{(n)}>0,
\end{eqnarray*}
we say the infinite-size series QN exhibits entanglement percolation.
Conversely, if,
\begin{eqnarray*}\label{eq-superior}
    \overline{\lim\limits_{N\to\infty}}E_{\rm{SC}}^{(N)}:=\lim_{N\to\infty}\sup_{n\geq N}E_{\rm{SC}}^{(n)}=0,
\end{eqnarray*}
we say the infinite-size series QN fails to exhibit entanglement percolation.
Moreover, if $E_{\rm{SC}}^{(\infty)}=1$, i.e. the final state achieves maximal entanglement, we term this saturated percolation.}

If $\Phi$ is $E$-based order-independent, we have the following theorem:
\begin{theorem}
    If $\Phi$ is $E$-based order-independent and $E$ is an entanglement monotone, then $E_{\rm SC}^{(N)}$ is monotonically decreasing with respect to $N$ and the limit $\lim_{N\to\infty}E_{\rm SC}^{(N)}$ exists (i.e., $\varliminf _{N\to\infty}E_{\rm{SC}}^{(N)}=\overline\lim_{N\to\infty}E_{\rm{SC}}^{(N)}=\lim_{N\to\infty}E_{\rm SC}^{(N)}\geq 0$).
\end{theorem}
\begin{proof}
    For the $(N+1)$-state series QN with $N>2$, if $\Phi$ is $E$-based order-independent, then we have
    \begin{eqnarray*}
        E_{\rm SC}^{(N+1)}
        &=&E\left[\Phi_{R_N^{(1,2)}}^{ST}\left(\rho_{SR_N^{(1)}}^{\pi_{N-1}}\otimes\rho_{R_{N}^{(2)}T}\right)\right]
    \end{eqnarray*}
    where
    \begin{eqnarray*}
        \rho_{SR_N^{(1)}}^{\pi_{N-1}}=\Phi^{\pi_{N-1}}\left(\rho_{SR_1^{(1)}}\otimes\rho_{R_1^{(2)}R_2^{(1)}}\otimes\cdots\otimes\rho_{R_{N-1}^{(2)}R_N^{(1)}}\right)
    \end{eqnarray*}
    with any swapping order $\pi_{N-1}$.
    Then the series-rule inequality \eqref{ineq-Seri} yields that
    \begin{eqnarray*}
        E_{\rm SC}^{(N+1)}\leq E\left(\rho_{SR_N^{(1)}}^{\pi_{N-1}}\right)=E_{\rm SC}^{(N)},
    \end{eqnarray*}
    implying the sequence $\left\{E_{\rm SC}^{(N)}\right\}_N$ is de creasing.
\end{proof}

Consequently, we obtain the following fact:
\begin{corollary}
    For all $N$, if the final state is maximally entangled,
    \begin{eqnarray*}
      E_{\rm{SC}}^{(N)}=1,
    \end{eqnarray*}
    then every initial state is maximally entangled, i.e.,
    \begin{eqnarray*}
        E\left(\rho_{SR_1^{(1)}}\right)=E\left(\rho_{R_1^{(2)}R_2^{(1)}}\right)=\cdots=E\left(\rho_{R_{N-1}^{(2)}T}\right)=1.
    \end{eqnarray*}
\end{corollary}

Consider the series QN model satisfying the product structure
\begin{eqnarray}\label{eq-EN_product}
    E_{\rm SC}^{(N)}=\prod_{n=1}^N E_n
\end{eqnarray}
for some entanglement measure $E$ and all $N$~\cite{Ratio_negativity2024}.
{Then we obtain the following sufficient condition under which entanglement percolation fails in infinite-size series QNs:}

\begin{theorem}
    Suppose that $E_{\rm SC}^{(N)}=\prod_{n=1}^N E_n$ holds for all $N$, and that
    \begin{eqnarray}\label{eq-limit_EN}
      \overline\lim_{n\to\infty}E_n<1,
    \end{eqnarray}
    then we have the limit $\lim_{N\to\infty} E_{\rm SC}^{(N)}$ exists and satisfies $E_{\rm SC}^{(\infty)}:=\lim_{N\to\infty} E_{\rm SC}^{(N)}=0$.
\end{theorem}

{\begin{proof}
   Let $E_{\rm SC}^{(N)}=\prod_{n=1}^N E_n$ hold for all $N$ and $\alpha=\overline\lim_{n\to\infty}E_n<1$.
   Then for sufficient large $N_0$, we have $\sup_{n\geq N_0}E_{\rm{SC}}^{(n)}<(1+\alpha)/2$, which implies that
   \begin{eqnarray}\label{eq-N0}
       E_{\rm{SC}}^{(n)}<\frac{1+\alpha}{2}
   \end{eqnarray}
   for all $n\geq N_0$.
   Since $E_{\rm SC}^{(N)}=\prod_{n=1}^{N}E_{n}$,
   the sequence $\{E_{\rm SC}^{(N)}\}_N$ is non-increasing and bounded, and hence the limit $\lim\limits_{N\to\infty}E_{\rm SC}^{(N)}$ exists.
   Using Eq.~\eqref{eq-N0}, we obtain
   \begin{eqnarray*}
       E_{\rm SC}^{(\infty)}
       &=&\lim_{N\to\infty}\prod_{n=1}^{N}E_{n}\\ &=&\prod_{n=1}^{N_0}E_{n}\prod_{m=N_0+1}^{\infty}E_{m}\\
       &<&\prod_{n=1}^{N_0}E_{n}\prod_{m=N_0+1}^{\infty}\frac{1+\alpha}{2}\\
       &=&0,
   \end{eqnarray*}
   where the last equality follows from $\frac{1+\alpha}{2} < 1$.
\end{proof}}

{Notably, both DV-based QNs composed of two-qubit states and CV-based QNs composed of TMSVSs satisfy the product structure in Eq.~\eqref{eq-EN_product} (see the series rules in Eqs.~\eqref{eq-swap_1D_concurrence} and~\eqref{eq-seri_chi_N}). Therefore, the following results follow directly:}

\begin{corollary}\label{th-seri_c}
    Consider that $E$ is the concurrence $c$ and aforementioned series QN is a DV-based QN of pure two-qubit states with the concurrence $c_1,c_2,...,c_N$. If the upper limit of the sequence $\{c_n\}_{n=1}^N$ satisfies
    \begin{eqnarray*}
        \overline\lim_{n\to\infty}c_n<1
    \end{eqnarray*}
    for $N\to\infty$, then $c_{\rm SC}^{(\infty)}=0$.
\end{corollary}

Similar conclusions apply to the CV-based series QN from the series rule ${\rm Seri}_{\Phi,\chi_{\mathcal{N}}}$ in Eq.~\eqref{eq-seri_chi}:
\begin{corollary}
    Consider $E$ as the ratio negativity $\chi_{\mathcal{N}}$ and the series QN as a CV-based QN of TMSVSs with the ratio negativity $\chi_{n}$. If
    \begin{eqnarray*}
        \overline\lim_{n\to\infty}\chi_{n}<1
    \end{eqnarray*}
    for $N\to\infty$, then $\chi_{\mathcal{N},\rm{SC}}^{(\infty)}=0$.
\end{corollary}

Nevertheless, the criterion in Eq.~\eqref{eq-limit_EN} is only a necessary yet insufficient condition.
A counterexample is the following.
\begin{theorem}\label{th-lp}
    If each initial entanglement $E_n$ in Eq.~\eqref{eq-EN_product} satisfies
    \begin{eqnarray}\label{eq-lp}
     E_n\geq1-{1}/{(n+1)^p},
    \end{eqnarray}
    for $p>1$, then we have $E_{\rm SC}^{(\infty)}>0$; Conversely, if
    \begin{eqnarray*}
     E_n\leq1-{1}/{(n+1)},
    \end{eqnarray*}
    then $E_{\rm SC}^{(\infty)}=0$.
\end{theorem}
\begin{proof}
    Let $E_n=1-{1}/{(n+1)^p}$, then
    \begin{eqnarray*}
        E_{\rm SC}^{(N)}=\exp\left\{\sum_{n=1}^{N}\ln{\left[1-{(n+1)^{-p}}\right]}\right\}
    \end{eqnarray*}
    Since the limit
    \begin{eqnarray*}
        \lim_{n\to\infty}\frac{-\ln{\left[1-{(n+1)^{-p}}\right]}}{{(n+1)^{-p}}}=1,
    \end{eqnarray*}
    and the sum $\sum_{n=1}^{\infty}{(n+1)^{-p}}$ converges for $p>1$ but diverges for $0<p\leq 1$, the sum $\sum_{n=1}^{N}\ln{\left[1-{(n+1)^{-p}}\right]}$ converges or diverges accordingly by the limit comparison test.
    Consequently, if $p>1$ and $E_n\geq 1-{1}/{(n+1)^p}$ for all $n$, then $E_{\rm SC}^{(\infty)}>0$.
    In contrast, if $E_n\leq 1-{1}/{(n+1)}$ for all $n$, then $E_{\rm SC}^{(\infty)}=0$.
\end{proof}
For example, when $p=2$, we have $E_{\rm SC}^{(\infty)}=1/2$.

For the QN models satisfying Eq.~\eqref{eq-EN_product}, assume the entanglement of each source state is $E_0$, then the sponge-crossing entanglement takes the form
\begin{eqnarray}\label{eq-eqWeight_series}
    E_{\rm{SC}}^{(N)}=E_0^N,
\end{eqnarray}
indicating that the overall entanglement decreases exponentially with increasing $N$ unless $E_0=1$.
In the asymptotic limit $N\to\infty$, $E_{\rm{SC}}^{(N)}$ remains strictly positive only when the initial entanglement is maximal, i.e., $E_0=1$, while for any $E_0<0$, it vanishes.
For instance, if $E_0<0.95$, the value of $E_{\rm{SC}}^{(N)}$ drops below 0.01 at $N=100$ [Fig.~\ref{fig-E0_ESC}], and even for $E_0=0.99$, it decreases below 0.01 once $N>458$ [Fig.~\ref{fig-N_E_E099}].
While for $N>458$ and $E_0<0.99$, the entanglement percolation almost can not be realize.
This highlights the fragility of large-scale entanglement under imperfect values of initial entanglement.

\begin{figure}
    \centering
    \subfigure[]{
    \includegraphics[width=150pt]{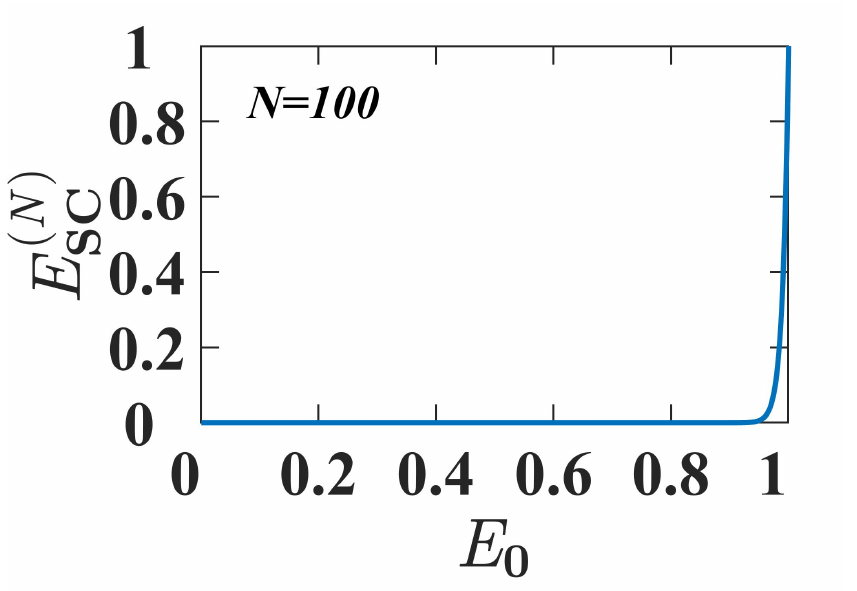}
    \label{fig-E0_ESC}
    }
    \subfigure[]{
    \includegraphics[width=150pt]{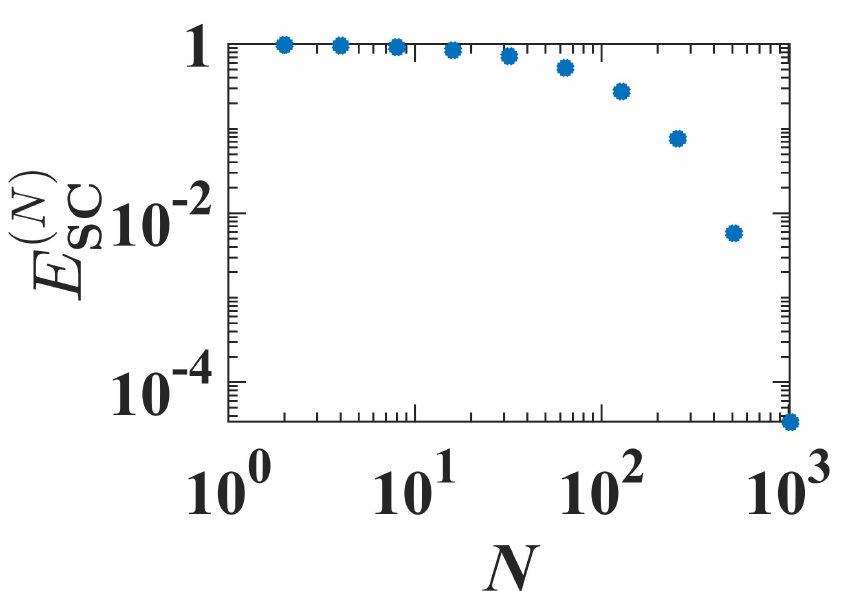}
    \label{fig-N_E_E099}
    }
    \caption{Sponge-crossing entanglement $E_{\rm{SC}}^{(N)}$ in series QN with initially identical resource-state entanglement $E_0$.
    \subref{fig-E0_ESC}~$E_{\rm{SC}}^{(N)}$ versus initial entanglement $E_0$ for $N=100$.
    \subref{fig-N_E_E099}~$E_{\rm{SC}}^{(N)}$ versus $N$ at fixed initial entanglement $E_0=0.99$.
    }
\end{figure}

For such QNs, to establish desired entanglement, it is found that:
\begin{theorem}
    In a finite-size series QN {with $N<\infty$} where each source state has entanglement $E_0$, achieving the desired entanglement target $E_{\rm{SC}}^{\rm{target}}$ requires
\begin{eqnarray*}
    E_0\geq \left[E_{\rm{SC}}^{\rm{target}}\right]^{1/N}.
\end{eqnarray*}
\end{theorem}

{In the DV case, a two-qubit pure state can attain maximal entanglement, i.e., unit concurrence. As a result, by choosing each initial link to be a maximally entangled two-qubit state, one obtains $E_n=1$ for all $n$, which directly yields $E_{\rm SC}^{(\infty)}=1$. Therefore, entanglement percolation over an infinitely extended series QN can be trivially achieved.
In contrast, for CV systems, a two-mode squeezed vacuum state (TMSVS) cannot reach maximal entanglement. Specifically, achieving $\chi_{\mathcal{N}}=1$ would require infinite squeezing, which is unphysical. Consequently, it is impossible to realize entanglement percolation in an infinite series QN simply by replacing each link with a maximally entangled CV state.
}

{
Nevertheless, the first statement of Theorem~\ref{th-lp} reveals a nontrivial alternative: although no individual TMSVS can be maximally entangled, one can still enable entanglement percolation in CV-based series QNs by appropriately engineering the resource states. In particular, by choosing a sequence ${\chi_n}$ that approaches unity sufficiently fast, the overall entanglement $\chi_{\mathcal{N},\mathrm{SC}}^{(\infty)}$ can remain strictly positive, thereby making infinite-range entanglement percolation feasible.
This result indicates that, although the TMSVS cannot achieve maximal entanglement---i.e., unit ratio negativity $\chi_{\mathcal{N}}=1$, which would require infinite squeezing---one can still enable entanglement percolation in an infinitely extended series network by appropriately engineering the resource states. In particular, by choosing a sequence of TMSVSs whose values ${\chi_n}$ of ratio negativity approach unity sufficiently fast (e.g., satisfying the conditions in Ineq.~\eqref{eq-lp}), the overall entanglement $\chi_{\mathcal{N},\mathrm{SC}}^{(\infty)}$ remains strictly positive.}

{\subsection{Parallel-then-series networks}

{The preceding section showed that series QNs satisfying Eq.~\eqref{eq-EN_product} cannot achieve entanglement percolation when all source states are identical non-maximally entangled states.
To overcome this limitation, we enhance entanglement percolation by increasing the network width. Specifically, instead of a single entangled state between each pair of adjacent nodes, we allow multiple entangled states to be shared in parallel. We define the \emph{network width} as the minimum number of parallel states shared across all adjacent node pairs. This modification transforms a purely series QN into a parallel-then-series architecture.
In the following, we show that this enhancement---by enabling more effective entanglement concentration---opens the possibility of sustaining entanglement percolation even in the infinite-size limit.}

Given a parallel-then-series QN [Fig.~\ref{fig-p_t_s}], {we have}: (1) $S$ and $R_1$ share $K_1$ states $\rho_{S^{(k_1)}R_1^{(k_1)}}$, $k_1=1,2,...,K_1$; (2) For all $j=2,3,...,N-1$, $R_{j-1}$ and $R_{j}$ share $K_{j}$ states $\rho_{R_{j-1}^{(K_{j-1}+k_j)}R_j^{(k_j)}}$, $k_j=1,2,..,K_j$; (3) $R_{N-1}$ and $T$ share $K_N$ states $\rho_{R_{N-1}^{(K_{N-1}+k_N)}T^{(k_N)}}$, $k_N=1,..,K_N$.
The network width $K$ is
$$K=\min_{n}K_n.$$
This QN can be decomposed as $N$ sub-networks, {each of which consists} of two adjacent nodes and shares all states between them.
Let $\Phi$ be the removal-node entanglement swapping map and $\Psi$ be the normalized entanglement concentration map.
Then by applying the simplification (${\rm {I}^2}$), this QN is mapped to a series QN of $N$ states, $\Psi^{\pi_{K_1}}_{SR_1}\left(\otimes_{k_1=1}^{K_1}\rho_{S^{(k_1)}R_1^{(k_1)}}\right)$, $\Psi^{\pi_{K_2}}_{R_1R_2}\left(\otimes_{k_2=1}^{K_2}\rho_{R_1^{(K_1+k_2)}R_2^{(k_2)}}\right)$,..., $\Psi^{\pi_{K_N}}_{R_{N-1}T}\left(\otimes_{k_N=1}^{K_N}\rho_{R_{N-1}^{(K_{N-1}+k_N)}T^{(k_N)}}\right)$, with corresponding optimal concentration orders {$\pi_{K_n}$}, $n=1,2,...,N$.
Then via (${\rm {II}^2}$), under the optimal swapping order $\pi_{N-1}$ of the series QN with $N$ states, we obtain the final state
\begin{eqnarray*}
    \rho_{ST}= \Phi^{\pi_{N-1}}&\Bigg[&\Psi^{\pi_{K_1}}_{SR_1}\left(\otimes_{k_1=1}^{K_1}\rho_{S^{(k_1)}R_1^{(k_1)}}\right) \otimes\Psi^{\pi_{K_2}}_{R_1R_2}\left(\otimes_{k_2=1}^{K_2}\rho_{R_1^{(K_1+k_2)}R_2^{(k_2)}}\right)\otimes\cdots\nonumber\\
    &\otimes&\Psi^{\pi_{K_N}}_{R_{N-1}T}\left(\otimes_{k_N=1}^{K_N}\rho_{R_{N-1}^{(K_{N-1}+k_N)}T^{(k_N)}}\right)\Bigg].
\end{eqnarray*}

For specific QNs, we obtain the following examples:

\begin{theorem}\label{th_DV_QNs} ({\bf DV-based QNs}) Let the aforementioned QN be DV-based of pure two-qubit states and the concurrence of the $k_n$-th state in the initial $n$-th sub-network be $c_{n,k_n}$. Then the final state is $\rho_{ST}=|\psi_{\lambda}\rangle\langle\psi_\lambda|$ with the concurrence
    \begin{eqnarray*}
        c_{\rm SC}^{(N)}=2\sqrt{\lambda(1-\lambda)}=\prod_{n=1}^N c_n,
    \end{eqnarray*}
    where $c_n$ denotes the entanglement of the output state obtained via entanglement concentration over the $K_n$ states in the $n$-th sub-network,
    \begin{eqnarray}\label{eq_stp_cn}
        \frac{1+\sqrt{1-c_n^2}}{2}=\max\left\{\frac{1}{2},\prod_{k_n=1}^{K_n}\frac{1+\sqrt{1-c_{n,k_n}^2}}{2}\right\}.
    \end{eqnarray}
    \end{theorem}
\begin{proof}
    Using parallel simplification (${\rm {I}^2}$) on the initial QN, we obtain an $N$-state series QN where the concurrence $c_n$ of the $n$-th state satisfies Eq.~\eqref{eq_stp_cn} governed by the parallel rule in Eq.~\eqref{eq-para_c_K}.
    Subsequent series simplification (${\rm {II}^2}$), from the series rule Eq.~\eqref{eq-swap_1D_concurrence}, yields the final concurrence $c_{\rm SC}^{(N)}=\prod_{n=1}^{N}c_n$.
\end{proof}

Given that the QN reduces to a series configuration after the first simplification step, Theorem~\ref{th-seri_c} provides a sufficient condition for the absence of entanglement percolation in parallel-then-series QNs {in the asymptotic limit as} $N\to\infty$.
\begin{corollary}
    For the limit $N\to\infty$, if $c_n$ in Eq.~\eqref{eq_stp_cn} satisfies $\overline\lim_{n\to\infty}c_n<1$, then $\lim_{N\to\infty}c_{\rm SC}^{(N)}=0$.
\end{corollary}

Now
{let us consider the scenario in which all sub-networks are identical and determine the network  width $K$ required for an infinite series network to percolate the entanglement successfully.}

\begin{corollary}\label{coro_lim_cn}
For a DV-based QN as in Theorem~\ref{th_DV_QNs} with all sub-networks are identical.  Assume that $K_n=K$ for each $n$ and denote $c_{n,k_n}=c$. Then {we have
\begin{eqnarray*}
    c_{\rm SC}^{(N)}= c_1^N,
\end{eqnarray*}
where $c_1$ satisfies
\begin{eqnarray*}
  \frac{1+\sqrt{1-c_1^2}}{2}=\max\left\{\frac{1}{2},\left[\frac{1+\sqrt{1-c^2}}{2}\right]^K\right\}.
\end{eqnarray*}}
In addition, let
$$c_{\rm th}(K):=\sqrt{2^{(2K-1)/K}-2^{(2K-2)/K}}.$$ The following statements are true:

{\rm (1)} If the scale of the network is finite ($N<\infty$), then $C_{\rm SC}^{(N)}(c)$ is continuous for $c\in[0,1]$, and $0<C_{\rm SC}^{(N)}<1$ for $0< c<c_{\rm th}(K)$ while $C_{\rm SC}^{(N)}=1$ for $c\geq c_{\rm th}(K)$.

{\rm (2)} If the scale of the network {is} infinite ($N=\infty$), then $C_{\rm SC}^{(\infty)}(c)$ is not continuous at $c=c_{\rm th}(K)$, and $C_{\rm SC}^{(\infty)}=0$ for $c\leq c_{\rm th}(K)$ while $C_{\rm SC}^{(\infty)}=1$ for $c\geq c_{\rm th}(K)$.
\end{corollary}

It is clear by Corollary~\ref{coro_lim_cn}, the DV-based parallel-then-series {QNs} with width $K\geq 2$ have good entanglement percolation property.
The case when the networks are finite-size with $N=10,100,1000$ is shown in Fig.~\ref{fig-s_t_p_finite}.
For the case when  the network is infinite-size {($N=\infty$)},
assume that $K=2$, then the threshold is $c_{\rm th}(2) = \sqrt{(2\sqrt{2} - 2)} \approx 0.91$ and the behaviour of entanglement percolation $C_{\rm SC}^{(\infty)}$ is shown in Fig.~\ref{fig-s_t_p_infinite}. As $C_{\rm SC}^{(N)}(c)=1$, attaining the maximal entanglement, whenever $c>c_{\rm th}(K)$, the threshold $c_{\rm th}(K)$ is also called the \emph{saturation point} of entanglement percolation.

From Corollary~\ref{coro_lim_cn},  the threshold $c_{\rm th}(K)$ admits a clear qualitative interpretation in terms of the network width $K$. Specifically, entanglement percolation in the infinite-size limit occurs if and only if $c> c_{\rm th}(K)$. The threshold $c_{\rm th}(K)$ decreases monotonically with increasing $K$. In particular, as $K\to\infty$, we have $c_{\mathrm{th}}(K)\to 0$, implying that even very weakly entangled initial states can support entanglement percolation, provided that the network width is sufficiently large.

{Equivalently, for a given initial concurrence $c$, there exists a \emph{critical width}
\begin{eqnarray*}
    K_c:=\left\lceil\left(\log_{2}\frac{2}{1+\sqrt{1-c^2}}\right)^{-1}\right\rceil,
\end{eqnarray*}
such that entanglement percolation is achievable {for all $N$} if $K \geq K_c$. From the asymptotic relation above, this critical width scales as
$K_c\sim {2\ln 2}/{c^2}$.
This result shows that the parallel-then-series architecture fundamentally alters the percolation behavior: while uniform series QNs prohibit infinite-range entanglement distribution for any $c<1$, introducing a finite but sufficiently large parallel width enables deterministic entanglement percolation even with non-maximally entangled resources.}

\begin{figure}
    \centering
    \subfigure[]{
    \includegraphics[width=150pt]{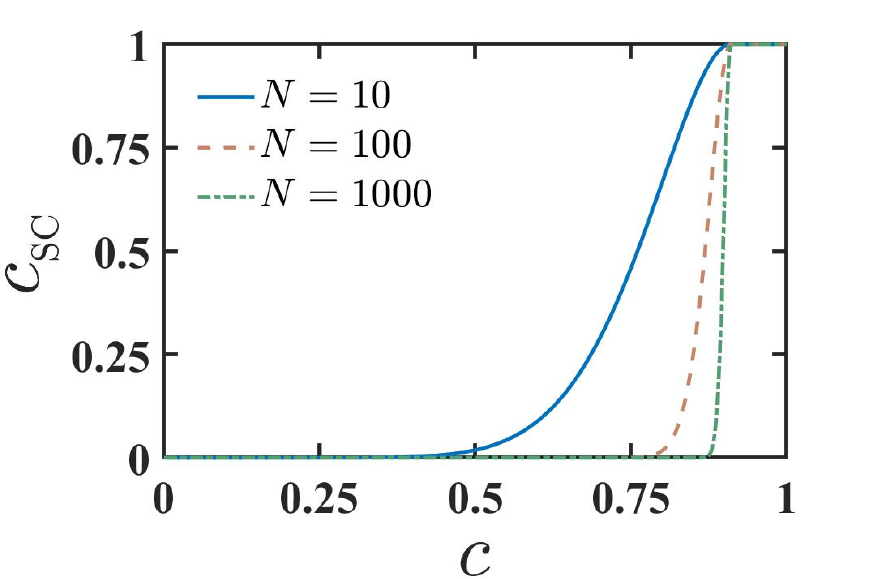}
    \label{fig-s_t_p_finite}
    }
    \subfigure[]{
    \includegraphics[width=150pt]{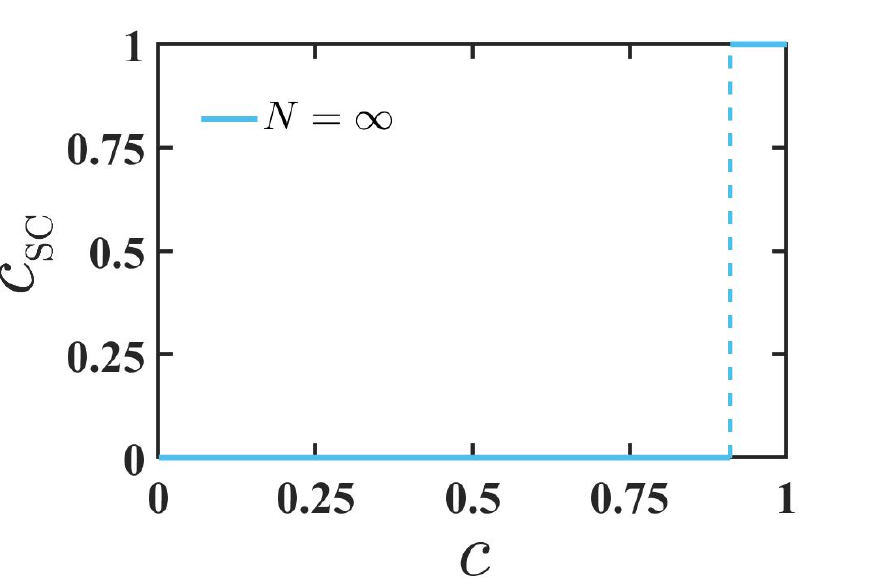}
     \label{fig-s_t_p_infinite}
    }
    \caption{Sponge-crossing concurrence $c_{\rm SC}=c_1^N$ for parallel-then-series networks.
    ~\subref{fig-s_t_p_finite} Finite-size networks with $N=10,100,1000$: $c_1<1$ for $c<c_{\rm th}(2)$ and $c_1=1$ for $c\ge c_{\rm th}(2)$ where $c_{\rm th}(2) = \sqrt{(2\sqrt{2} - 2)} \approx 0.91$.
    ~\subref{fig-s_t_p_infinite} Infinite-size networks ($N\to\infty$) for $K=2$: $c_{\rm SC}=0$ for $c<c_{\rm th}(2)$ and $c_{\rm SC}=1$ for $c\ge c_{\rm th}(2)$.}
    \label{fig-s_t_p_example}
\end{figure}

If the source states are TMSVSs, via the series-parallel rules in Eqs.~\eqref{eq-seri_chi_N}~and~\eqref{eq-para_chi_K},  the following result is easily checked:

\begin{theorem} (\textbf{CV-based QNs})
    Let the parallel-then-series QN be a finite-size CV-based QN {composed} of TMSVSs and the ratio negativity of the $k_n$-th state in the initial $n$-th sub-network be $\chi_{n,k_n}$. Then the final state is a TMSVS $\rho_{ST}=|\psi^r\rangle\langle\psi^r|$ with the ratio negativity
    \begin{eqnarray*}
        \chi_{\rm SC}^{(N)}=\tanh r=\prod_{n=1}^{N}\chi_{n}
    \end{eqnarray*}
    where
    \begin{eqnarray}\label{eq-pts_chi_n}
        \chi_n=\frac{\max\limits_{1\leq k_n\leq K_n}\chi_{n,k_n}}{\sqrt{\max\limits_{1\leq k_n\leq K_n}\chi_{n,k_n}^2+\prod_{k_n=1}^{K_n}(1-\chi_{n,k_n}^2)}}
    \end{eqnarray}
    is the ratio negativity of the $n$-th state obtained by concentrating $K_n$ states in the $n$-th sub-network via entanglement concentration.
\end{theorem}

Since any TMSVS cannot undergo infinite squeezing (i.e., the squeezing parameter cannot approach infinity), its ratio negativity cannot reach unity. Consequently, it follows that multiple TMSVSs cannot be concentrated into a perfectly entangled state via entanglement concentration, implying that the saturation {point and the critical width} cannot be defined for identical TMSVSs.
Nevertheless, we can design {parallel-then-series} CV-based QNs that enable entanglement percolation.

Specifically, we have

\begin{corollary}\label{coro_lim_chi} Assume all source states in the aforementioned QN be the TMSVS $|\psi^r_0\rangle$ with the ratio negativity $\chi_0=\tanh r_0>0$ and $K_n=n$ for all $n$. For $N\to\infty$, via the simplifications ${\rm ({I}^2)}$ and ${\rm ({II}^2)}$, the final ratio negativity $\chi_{\rm SC}^{(\infty)}$ established between $S$ and $T$ is non-zero.
\end{corollary}

\begin{proof}
Substituting  $\chi_{n,k_n}=\chi_0$ into Eq.~\eqref{eq-pts_chi_n} yields $\chi_n={\chi_0}/{\sqrt{\chi_0^2+\left(1-\chi_0^2\right)^n}}$, and consequently,
    \begin{eqnarray*}
        \chi_{\rm SC}^{(N)}=\exp{\sum_{n=1}^{N}\ln{\chi_n}}=\exp{\sum_{n=1}^{N}\ln{\frac{\chi_0}{\sqrt{\chi_0^2+\left(1-\chi_0^2\right)^n}}}}.
    \end{eqnarray*}
    Since $\lim_{n\to\infty}{\ln {\chi_n}}/{\left(1-\chi_0^2\right)^{n}}=-({2\chi_0^2})^{-1}$, the two series $\sum_{n=1}^{\infty}\ln{\chi_n}$ and $\sum_{n=1}^{\infty}\left(1-\chi_0^2\right)^{n}$ converge or diverge simultaneously.
    Since $\sum_{n=1}^{\infty}\left(1-\chi_0^2\right)^{n}=1/\chi_0^2-1$ converges for $0<\chi_0<1$, $\sum_{n=1}^{\infty}\ln{\chi_n}$ likewise converges, implying $\chi_{\rm SC}^{(\infty)}=\exp{\sum_{n=1}^{\infty}\ln{\chi_n}}>0$.
\end{proof}
{For instance, if $\chi_0 = 0.5$, then $\chi_{\rm SC}^{(N)}$ (Fig.~\ref{fig-XSC_N}) asymptotically approaches 0.27 as $N$ becomes large.}

\begin{figure}
    \centering
    \includegraphics[width=200pt]{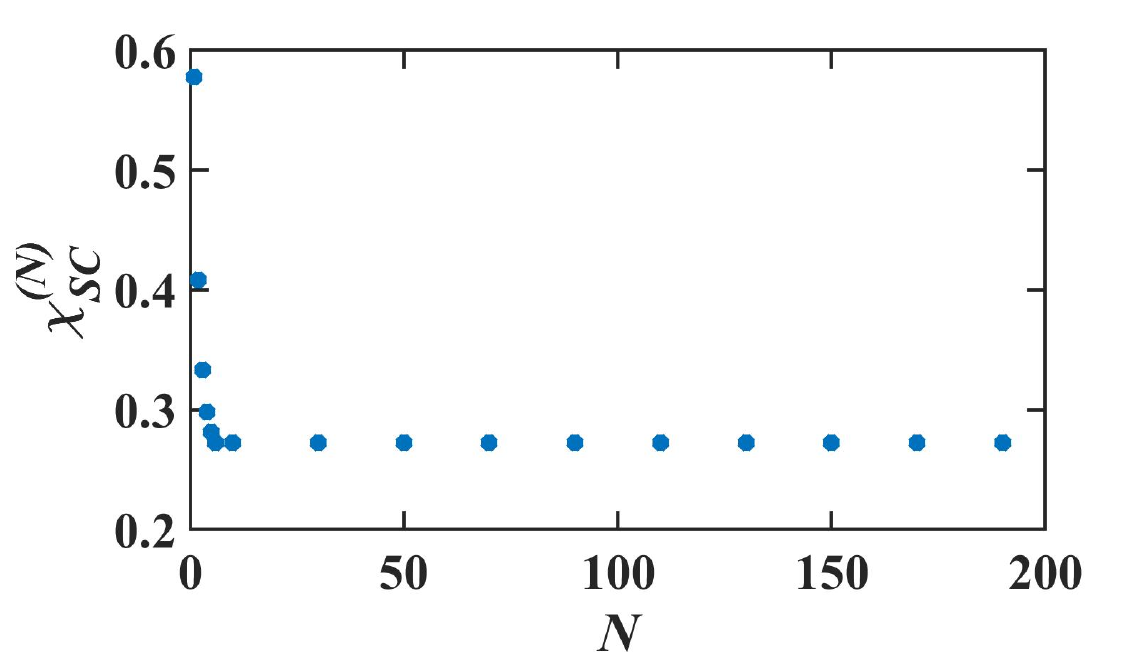}
    \caption{Sponge-crossing ratio negativity $\chi_{\rm SC}^{(N)}$. For sufficient large $N$, $\chi_{\rm SC}^{(N)}\approx0.272$.  }
    \label{fig-XSC_N}
\end{figure}

In summary, by Corollaries~\ref{coro_lim_cn}~and~\ref{coro_lim_chi}, achieving entanglement percolation between infinitely distant nodes requires supplementing parallel source states. This approach enhances entanglement through concentration while counteracting the operational resource consumption inherent to entanglement swapping. The corollaries also suggest that, in practical scenarios,
 if we need to expand a {QN} while maintaining the high-quality entanglement percolation of the network, in the {DV} system case, each expansion requires adding parallel a subnetwork with width $K \geq K_c=\left\lceil\left(\log_{2}\frac{2}{1+\sqrt{1-c^2}}\right)^{-1}\right\rceil$; whereas in the {CV} system case, although there is no constant similar to $K_c$ above, adding parallel network with width $K_n \geq n$ is sufficient for the $n$-th expansion. This finding a provides helpful guidance for the design and construction of {QNs}.}

\section{Conclusion and discussion}

This work establishes an operator-theoretic framework for entanglement \yaqi{distribution} and percolation in QNs under deterministic schemes. We first investigate the fundamental properties of deterministic entanglement swapping and entanglement concentration maps. Subsequently, for the maximal sub-network decomposition of series-parallel QNs, we systematically analyze two classes of operational sequences: internal operation orders (entanglement swapping sequences and concentration sequences) and the order of operations between different sub-networks (simplifications (${\rm {I}^2}$) and (${\rm {II}^2}$)).
Building on this analysis, we derive sufficient conditions for achieving entanglement percolation in specific infinite-range QNs, providing theoretical support for unlimited network expansion.

Nevertheless, there are several open questions:
\begin{enumerate}
    \item Scalable order-dependent protocols: Do deterministic entanglement swapping protocols exist that are simultaneously scalable and entanglement-measure-based order-dependent?
    \item Universal topological conditions: How can topological conditions {be established for} general uniform scalable series-parallel QNs ensure entanglement percolation during infinite expansion?
    \item Novel \yaqi{distribution} architectures: Do {there exist} QN architectures, {together} with corresponding DET schemes, beyond the current framework?
    \item Hybrid quantum states: Do {there exist} deterministic entanglement \yaqi{distribution} schemes for QNs with heterogeneous source states (e.g., coexisting pure and mixed states)?
    \item It should be noted that our definition of entanglement percolation for series QNs in Sec.~\ref{sec-series} remains incomplete. For instance, when both $\overline\lim_{n\to \infty}E_{\rm SC}^{(n)}>0$ and $\varliminf_{n\to \infty}E_{\rm SC}^{(n)}=0$ hold--certain subsequences converge to positive values while others converge to zero--how should such percolation behavior be rigorously defined?
\end{enumerate}


\newpage
\clearpage

\section*{APPENDIX}
\appendix
\addcontentsline{toc}{section}{Appendix}

\renewcommand{\appendixname}{Appendix}
\renewcommand{\thesubsection}{\Roman{subsection}}
\setcounter{equation}{0}
\renewcommand{\theequation}{\Roman{subsection}.\arabic{equation}}

\setcounter{figure}{0}
\renewcommand{\thefigure}{\Roman{subsection}.\arabic{figure}}
\makeatletter
\renewcommand{\p@subfigure}{\thefigure}
\makeatother

\subsection{Proof of Theorem~\ref{thm-commutative}}\label{SI-commutative}

Assume that $\Phi$ is both $E$-based order-independent and $E$-preserving, we now prove that
\begin{eqnarray*}
    {\rm Seri}_{\Phi,E}^{\pi_{N-1}}\left[\rho_{SR_1^{(1)}},\cdots,\rho_{R_{N-1}^{(2)}T}\right]
    ={\rm Seri}_{\Phi,E}^{(1,2,...,N-1)}\left[\rho_{SR_1^{(1)}},\cdots,\rho_{R_{N-1}^{(2)}T}\right].
\end{eqnarray*}
holds true for all $N>3$ and any swapping order $\pi_{N-1}$ by mathematical induction.

(1) Base case ($N = 4$): Consider all permutations of $(1,2,3)$: $(1,3,2)$, $(2,1,3)$, $(2,3,1)$, $(3,1,2)$, and $(3,2,1)$. Since $\Phi$ is $E$-based order-independent, we have:
\begin{align*}
\mathrm{Seri}_{\Phi,E}^{(1,3,2)}\left[\rho_{SR_1^{(1)}}, \rho_{R_1^{(2)}R_2^{(1)}}, \rho_{R_2^{(2)}R_3^{(1)}}, \rho_{R_3^{(2)}T}\right] &= \mathrm{Seri}_{\Phi,E}^{(1,2,3)}\left[\rho_{SR_1^{(1)}}, \rho_{R_1^{(2)}R_2^{(1)}}, \rho_{R_2^{(2)}R_3^{(1)}}, \rho_{R_3^{(2)}T}\right], \\
\mathrm{Seri}_{\Phi,E}^{(2,3,1)}\left[\rho_{SR_1^{(1)}}, \rho_{R_1^{(2)}R_2^{(1)}}, \rho_{R_2^{(2)}R_3^{(1)}}, \rho_{R_3^{(2)}T}\right] &= \mathrm{Seri}_{\Phi,E}^{(2,1,3)}\left[\rho_{SR_1^{(1)}}, \rho_{R_1^{(2)}R_2^{(1)}}, \rho_{R_2^{(2)}R_3^{(1)}}, \rho_{R_3^{(2)}T}\right], \\
\mathrm{Seri}_{\Phi,E}^{(3,2,1)}\left[\rho_{SR_1^{(1)}}, \rho_{R_1^{(2)}R_2^{(1)}}, \rho_{R_2^{(2)}R_3^{(1)}}, \rho_{R_3^{(2)}T}\right] &= \mathrm{Seri}_{\Phi,E}^{(3,1,2)}\left[\rho_{SR_1^{(1)}}, \rho_{R_1^{(2)}R_2^{(1)}}, \rho_{R_2^{(2)}R_3^{(1)}}, \rho_{R_3^{(2)}T}\right] \\
&= \mathrm{Seri}_{\Phi,E}^{(1,3,2)}\left[\rho_{SR_1^{(1)}}, \rho_{R_1^{(2)}R_2^{(1)}}, \rho_{R_2^{(2)}R_3^{(1)}}, \rho_{R_3^{(2)}T}\right] \\
&= \mathrm{Seri}_{\Phi,E}^{(1,2,3)}\left[\rho_{SR_1^{(1)}}, \rho_{R_1^{(2)}R_2^{(1)}}, \rho_{R_2^{(2)}R_3^{(1)}}, \rho_{R_3^{(2)}T}\right].
\end{align*}
Additionally:
\begin{align*}
\mathrm{Seri}_{\Phi,E}^{(2,1)}\left[\rho_{SR_1^{(1)}}, \rho_{R_1^{(2)}R_2^{(1)}}, \rho_{R_2^{(2)}R_3^{(1)}}\right] = \mathrm{Seri}_{\Phi,E}^{(1,2)}\left[\rho_{SR_1^{(1)}}, \rho_{R_1^{(2)}R_2^{(1)}}, \rho_{R_2^{(2)}R_3^{(1)}}\right].
\end{align*}

Since $\Phi$ is $E$-preserving, we obtain:
\begin{align*}
\mathrm{Seri}_{\Phi,E}^{(2,1,3)}\left[\rho_{SR_1^{(1)}}, \rho_{R_1^{(2)}R_2^{(1)}}, \rho_{R_2^{(2)}R_3^{(1)}}, \rho_{R_3^{(2)}T}\right] = \mathrm{Seri}_{\Phi,E}^{(1,2,3)}\left[\rho_{SR_1^{(1)}}, \rho_{R_1^{(2)}R_2^{(1)}}, \rho_{R_2^{(2)}R_3^{(1)}}, \rho_{R_3^{(2)}T}\right].
\end{align*}
Thus, for any permutation $\pi_3$ of $(1,2,3)$:
\begin{align*}
\mathrm{Seri}_{\Phi,E}^{\pi_3}\left[\rho_{SR_1^{(1)}}, \rho_{R_1^{(2)}R_2^{(1)}}, \rho_{R_2^{(2)}R_3^{(1)}}, \rho_{R_3^{(2)}T}\right] = \mathrm{Seri}_{\Phi,E}^{(1,2,3)}\left[\rho_{SR_1^{(1)}}, \rho_{R_1^{(2)}R_2^{(1)}}, \rho_{R_2^{(2)}R_3^{(1)}}, \rho_{R_3^{(2)}T}\right].
\end{align*}

(2) Induction step ($N = m+1$): Assume the result holds for all $N \leq m$. For any swapping sequence $\pi_m$, we consider two cases:

Case 1: $\pi_m(m) = m$.
By the induction hypothesis, for any swapping sequence $\pi_{m-1}$:
\begin{align*}
&\mathrm{Seri}^{\pi_{m-1}}\left[\rho_{SR_1^{(1)}}, \rho_{R_1^{(2)}R_2^{(1)}}, \dots, \rho_{R_{m-1}^{(2)}R_m^{(1)}}\right]\\
= &\mathrm{Seri}^{\pi_{m-1}'}\left[\rho_{SR_1^{(1)}}, \rho_{R_1^{(2)}R_2^{(1)}}, \dots, \rho_{R_{m-1}^{(2)}R_m^{(1)}}\right],
\end{align*}
where $\pi_{m-1}'=(1,2,\dots,m-1)$.
Since $\Phi$ is $E$-preserving:
\begin{align*}
&\mathrm{Seri}_{\Phi,E}^{\pi_m}\left[\rho_{SR_1^{(1)}}, \rho_{R_1^{(2)}R_2^{(1)}}, \dots, \rho_{R_m^{(2)}T}\right] \\
=& \mathrm{Seri}_{\Phi,E} \left[ \Phi^{\pi_{m-1}} \left( \rho_{SR_1^{(1)}}\otimes\rho_{R_1^{(2)}R_2^{(1)}}\otimes\dots\otimes\rho_{R_{m-1}^{(2)}R_m^{(1)}}\right), \rho_{R_m^{(2)}T} \right] \\
=& \mathrm{Seri}_{\Phi,E} \left[ \Phi^{\pi_{m-1}'} \left( \rho_{SR_1^{(1)}}\otimes\rho_{R_1^{(2)}R_2^{(1)}}\otimes\dots\otimes\rho_{R_{m-1}^{(2)}R_m^{(1)}}\right), \rho_{R_m^{(2)}T} \right] \\
=& \mathrm{Seri}_{\Phi,E}^{\pi_m'}\left[\rho_{SR_1^{(1)}}, \rho_{R_1^{(2)}R_2^{(1)}}, \dots, \rho_{R_m^{(2)}T}\right],
\end{align*}
where $\pi_m'=(1,2,\dots,m)$.

Case 2: $\pi_m(m) = n_0 \neq m$.
There exist a permutation $\pi_{n_0-1}$ of $(1,\dots,n_0-1)$ and a permutation $\pi_{m-n_0}$ of $\{n_0+1,\dots,m\}$ such that:
\begin{align*}
\mathrm{Seri}_{\Phi,E}^{\pi_m}\Big[
&\rho_{SR_1^{(1)}},  \rho_{R_1^{(2)}R_2^{(1)}},\dots, \rho_{R_m^{(2)}T}\Big]\\
= {\rm Seri}_{\Phi,E}\Biggl[
&\Phi^{\pi_{n_0-1}}\left(\rho_{SR_1^{(1)}}\otimes...\otimes\rho_{R_{n_0-1}^{(2)}R_{n_0}^{(1)}}\right), \Phi^{\pi_{m-n_0}}\left( \rho_{R_{n_0}^{(2)}R_{n_0+1}^{(1)}} \otimes...\otimes \rho_{R_m^{(2)}T}\right)\Biggr]
.
\end{align*}

By the induction hypothesis:
\begin{align}\label{eq-proof1}
{\rm Seri}_{\Phi,E}\Big[
&\Phi^{\pi_{n_0-1}}\left(\rho_{SR_1^{(1)}}\otimes...\otimes\rho_{R_{n_0-1}^{(2)}R_{n_0}^{(1)}}\right),\nonumber\\
&\Phi^{\pi_{m-n_0}}\left( \rho_{R_{n_0}^{(2)}R_{n_0+1}^{(1)}} \otimes...\otimes \rho_{R_m^{(2)}T}\right)\Big]\nonumber\\
= \mathrm{Seri}_{\Phi,E} \Big[
&\Phi^{(1,\dots,n_0-1)} \left(\rho_{SR_1^{(1)}}\otimes...\otimes\rho_{R_{n_0-1}^{(2)}R_{n_0}^{(1)}}\right),\nonumber\\ &\Phi^{(m,m-1,\dots,n_0+1)} \left( \rho_{R_{n_0}^{(2)}R_{n_0+1}^{(1)}} \otimes...\otimes \rho_{R_m^{(2)}T}\right)\Big]
\end{align}

Using the $E$-based order-independent property of $\Phi$, the right function in Eq.~\eqref{eq-proof1} equals to
\begin{align*}
{\rm Seri}_{\Phi,E}\Big[
&\Phi^{(1,2,...,n_0-1)}\left(\rho_{SR_1^{(1)}}\otimes...\otimes\rho_{R_{n_0-1}^{(2)}R_{n_0}^{(1)}}\right), \\
&\Phi^{(m,m-1,...,n_0+1)}\left( \rho_{R_{n_0}^{(2)}R_{n_0+1}^{(1)}} \otimes...\otimes \rho_{R_m^{(2)}T}\right)\Big]\\
={\rm Seri}_{\Phi,E}\Big[
&\Phi^{(1,2,...,n_0)}\left(\rho_{SR_1^{(1)}}\otimes...\otimes\rho_{R_{n_0}^{(2)}R_{n_0+1}^{(1)}}\right),\\
&\Phi^{(m,m-1,...,n_0+2)}\left( \rho_{R_{n_0+1}^{(2)}R_{n_0+2}^{(1)}} \otimes...\otimes \rho_{R_m^{(2)}T}\right)\Big]\\
\vdots \\
={\rm Seri}_{\Phi,E}\Big[&\Phi^{(1,2,...,m-1)}\left(\rho_{SR_1^{(1)}}\otimes...\otimes\rho_{R_{m-1}^{(2)}R_{m}^{(1)}}\right),  \rho_{R_m^{(2)}T}\Big]\\
={\rm Seri}_{\Phi,E}^{\pi_m'}\Big[&\rho_{SR_1^{(1)}},\rho_{R_{1}^{(2)}R_{2}^{(1)}},...,  \rho_{R_m^{(2)}T}\Big].
\end{align*}

\subsection{Impact of order-dependent swapping and concentration}\label{SI-simplification_sequence}
\setcounter{equation}{0}

Consider $\Phi$ or $\Psi$ does not $E$-based order-independent on $\mathcal{D}_{\mathcal{S}}$.
Now we provide two examples to demonstrate that the execution order of (${\rm {I}^2}$) and (${\rm {II}^2}$) may impact the final outcome.

\begin{figure}
    \centering
    \subfigure[]{
    \includegraphics[width=200pt]{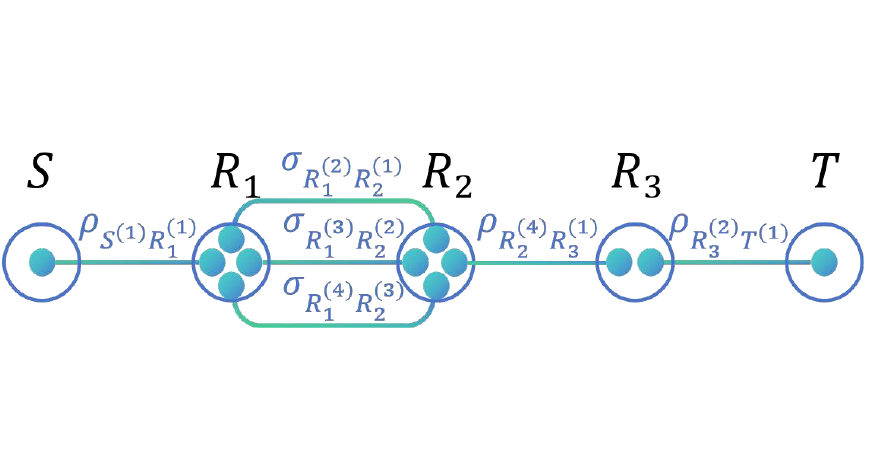}
    \label{fig-non_commutative_swapping}
    }
    \subfigure[]{
    \includegraphics[width=200pt]{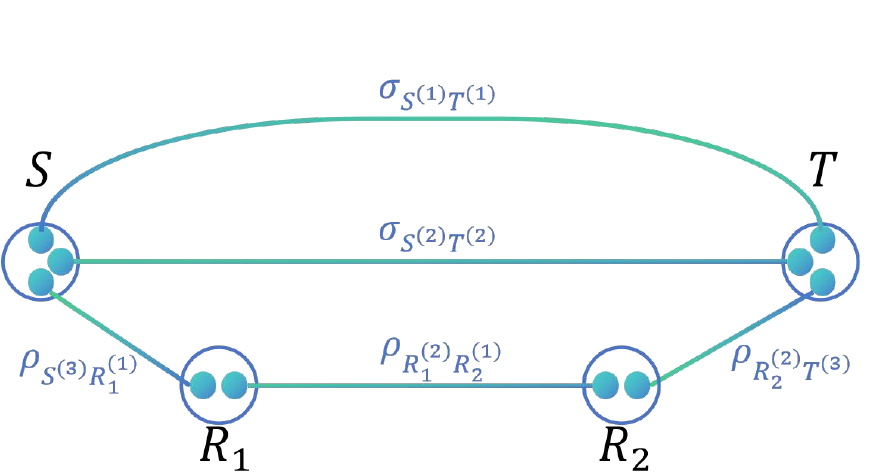}
    \label{fig-non_commutative_concentration1}
    }
    \caption{Two network topologies of series-parallel QNs.}
\end{figure}

\subsubsection{Order-dependent entanglement swapping}

Consider a ``parallel-then-series'' network topology [Fig.~\ref{fig-p_t_s}] shown in Fig.~\ref{fig-non_commutative_swapping}. Suppose nodes $S$ and $R_1$ share a state $\rho_{S^{(1)}R_1^{(1)}}$; $R_1$ and $R_2$ share three states $\sigma_{R_1^{(2)}R_2^{(1)}}$, $\sigma_{R_1^{(3)}R_2^{(2)}}$, and $\sigma_{R_1^{(4)}R_2^{(3)}}$; $R_2$ and $R_3$ share a state $\rho_{R_2^{(4)}R_3^{(1)}}$; and $R_3$ and $T$ share a state $\rho_{R_3^{(2)}T^{(1)}}$.
If $\Phi$ is not $E$-based order-dependent for $\mathcal{D}_{\mathcal{S}}$, we show in the following two cases how different simplification orders can affect the resulting entangled state between terminals.

(1) {Sequence [(${\rm {I}^2}$),(${\rm {II}^2}$)].} First performing (${\rm {I}^2}$) concentrates the three states between $R_1$ and $R_2$ into a single state
\begin{eqnarray}\label{eq-rho2}
    \rho_{R_1^{(2,3,4)}R_2^{(1,2,3)}} =\Psi^{\rm{opt}}_{R_1R_2} \left(\sigma_{R_1^{(2)}R_2^{(1)}}\otimes\sigma_{R_1^{(3)}R_2^{(2)}}\otimes\sigma_{R_1^{(4)}R_2^{(3)}}\right)
\end{eqnarray}
using the optimal concentration order of $\sigma_{R_1^{(2)}R_2^{(1)}}$, $\sigma_{R_1^{(3)}R_2^{(2)}}$, and $\sigma_{R_1^{(4)}R_2^{(3)}}$.
Next, executing (${\rm {II}^2}$) applies optimal entanglement swapping with permutation $\pi_3^{a}$ of $(1,2,3)$ to the four states $\rho_{S^{(1)}R_1^{(1)}}$, $\rho_{R_1^{(2,3,4)}R_2^{(1,2,3)}}$, $\rho_{R_2^{(4)}R_3^{(1)}}$, and $\rho_{R_3^{(2)}T^{(1)}}$.
Then one obtains the final state
\begin{eqnarray}\label{eq-rho_a}
    \rho_{S^{(1)}T^{(1)}}^{a}
    =\Phi^{\pi_3^{a}} \left(\rho_{S^{(1)}R_1^{(1)}}\otimes\rho_{R_1^{(2,3,4)}R_2^{(1,2,3)}}\otimes\rho_{R_2^{(4)}R_3^{(1)}}\otimes\rho_{R_3^{(2)}T^{(1)}}\right)
\end{eqnarray}
established between $S$ and $T$.

(2) {Sequence [(${\rm {II}^2}$),(${\rm {I}^2}$),(${\rm {II}^2}$)].}
In contrast, if (${\rm {II}^2}$) is performed first, the states $\rho_{R_2^{(4)}R_3^{(1)}}$ and $\rho_{R_3^{(2)}T^{(1)}}$ are first converted into a state
\begin{eqnarray}\label{eq-rho34}
    \rho_{R_2^{(4)}T^{(1)}}
    =\Phi_{R_3}^{R_2T} \left(\rho_{R_2^{(4)}R_3^{(1)}}\otimes\rho_{R_3^{(2)}T^{(1)}}\right)
\end{eqnarray}
between $R_3$ and $T$ via the entanglement swapping.
Then, executing (${\rm {I}^2}$) produces the state $\rho_{R_1^{(2,3,4)}R_2^{(1,2,3)}}$ between $R_1$ and $R_2$ as in Eq.~\eqref{eq-rho2} by concentrating $\sigma_{R_1^{(2)}R_2^{(1)}}$, $\sigma_{R_1^{(3)}R_2^{(2)}}$, and $\sigma_{R_1^{(4)}R_2^{(3)}}$.
Finally, performing (${\rm {II}^2}$) creates a final state
\begin{eqnarray}\label{eq-rho_a2}
    \rho_{S^{(1)}T^{(1)}}^{a'}
    &=&\Phi^{\pi_2}\left(\rho_{S^{(1)}R_1^{(1)}}\otimes\rho_{R_1^{(2,3,4)}R_2^{(1,2,3)}}\otimes\rho_{R_2^{(4)}T^{(1)}}\right)\nonumber\\
    &=&\Phi^{\pi_2}\Big[\rho_{S^{(1)}R_1^{(1)}}\otimes\Psi^{\rm{opt}}_{R_1R_2} \left(\sigma_{R_1^{(2)}R_2^{(1)}}\otimes\sigma_{R_1^{(3)}R_2^{(2)}}\otimes\sigma_{R_1^{(4)}R_2^{(3)}}\right)\nonumber\\
    &&\ \ \ \ \ \ \otimes\Phi_{R_3}^{R_2R_3} \left(\rho_{R_2^{(4)}R_3^{(1)}}
    \otimes\rho_{R_3^{(2)}T^{(1)}}\right)\Big]\nonumber\\
    &=&\Phi^{(3,\pi_2)}\Big[\rho_1\otimes\Psi^{\rm{opt}}_{R_1R_2} \left(\sigma_{R_1^{(2)}R_2^{(1)}}\otimes\sigma_{R_1^{(3)}R_2^{(2)}}\otimes\sigma_{R_1^{(4)}R_2^{(3)}}\right)\otimes\rho_{R_2^{(4)}R_3^{(1)}}\nonumber\\
    &&\ \ \ \ \ \ \ \ \ \ \otimes\rho_{R_3^{(2)}T^{(1)}}\Big]
\end{eqnarray}
between $S$ and $T$ via the optimal entanglement swapping of $\rho_{S^{(1)}R_1^{(1)}}$, $\rho_{R_1^{(2,3,4)}R_2^{(1,2,3)}}$, and $\rho_{R_2^{(4)}T^{(1)}}$ with permutation $\pi_2$ of $(1,2)$.

Because $\Phi$ is not $E$-based order-independent on state set $\mathcal{D}_{\mathcal
S}$, we cannot obtain that $E\left[\rho_{S^{(1)}T^{(1)}}^{a}\right]=E\left[\rho_{S^{(1)}T^{(1)}}^{a'}\right]$ when $\pi_3^{a}\neq(3,\pi_2)$.
However, if $\Phi$ is $E$-based order-independent on $\mathcal{D}_{\mathcal{S}}$, then the order of (${\rm {I}^2}$) and (${\rm {II}^2}$)--—regardless of $\Psi$’s order independence—--does not affect the final entanglement (i.e., $E\left[\rho_{S^{(1)}T^{(1)}}^{a}\right]=E\left[\rho_{S^{(1)}T^{(1)}}^{a'}\right]$).

\begin{figure}
    \centering
    \includegraphics[width=250pt]{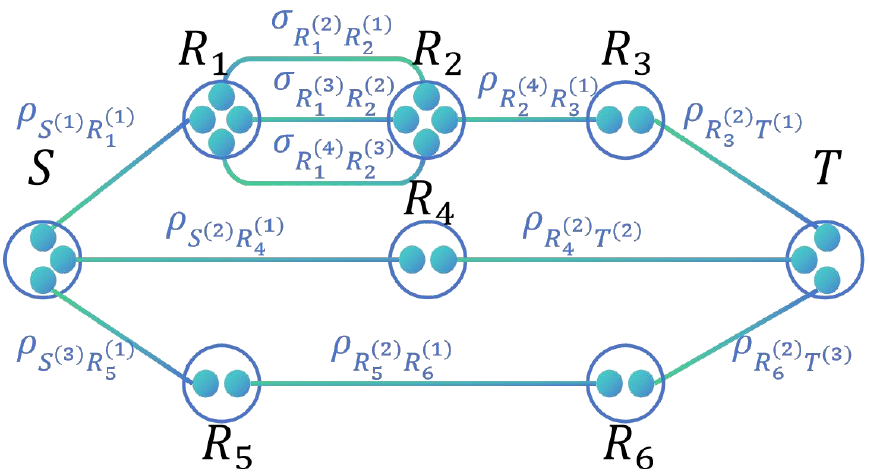}
    \caption{A series-parallel QN.}
    \label{fig-non_commutative_concentration}
\end{figure}

\subsubsection{Order-dependent entanglement concentration}

Let the entanglement concentration $\Psi$ fail to be $E$-based order-independent on $\mathcal{D}_{\mathcal{S}}$.
Consider a series-parallel QN with a ``parallel-then-series'' network topology [Fig.~\ref{fig-p_t_s}] depicted in Fig.~\ref{fig-non_commutative_concentration1}.
In this QN, nodes $S$ and $T$ share two states---denoted by $\sigma_{S^{(1)}T^{(1)}}$ and $\sigma_{S^{(2)}T^{(2)}}$---which together form a maximal parallel sub-network.
Additionally, there is a maximal series sub-network $(S,R_1,R_2,T)$: $S$ and $R_1$ share the state $\rho_{S^{(3)}R_1^{(1)}}$, $R_1$ and $R_2$ share $\rho_{R_1^{(2)}R_2^{(1)}}$, and $R_2$ and $T$ share $\rho_{R_2^{(2)}T^{(3)}}$.
Assume that all states in this QN lie within $\mathcal{D}_{\mathcal{S}}$.
We now compare the resulting state under two different simplification orders:

(1) {Sequence [(${\rm {I}^2}$), (${\rm {II}^2}$), (${\rm {I}^2}$)].}
First, by performing the parallel simplification (${\rm {I}^2}$), the two states $\sigma_{S^{(1)}T^{(1)}}$ and $\sigma_{S^{(2)}T^{(2)}}$ are mapped, via the entanglement concentration $\Psi_{ST}$, into a single state
\begin{eqnarray*}
    \sigma_{S^{(1,2)}T^{(1,2)}} =\Psi_{ST}\left[\sigma_{S^{(1)}T^{(1)}}\otimes\sigma_{S^{(2)}T^{(2)}}\right]
\end{eqnarray*}
shared between $S$ and $T$.
Next, the series simplification (${\rm {II}^2}$) is performed on the maximal series sub-network $(S,R_1,R_2,T)$, yielding a new state
\begin{eqnarray}\label{eq-b1_sigma3}
    \sigma_{S^{(3)}T^{(3)}}=\Phi^{\rm{opt}}\left[\rho_{S^{(3)}R_1^{(1)}}\otimes\rho_{R_1^{(2)}R_2^{(1)}}\otimes\rho_{R_2^{(2)}T^{(3)}}\right]
\end{eqnarray}
between $S$ and $T$ where $\Phi^{\pi_2}$ denotes the entanglement swapping applied on three states $\rho_{S^{(3)}R_1^{(1)}}$, $\rho_{R_1^{(2)}R_2^{(1)}}$, and $\rho_{R_2^{(2)}T^{(3)}}$ in the optimal order.
Now the QN is a parallel one with two states $\sigma_{S^{(1,2)}T^{(1,2)}}$ and $\sigma_{S^{(3)}T^{(3)}}$.
Finally, by (${\rm {I}^2}$) again, the two states are used to produced the final resource
\begin{eqnarray*}
    \rho_{ST}^b
    &=&\Psi_{ST}\left[\sigma_{S^{(1,2)}T^{(1,2)}}\otimes\sigma_{S^{(3)}T^{(3)}}\right]\nonumber\\
    &=&\Psi_{ST}\left[\Psi_{ST}\left[\sigma_{S^{(1)}T^{(1)}}\otimes\sigma_{S^{(2)}T^{(2)}}\right]\otimes\sigma_{S^{(3)}T^{(3)}}\right]\nonumber\\
    &=&\Psi_{ST}^{(1,2,3)}\left[\sigma_{S^{(1)}T^{(1)}}\otimes\sigma_{S^{(2)}T^{(2)}}\otimes\sigma_{S^{(3)}T^{(3)}}\right].
\end{eqnarray*}

(2) {Sequence [(${\rm {II}^2}$), (${\rm {I}^2}$)].}
Now performing (${\rm {II}^2}$) first, we then obtain $\sigma_{S^{(3)}T^{(3)}}$ [Eq.~\eqref{eq-b1_sigma3}] to replace the three states in the maximal series sub-network $(S,R_1,R_2,T)$.
The QN becomes a parallel QN with three states $\sigma_{S^{(1)}T^{(1)}}$, $\sigma_{S^{(2)}T^{(2)}}$, and $\sigma_{S^{(3)}T^{(3)}}$ which are converted, via the entanglement concentration $\Psi^{\pi_3^b}$ of these states with the optimal order denoted by $\pi_3^b$, into the final state
\begin{eqnarray*}
    \rho_{ST}^{b'}
    &=&\Psi^{\pi_3^b}\left[\sigma_{S^{(1)}T^{(1)}}\otimes\sigma_{S^{(2)}T^{(2)}}\otimes\sigma_{S^{(3)}T^{(3)}}\right].
\end{eqnarray*}

Since $\Psi$ is not $E$-based order-independent on $\mathcal{D}_{\mathcal{S}}$, the order $\pi_3^b$ need not coincide with $(1,2,3)$, and thus in general $E\left[\rho_{ST}^b\right]\neq E\left[\rho_{ST}^{b'}\right]$.
By contrast, if $\Psi$ is $E$-based order-independent for $\mathcal{D}_{\mathcal{S}}$, then regardless of whether $\Phi$ is $E$-based order-independent, one immediately obtains $\rho_{ST}^b=\rho_{ST}^{b'}$.

For example, consider the CV-based QN of TMSVSs with the network topology as Fig.~\ref{fig-non_commutative_concentration1} where the values of the ratio negativity are $\chi_{\mathcal{N}}\left[\sigma_{S^{(1)}T^{(1)}}\right]=0.5$, $\chi_{\mathcal{N}}\left[\sigma_{S^{(2)}T^{(2)}}\right]=0.5$, $\chi_{\mathcal{N}}\left[\rho_{S^{(3)}R_1^{(1)}}\right]=0.9$, $\chi_{\mathcal{N}}\left[\rho_{R_1^{(2)}R_2^{(1)}}\right]=0.9$, and $\chi_{\mathcal{N}}\left[\rho_{R_2^{(2)}T^{(3)}}\right]=0.9$.
Then the first operation sequence [(${\rm {I}^2}$), (${\rm {II}^2}$), (${\rm {I}^2}$)] yields the final ratio negativity is $\chi_{1}\approx0.788$. In contrast, the second operation sequence gives different ratio negativity $\chi_{2}\approx 0.818\neq\chi_1$.

\textbf{Remark:}
From the perspective of QN topology design, these two examples show that:

(1) If the entanglement swapping map $\Phi$ is $E$-based order-independent while the entanglement-concentrating map $\Psi$ is not, one should employ the ``parallel-then-series'' network topology such as the architecture in Fig.~\ref{fig-non_commutative_swapping} rather than the ``series-then-parallel'' network topology in Fig.~\ref{fig-non_commutative_concentration1}.

(2) Conversely, if $\Psi$ is $E$-based order-independent but $\Phi$ is not, the topology in Fig.~\ref{fig-non_commutative_concentration1} becomes preferable to that in Fig.~\ref{fig-non_commutative_swapping}.

\subsubsection{Order-dependent entanglement operation}

In fact, for a general series-parallel QN, if either $\Phi$ or $\Psi$ is not $E$-based order-independent, the choice of whether to apply simplification operation (${\rm {I}^2}$) or (${\rm {II}^2}$) first affects the final output state. We now present an example to illustrate this.

Consider the QN with topology shown in Fig.~\ref{fig-non_commutative_concentration}, which contains three parallel paths between nodes $S$ and $T$.
The first path $(S,R_1,R_2,R_3,T)$ is identical to the model in Fig.~\ref{fig-non_commutative_swapping}.
The second path $(S,R_4,T)$ is a maximal series sub-network with source states $\rho_{S^{(2)}R_4^{(1)}}$ and $\rho_{R_4^{(2)}T^{(2)}}$.
The third path $(S,R_5,R_6,T)$ is another maximal series sub-network with source states $\rho_{S^{(3)}R_5^{(1)}}$, $\rho_{R_5^{(2)}R_6^{(1)}}$, and $\rho_{R_6^{(2)}T^{(3)}}$.
Let either $\Phi$ or $\Psi$ be not $E$-based order-independent.
We now compare the resulting state under two different simplification orders:

(1) {Sequence [(${\rm {I}^2}$),(${\rm {II}^2}$),(${\rm {I}^2}$)].} First perform (${\rm {I}^2}$), which concentrates the three states $\sigma_1$, $\sigma_2$, and $\sigma_3$ between $R_1$ and $R_2$ into
$\rho_{R_1^{(2,3,4)}R_2^{(1,2,3)}}$ [Eq.~\eqref{eq-rho2}].
Next execute (${\rm {II}^2}$).
On the first path, convert the four states $\rho_{S^{(1)}R_1^{(1)}}$, $\rho_{R_1^{(2,3,4)}R_2^{(1,2,3)}}$, $\rho_{R_2^{(4)}R_3^{(1)}}$, and $\rho_{R_3^{(2)}T^{(1)}}$
 into
 \begin{eqnarray*}
     \sigma_{S^{(1)}T^{(1)}}'=\rho_{S^{(1)}T^{(1)}}^a
 \end{eqnarray*}
via the entanglement swapping $\Phi^{\pi_3^a}$ in Eq.~\eqref{eq-rho_a}.
On the second path, convert $\rho_{S^{(2)}R_4^{(1)}}$ and $\rho_{R_4^{(2)}T^{(2)}}$ into
\begin{eqnarray}\label{eq-sigma2}
    \sigma_{S^{(2)}T^{(2)}}'=\Phi_{R_4^{(1,2)}}^{S^{(2)}T^{(2)}}\left(\rho_{S^{(2)}R_4^{(1)}}\otimes\rho_{R_4^{(2)}T^{(2)}}\right)
\end{eqnarray}
via entanglement swapping $\Phi_{R_4^{(1,2)}}^{S^{(2)}T^{(2)}}$.
On the third path, convert $\rho_{S^{(3)}R_5^{(1)}}$, $\rho_{R_5^{(2)}R_6^{(1)}}$, and $\rho_{R_6^{(2)}T^{(3)}}$ into
\begin{eqnarray}\label{eq-sigma3}
   \sigma_{S^{(3)}T^{(3)}}'&=&\Phi^{\rm{opt}}\left(\rho_{S^{(3)}R_5^{(1)}}\otimes\rho_{R_5^{(2)}R_6^{(1)}}\otimes\rho_{R_6^{(2)}T^{(3)}}\right)
\end{eqnarray}
via the optimal entanglement swapping $\Phi^{\rm{opt}}$ of the three states.
Finally, perform entanglement concentration across the three parallel outputs $\sigma_{S^{(1)}T^{(1)}}'$, $\sigma_{S^{(2)}T^{(2)}}'$, and $\sigma_{S^{(3)}T^{(3)}}'$ with optimal concentration ordering $\pi_3'$ of these three states, the final state established between $S$ and $T$ is
\begin{eqnarray*}
    \sigma_{ST}'
    &=&\Psi^{\pi_3'}_{ST} \left[\sigma_{S^{(1)}T^{(1)}}'\otimes\sigma_{S^{(2)}T^{(2)}}'\otimes\sigma_{S^{(3)}T^{(3)}}'\right]\nonumber\\
    &=&\Psi^{\pi_3'}_{ST} \left[\rho_{S^{(1)}T^{(1)}}^a\otimes\sigma_{S^{(2)}T^{(2)}}'\otimes\sigma_{S^{(3)}T^{(3)}}'\right].
\end{eqnarray*}

(1) {Sequence [(${\rm {II}^2}$),(${\rm {I}^2}$),(${\rm {II}^2}$),(${\rm {I}^2}$)].}
Alternatively, start by performing  (${\rm {II}^2}$).
Specifically, map $\rho_{R_2^{(4)}R_3^{(1)}}$ and $\rho_{R_3^{(2)}T^{(1)}}$ are first converted into a state into $\rho_{R_2^{(4)}T^{(1)}}$ [Eq.~\eqref{eq-rho34}]
between $R_2$ and $T$ via entanglement swapping $\Phi_{R_3^{(1,2)}}^{R_2^{(4)}T^{(1)}}$.
Simultaneously, on the second path, convert $\rho_{S^{(2)}R_4^{(1)}}$ and $\rho_{R_4^{(2)}T^{(2)}}$ into $\sigma_{S^{(2)}T^{(2)}}'$ [Eq.~\eqref{eq-sigma2}] between $S$ and $T$ via the entanglement swapping $\Phi_{R_4^{(1,2)}}^{S^{(2)}T^{(2)}}$, and on the third path map $\rho_{S^{(3)}R_5^{(1)}}$, $\rho_{R_5^{(2)}R_6^{(1)}}$, and $\rho_{R_6^{(2)}T^{(3)}}$ into $\sigma_{S^{(3)}T^{(3)}}'$ between $S$ and $T$ via $\Phi^{\rm{opt}}$ as shown in Eq.~\eqref{eq-sigma3}.
Next, execute (${\rm {I}^2}$).
This yields a state $\rho_{R_1^{(2,3,4)}R_2^{(1,2,3)}}$ [Eq.~\eqref{eq-rho2}] between $R_1$ and $R_2$ by concentrating $\sigma_{R_1^{(2)}R_2^{(1)}}$, $\sigma_{R_1^{(3)}R_2^{(2)}}$, and $\sigma_{R_1^{(4)}R_2^{(3)}}$, and creates another state
\begin{eqnarray*}
    \sigma_{S^{(2,3)}T^{(2,3)}}=\Psi_{ST}(\sigma_{S^{(2)}T^{(2)}}'\otimes\sigma_{S^{(3)}T^{(3)}}')
\end{eqnarray*}
between $S$ and $T$ by concentrating $\sigma_{S^{(2)}T^{(2)}}'$ and $\sigma_{S^{(3)}T^{(3)}}'$.
Then, perform (${\rm {II}^2}$) to convert the three states $\rho_{S^{(1)}R_1^{(1)}}$, $\rho_{R_1^{(2,3,4)}R_2^{(1,2,3)}}$, and $\rho_{R_2^{(4)}T^{(1)}}$
in the first path to a single state
\begin{eqnarray*}
    \sigma_{S^{(1)}T^{(1)}}''=\rho_{S^{(1)}T^{(1)}}^{a'}
\end{eqnarray*}
between $S$ and $T$, where $\rho_{S^{(1)}T^{(1)}}^{a'}$ is defined in Eq.~\eqref{eq-rho_a2}.
Finally, apply (${\rm {I}^2}$) to
concentrate two outputs $\sigma_{S^{(1)}T^{(1)}}''$ and $\sigma_{S^{(2,3)}T^{(2,3)}}$.
The resulting state is
\begin{eqnarray*}
    \sigma_{ST}'' &=&\Psi_{ST}\left[\sigma_{S^{(1)}T^{(1)}}''\otimes\sigma_{S^{(2,3)}T^{(2,3)}}\right]\nonumber\\ &=&\Psi_{ST}^{(2,3,1)}\left[\sigma_{S^{(1)}T^{(1)}}''\otimes\sigma_{S^{(2)}T^{(2)}}'\otimes\sigma_{S^{(3)}T^{(3)}}'\right]\nonumber\\
    &=&\Psi_{ST}^{(2,3,1)}\left[\rho_{S^{(1)}T^{(1)}}^{a'}\otimes\sigma_{S^{(2)}T^{(2)}}'\otimes\sigma_{S^{(3)}T^{(3)}}'\right]
\end{eqnarray*}
where $(2,3,1)$ indicates that the three states $\sigma_{S^{(1)}T^{(1)}}''$, $\sigma_{S^{(2)}T^{(2)}}'$, and $\sigma_{S^{(3)}T^{(3)}}'$ are concentrated in order $\sigma_{S^{(2)}T^{(2)}}'$, $\sigma_{S^{(3)}T^{(3)}}'$, $\sigma_{S^{(1)}T^{(1)}}''$.
From the first example, if $\Phi$ being not $E$-based order-independent on $\mathcal{D}_{\mathcal
S}$,then one cannot conclude $E\left[\rho_{S^{(1)}T^{(1)}}^{a}\right]=E\left[\rho_{S^{(1)}T^{(1)}}^{a'}\right]$, and hence $E\left[\sigma_{ST}'\right]=E\left[\sigma_{ST}''\right]$ does not necessarily hold.
Now suppose instead that $\Phi$ is $E$-based order-independent while $\Psi$ is not for $\mathcal{D}_{\mathcal{S}}$.
Although this guarantees $E\left[\rho_{S^{(1)}T^{(1)}}^{a}\right]=E\left[\rho_{S^{(1)}T^{(1)}}^{a'}\right]$, we still cannot ensure $E\left[\sigma_{ST}'\right]=E\left[\sigma_{ST}''\right]$ unless two orders are same, $\pi_3'=(2,3,1)$.
In summary, for a general series-parallel network topology as shown in Fig.~\ref{fig-non_commutative_concentration}, if either $\Phi$ or $\Psi$ fails to be $E$-based order-independent for $\mathcal{D}_{\mathcal{S}}$, then the order in which operations (${\rm {I}^2}$) and (${\rm {II}^2}$) are applied first can affect the final output state.

When both maps $\Phi$ and $\Psi$ are not generalized $E$-based order-independent, a suitable QN topology—based on the two simplification operations—should satisfy one of the following criteria:

\textbf{Remark:} When entanglement swapping or entanglement concentration is not $E$-based order-independent for the state set $\mathcal{D}_{\mathcal{S}}$, constructing series–parallel QNs with states in $\mathcal{D}_{\mathcal{S}}$ requires not only selecting the operation order within each maximal serial or parallel sub-network, but also determining the sequence of the simplification operations (${\rm {I}^2}$) and (${\rm {II}^2}$) between sub-networks, which adds an extra layer of complexity to QN design.

Moreover, these findings suggest that, in the design of QNs, the order in which simplification operations (${\rm {I}^2}$) and (${\rm {II}^2}$) are applied can be disregarded. Specifically, assuming all source states in a QN belong to the set $\mathcal{D}_{\mathcal{S}}$, constructing a QN with topology as depicted in Fig.~\ref{fig-non_commutative_swapping} is not appropriate when entanglement swapping is not $E$-based order-independent, while establishing a QN with topology as shown in Fig.~\ref{fig-non_commutative_concentration} is not suitable when either the entanglement swapping or the entanglement concentration is not $E$-based order-independent.

\section*{Competing interests}
The Authors declare no Competing Financial or Non-Financial Interests.


\end{document}